\documentclass[aps,prb,twocolumn,superscriptaddress]{revtex4-2}

\usepackage{amsmath,amssymb,wasysym,graphicx}

\usepackage[utf8]{inputenc}
\usepackage[T1]{fontenc}
\usepackage{xcolor}
\usepackage{float}
\usepackage{pgf}
\usepackage{subfigure}
\usepackage{longtable}
\usepackage{amsmath}
\usepackage{bbm}
\usepackage{multirow}
\usepackage{framed}
\usepackage{relsize}

\usepackage{amsthm}
\newtheorem{theorem}{Theorem}
\newtheorem{lemma}{Lemma}
\newtheorem{corollary}{Corollary}

\IfFileExists{newtxtext.sty}
{\usepackage{newtxtext,newtxmath}}
{\IfFileExists{stix.sty}
  {\usepackage{stix}}
  {\IfFileExists{mathptmx.sty}
    {\usepackage{mathptmx}}{} } }

\usepackage{textcomp}

\usepackage{bm}

\IfFileExists{siunitx.sty}{\usepackage{booktabs,siunitx}}{}

\usepackage{color}
\definecolor{LinkColor}{rgb}{0.256,0.439,0.588}
\usepackage{hyperref}
\hypersetup{
  pdfauthor={good guys},
  pdftitle={good title},
  colorlinks=true,
  citecolor=LinkColor,
  linkcolor=LinkColor,
  urlcolor=LinkColor
}

\usepackage{bbm}
\usepackage{pifont}
\usepackage[normalem]{ulem}

\newcommand{\RNum}[1]{\uppercase\expandafter{\romannumeral #1\relax}}

\newcommand{\Rom}[1]{ \uppercase\expandafter{\romannumeral#1}}

\newcommand{\kk}{{\bm{\mathrm{k}}}}
\newcommand{\rr}{{\bm{r}}}

\newcommand{\bb}{{\bm{b}}}
\newcommand{\ssigma}{\bm{\sigma}}
\newcommand{\ggamma}{\bm{\gamma}}
\newcommand{\ttau}{\bm{\tau}}

\newcommand{\nn}{\hat{\bm{n}}}

\newcommand{\cob}{}
\newcommand{\cog}{\color{darkblue}}
\newcommand{\zz}{\mathbbm{Z}}
\newcommand{\ee}{\bm{e}}

\newcommand{\dd}{{\rm{d}}}
\definecolor{darkblue}{rgb}{0.15,0.25,0.6}
\definecolor{ZYcolor}{rgb}{0.1,0.5,0.4}
\definecolor{darkred}{rgb}{0.65,0.2,0.15}

\allowdisplaybreaks

\begin{document}
\title{Topological classification of intrinsic 3D superconductors using anomalous surface construction}
\author{Zhongyi Zhang}
\thanks{These author contributed equally to this study.}
\affiliation{Beijing National Laboratory for Condensed Matter Physics and Institute of Physics, Chinese Academy of Sciences, Beijing 100190, China}
\affiliation{University of Chinese Academy of Sciences, Beijing 100049, China}
\author{Jie Ren}
\thanks{These author contributed equally to this study.}
\affiliation{Beijing National Laboratory for Condensed Matter Physics and Institute of Physics, Chinese Academy of Sciences, Beijing 100190, China}
\affiliation{University of Chinese Academy of Sciences, Beijing 100049, China}
\author{Yang Qi}
\affiliation{Center for Field Theory and Particle Physics, Department of Physics, Fudan University, Shanghai 200433, China}
\affiliation{State Key Laboratory of Surface Physics, Fudan University, Shanghai 200433, China}
\author{Chen Fang}
\email{cfang@iphy.ac.cn}
\affiliation{Beijing National Laboratory for Condensed Matter Physics and Institute of Physics, Chinese Academy of Sciences, Beijing 100190, China}
\affiliation{Songshan Lake Materials Laboratory, Dongguan, Guangdong 523808, China}
\affiliation{Kavli Institute for Theoretical Sciences, Chinese Academy of Sciences, Beijing 100190, China}

\begin{abstract}
Intrinsic topological superconductors have protected gapless Majorana modes, bound and/or propagating, at the natural boundaries of the sample, without requiring field, defect, or heterostructure.
We establish the complete classification/construction of intrinsic topological superconductors jointly protected by point-group and time-reversal symmetries in three dimensions.
This is obtained from enumerating distinct ways for stacking \textit{$n$th-order irreducible building blocks}, minimal anomalous surface states of $n$th-order topological superconductors.
Particularly, our method provides a unified description of possible surface anomalies
away from high-symmetry points/lines in terms of the homotopy group of the surface mass field.
\end{abstract}
\maketitle

\section*{Introduction}
Topological superconductors (TSC) exhibit topological edge modes \cite{qi2011topological,RevModPhys.88.035005,RevModPhys.88.021004,RevModPhys.82.3045,doi:10.1146/annurev-conmatphys-031214-014501,Sato_2017,alicea2012new,PhysRevB.88.155420,PhysRevLett.115.127003,PhysRevLett.117.047001}, also called the Majorana modes \cite{PhysRevB.81.045120,ryu2010topological,Kitaev_2001,PhysRevB.61.10267,RevModPhys.75.657,PhysRevX.9.011033,PhysRevLett.100.096407,PhysRevLett.105.077001,PhysRevLett.105.177002,PhysRevB.93.224505,PhysRevLett.107.097001,PhysRevLett.111.047006,PhysRevResearch.2.043300}, at the physical edges or topological defects.
These Majorana modes are believed to have non-Abelian braiding property, an attribute wanted for fault-tolerant quantum computation \cite{PhysRevB.100.054513,PhysRevB.101.085401,PhysRevResearch.2.032068,RevModPhys.80.1083,PhysRevLett.104.040502,PhysRevLett.105.046803,PhysRevX.5.041038,Lian10938}.
When the topology comes from the nontrivial structure of the pairing amplitude on the Fermi surface \cite{qi2010topological,qi2013axion,PhysRevB.81.220504,PhysRevLett.105.097001}, the TSC is intrinsic \cite{PhysRevB.98.165144,PhysRevLett.119.246401,PhysRevB.88.075142,wu2020pursuit,Kitaev_2001,RevModPhys.75.657,maeno2011evaluation,Kallin_2016,jiao2020chiral,PhysRevB.82.184516,li2020artificial,li2021higher,scammell2021intrinsic}, as opposed to extrinsic TSC where external factors, such as fields or defects (vortices), are necessary \cite{PhysRevB.82.115120,PhysRevLett.104.040502,PhysRevLett.105.077001,PhysRevLett.105.177002,PhysRevB.81.125318}.
Intrinsic TSCs have Majorana modes that appear at the natural physical boundaries of the sample and, as such, are more suitable for devices and applications.
However, their existence is believed to be even rarer than extrinsic TSC, and while there have been theoretical proposals, an experimental discovery is yet to be made \cite{10.1093/nsr/nwab087,hao2019topological,wang2020evidence,Vaitiekenaseaav3392,nadj2014observation,PhysRevLett.114.017001,PhysRevX.8.041056,kong2019half,machida2019zero,PhysRevLett.107.217001,kong2020tunable,das2012zero,chen2019observation,PhysRevLett.121.196801}.
(Below the adjective ``intrinsic'' is suppressed, but always implied, before ``TSC''.)

The scarcity of TSC is partly due to our comparatively less developed theory for their classification.
Let us use the example of the progress of topological normal (non-superconducting) states to illustrate this point.
These states were at first believed to be rare because, at that time, the definition was restricted to topological insulators protected by time-reversal symmetry \cite{RevModPhys.88.035005,RevModPhys.82.3045,RevModPhys.88.021004}.
However, the theories for topological semimetals \cite{PhysRevB.83.205101,PhysRevLett.108.046602,PhysRevLett.108.140405,PhysRevB.84.235126,PhysRevLett.108.266802,huang2015weyl,PhysRevB.85.195320,PhysRevB.92.081201} and topological crystalline insulators (second-order topological insulators) \cite{khalaf2018symmetry,PhysRevLett.119.246402,schindler2018higher,fang2015new,PhysRevLett.106.106802,PhysRevB.78.045426,hsieh2012topological,wang2016hourglass,kruthoff2017topological,po2017symmetry,van2018higher,PhysRevLett.111.056402} have greatly enhanced the family of topological normal states, such that many materials previously considered trivial were ``reinstated'' as topological \cite{zhang2019catalogue,vergniory2019complete,tang2019comprehensive}.
While TSC protected by time-reversal and by particle-hole symmetry have been classified in the "tenfold way" \cite{schnyder2008classification}, a general classification of TSC protected by crystalline symmetries is still incomplete \cite{PhysRevB.99.075105,song2019topological,geier2020symmetry,cornfeld2021tenfold,PhysRevB.102.180505,PhysRevB.90.165114,trifunovic2019higher}.
We expect that, similar to the scenario in topological normal states, such classification will motivate the discovery of new topology in superconductors so far considered trivial.

In this Letter, we establish a new scheme for constructing and classifying all topological superconductors jointly protected by time reversal and crystallographic point groups $G_{c}$ in three dimensions.
Our classification is based on the observation that any type of anomalous surface mode of TSC can be constructed from multiple copies of the Majorana cones, which are the minimal anomalous elements appearing in TSC with time-reversal symmetry (DIII class in Altland-Zirnbauer classification).
Specifically, Majorana cones are used to construct \textit{$n$th-order building blocks}, which can generate all $n$th-order TSC.
These surface Majorana cones furnish a projective co-representation of the symmetry group, the factor system of which depends on the pairing symmetry.
In general, the symmetric mass field can be added to gap out Majorana cones on the surface except for some gapless points and lines.
Some point-group symmetries can enforce the gaplessness of the mass field at high-symmetry points and lines, while other symmetries may result in gapless points and lines that can be gapped locally but cannot be removed globally.
The latter case can be detected by the zeroth/first homotopy classes of the mass field.
By the virtue of the bulk-edge correspondence, each distinct configuration of the mass field uniquely corresponds to a TSC in the bulk and thus constitutes an element in the classification group \cite{Fangeaat2374,khalaf2018symmetry,PhysRevB.92.081304,PhysRevB.97.115153,PhysRevB.96.205106,xiong2018minimalist,PhysRevX.8.011040,PhysRevX.7.011020}.
This surface-homotopy-class perspective, for the first time, gives a unified and rigorous understanding of the edge states in third-order TSC.

\section{Surface Majorana Cone}
The minimal anomalous surface state in class-DIII is a Majorana cone protected by time-reversal, particle-hole, and chiral symmetries (denoted as $T,P,S$), which has a $\mathbbm Z$ classification.
Since any surface termination of TSC can be regarded as a tangent plane of $S^2$, we can study the surface theory of TSC on the boundary of a fictitious sphere in the limit of infinite, where local Hilbert space of planar momentum and pseudospin degrees of freedom is defined on the tangent plane at every $\rr\in S^2$ \cite{SM}.
The effective theory of the Majorana cone can be expressed as $h_{z=\pm1}(\kk;\rr)=\pm(\kk\times\nn_\rr)\cdot\ssigma$, where $\kk$ is the in-plane momentum, $\nn_\rr$ is the unit normal vector of tangent plane, and $\ssigma$ are Pauli matrices in spin space.
The symmetries $T,P,S$ restricted on the surface are $T=i\sigma_yK,$ $ P_{\rr}^{{z=\pm 1}}=\pm\nn_\rr\cdot\ssigma\sigma_yK,$ $S_{\rr}^{{z=\pm 1}}=\pm\nn_\rr\cdot\ssigma$.
A point-group action $g$ connects the surface Hilbert space at $\rr$ with that at $g\rr$, imposing a $v_gh_\pm(\kk;\rr)v_g^\dagger=h_\pm(g \kk;g\rr)$ restriction to the local Hamiltonian, where $v(g)$ is the fundamental representation of $\mathrm{SU}(2)$ restricted to $G_c$.

\section{Symmetry Group and Pairing Symmetry}
The action of the the symmetry group $G=G_c\times Z_2^T\times Z_2^S$ on the surface Hilbert space forms a projective co-representation of $G$.
The commutation relation between two point-group symmetry representations is the same as that of the double-valued representation of $G_c$.
A natural gauge can be chosen such that the time-reversal symmetry commutes with all other symmetries.
The commutation relation between point-group symmetry $U_\rr(g)$ (the representation defined on the Hilbert space at $\rr$) and chiral symmetry $S_\rr$ depends on the pairing symmetry $\chi_g$, i.e., the sign change of the pairing potential $\Delta(\kk)$ under point-group symmetries \cite{SM}:
\begin{equation}\label{main text surface S and g}
U_\rr(g)S_\rr - \chi_g\det R_gS_{g\rr} U_\rr(g)=0,
\end{equation}
where $R_g$ is the fundamental representation matrix for $g$ of $O(3)$.
The term $\det R_g$ arises from the fact that chirality $S_\rr$ is a pseudoscalar, and $\chi_g=\pm1$ if $\Delta(\kk)$ is even/odd under $g$ \cite{footnote}.

\section{Nth-order Surface Mass Field}
The first-order surface state of TSC with winding number $z_w$ can be obtained by stacking $|z_w|$ Majorana cones.
However, the point-group symmetries may prohibit any superconductor with non-zero winding number due to the constraint $z_w=\chi_g\det{R_g}\  z_w$.
For example, if $\chi_{g^\prime}=-\det R_{g^\prime}$ for any $g^\prime\in G_c$, the winding number must be zero, indicating (first-order) topologically trivial surface states.

In order to construct the higher-order surface states of TSC, the first-order topology must be trivialized.
We achieve this by stacking multiple copies of Majorana cones with zero net winding number: $H_{2L}(\kk;\rr) = \oplus_Lh_+\oplus_Lh_-$.
There exists a canonical form of the symmetry representation:
\begin{equation}
\begin{aligned}
T=&\begin{pmatrix}\mathbbm{1}_L& \\ & \mathbbm{1}_L \\\end{pmatrix}\otimes i\sigma_y K,\\
S=&\begin{pmatrix}\mathbbm{1}_L& \\ & -\mathbbm{1}_L \\\end{pmatrix}\otimes \nn\cdot \ssigma,
U_\rr(g)=F(g)\otimes u_\rr(g),
\end{aligned}
\end{equation}
where $\mathbbm{1}$ is $L$-dimensional identity matrix, $g$ the group element in $G_c$, and $u_\rr(g)$ the representations in local surface Hilbert space \cite{SM}.
As the actions on the spin space has been fixed, we will focus on the representations in the $2L$ flavor degrees of freedom, where
\begin{equation}
H_F=\tau_z\otimes\mathbbm{1},\ F(T)=\tau_0\otimes\mathbbm{1}K,\ F(S)=\tau_z\otimes\mathbbm{1},
\end{equation}
and $F(g)$ constitute a linear representation of $G_c$.
Here, $\bm{\tau}$ are the Pauli matrices in the flavor space.
Note that in the flavor space, $[F(T)]^2=+1$ and $F(g)F(S)=\chi_g\det R_g F(S)F(g)$.
The $T$-, $S$-symmetric mass field on the $\nn_\rr$-surface has a general form
\begin{equation}\label{main mass field}
 M_\rr=\begin{pmatrix}
      0&m_{\rr}\\
      m_{\rr}^T &0\\
    \end{pmatrix},\ m_{\rr}=m^*_{\rr}.
\end{equation}
The spectrum is gapped at $\rr$ point if and only if $m_\rr\in \mathrm{GL}(L,\mathbbm{R})$, i.e. invertible.
However, the point-group symmetries may force the mass field to be singular (non-invertible) at certain points or lines.
To see this, consider the constraint given by the invariance of mass field under spatial symmetries \cite{SM}:
\begin{equation}\label{main mass fields constraints}
[F(g)\tau_z^{(\chi_g-1)/2}]  M_\rr [\tau_z^{(1-\chi_g)/2}F^{-1}(g)] =  M_{g\cdot \rr}.
\end{equation}

\begin{table}[t]
\renewcommand{\arraystretch}{1.3}
\LTcapwidth=0.47\textwidth
\caption{The second/third-order topological invariants $v_{2/3}(g)$ protected by single point-group symmetry.
The nontrivial topological invariants marked in black/blue indicate the appearance of surface gapless modes at high-symmetry sub-manifold/generic points.
}
\label{main text invariant table}
\centering
\vspace{0.2cm}
\begin{tabular}{c|c|c|c|c|c|c|c|c|c|c|c|c|c}
    \hline\hline
    $g$&\multicolumn{2}{c|}{$C_{2}$}&$C_3$&\multicolumn{2}{c|}{$C_{4}$}&\multicolumn{2}{c|}{$C_{6}$}&\multicolumn{2}{c|}{$M$}&$I$&$S_3$&$S_4$&$S_6$\\\hline
    $\chi_g$&$+1$&$-1$&$+1$&$+1$&$-1$&$+1$&$-1$&$+1$&$-1$&$\pm1$&$\pm1$&$\pm1$&$\pm1$\\\hline
    $ v_{2}(g)$&\multicolumn{2}{c|}{$\cog{ {\mathbbm Z}_2}$}&$\cog{ {\mathbbm Z}_2}$&\multicolumn{2}{c|}{$\cog{ {\mathbbm Z}_2}$}&\multicolumn{2}{c|}{$\cog{ {\mathbbm Z}_2}$}&${\emptyset}$&$\cob{\mathbbm Z}$&$\cog{ {\mathbbm Z}_2}$&$\cog{ {\mathbbm Z}_2}$&$\cog{ {\mathbbm Z}_2}$&$\cog{ {\mathbbm Z}_2}$\\\hline
    $ v_{3}(g)$&$\cob{\mathbbm Z}$&$\cob{\mathbbm Z_2}$&${\cob{\mathbbm Z}}\otimes \cob{{\mathbbm Z_2}}$&$\cob{\mathbbm Z^2}$&${\emptyset}$&$\cob{\mathbbm Z^3}$&$\cob{\mathbbm Z_2}$&\multicolumn{2}{c|}{$\emptyset$}&$\cog{ {\mathbbm Z}_2}$&$\cog{ {\mathbbm Z}_2}$&$\emptyset$&$\emptyset$\\\hline\hline
\end{tabular}
\renewcommand{\arraystretch}{1}
\end{table}


\noindent
There are two types of symmetry constraints, leading to two types of gapless regions respectively.

One type enforces the high-symmetry lines/points (sub-manifolds) to be gapless, corresponding to the second/third order TSC.
The high-symmetry line is invariant under the mirror symmetry and the high-symmetry point is invariant under rotational symmetry.
In this case, the point-group symmetry imposes a local constraint on the mass field and the Hamiltonian $H$ can be block-diagonalized into different symmetry sectors:
\begin{equation}
 [F(g)\tau_z^{(\chi_g-1)/2},M_\rr]=0\rightarrow H = H_{e_1}\oplus\cdots\oplus H_{e_n}
\end{equation}
In each sector, when the restricted mass field $M_{\rr}$ is singular, the high-symmetry sub-manifold on the sphere is gapless and can host Majorana zero modes.
The internal symmetries are different in each symmetry sector and we classify the anomalous $H_{e_i}$ according to the remaining symmetries of the sector, and label it by the corresponding invariant $ v_{2/3}(g)$ of the Hamiltonian on the sub-manifold \cite{PhysRevB.90.205136,RevModPhys.88.035005,fang2017topological}.
We take the mirror symmetry $M_z$ and $L=1$ as an example.
When the pairing symmetry for $M_z$ is odd, $F(M_z)=\tau_0$ is a representation that satisfies the commutation relation Eq.~(\ref{main text surface S and g}).
The mirror-invariant line, on which the mass field $m_\rr\tau_x$ vanishes, hosts one pair of Majorana zero modes with mirror eigenvalues $\pm i$.
Each mirror sector, having only particle-hole symmetry, is reduced to one-dimensional anomalous D class and thus has a $\mathbbm Z$ classification, characterized by the mirror Chern number $v_2(M_z)$ \cite{PhysRevLett.111.056403}.
Furthermore, the time-reversal symmetry relates the numbers of zero modes in the two mirror sectors.
So in the $F(M_z)=\tau_0$ case, the superconductor is a $v_2(M_z)=1$ second-order TSC.
Similar analysis can be extended to other representations and high-symmetry points that have rotational symmetries, as shown in Supplementary Materials \cite{SM}.
That is to say, we assign the topological invariants $v_{2/3}(g)$ to $F(g)$ according to the Majorana zero modes appearing at high-symmetry sub-manifold for this type of symmetry constraint.

\begin{figure}[t]
\centering 
\includegraphics[width=0.43\textwidth]{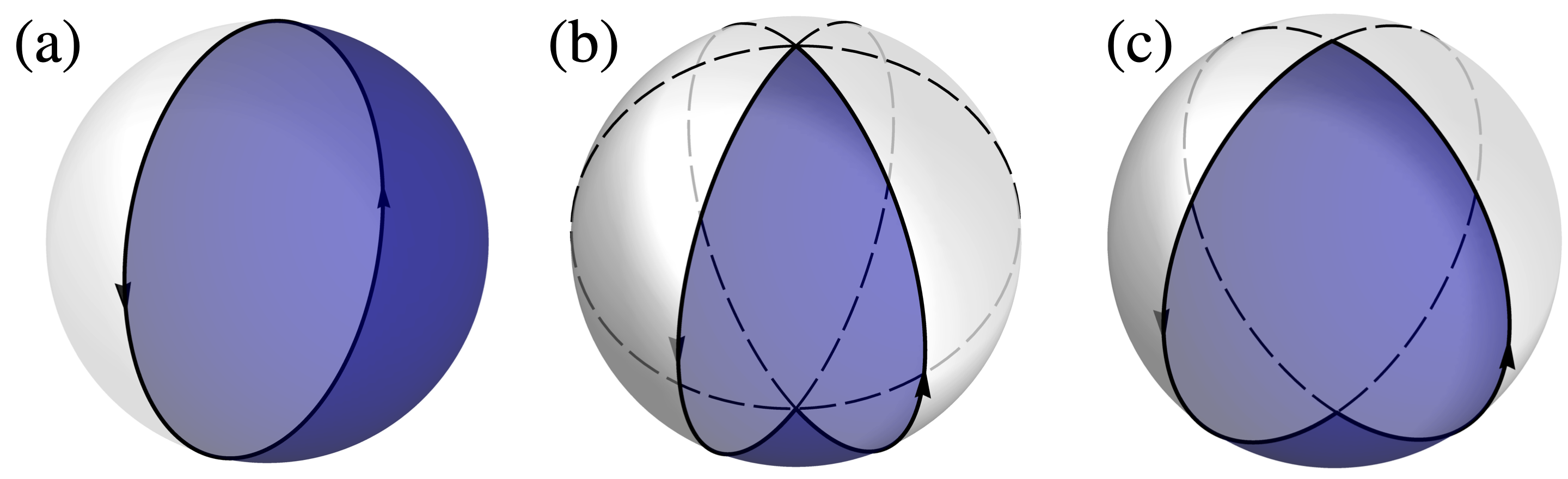}
\caption{The colored regions in (a)-(c) are the 2-cells defined by $I$, $S_3$, $S_4$, respectively.
The $\mathbbm{Z}_2$ invariants of first homotopy groups are defined along the solid black arrowed 1-loops on the sphere.
}
\label{1-loop}
\end{figure}

Besides the local symmetry constraints, there are also symmetries with no fixed point, i.e., $I$ and $S_{n}$.
Such symmetries may enforce singularities at generic point, as the result of the nontrivial zeroth/first homotopy class of the mass field.
First, consider a generic point $\rr$ and its partner $g\rr$, then if $m_\rr$ and $m_{g\rr}$ belong to different elements of $\pi_0(\mathrm{GL}(L,\mathbbm R))$, there must be a gapless nodal line separating the two points.
This implies a second-order TSC protected by $g$, the topological invariants of which can be assigned by the transformation property of mass field under $g$, i.e., the invariant of zeroth homotopy group of mass field:
\begin{equation}\label{2nd invariant def}
   v_{2}(F(g)) = \mathrm{sign} (\det m_\rr \times \det m_{g\rr}) \in \mathbbm Z_2.
\end{equation}
We prove that this $\mathbbm Z_2$ invariant is fixed by $F(g)$ and does not depend on the specific choice of mass field \cite{SM}.
Next, consider the loop $\iota$ shown in Fig.~\ref{1-loop}, which is the boundary of a \textit{2-cell} defined by $g$ \cite{song2020real}, along which the first homotopy group invariant is defined as
\begin{equation}\label{3rd invariant def}
   v_3(F(g)) = e^{i\pi W_{\iota}[m_\rr]}\in \mathbbm Z_2,
\end{equation}
where $W_\iota[m_\rr]$ is the "winding number" on of $m_\rr$ along some symmetric loops $\iota$ \cite{SM}.
The winding number is well-defined up to an even integer.
Along the loop $\iota$, $v_3(g)$ can be fixed by $F(g)$.
This is can be seen more clearly by dividing $\iota$ into two connected paths $\iota_1$ and $\iota_2=g\iota_1$, along which the mass fields are related by $F(g)$ according to Eq.~(\ref{main mass fields constraints}).
If $F(g)$ restricts the winding number on the path $\iota_{1,2}$ to $n+\frac{1}{2}$, $v_3(g)$ is $-1$.
Such a nontrivial loop promises us singularities inside, corresponding to a third-order TSC protected by $g$ labelled by $v_3(g)$.
The boundaries of 2-cells, along which $v_3$ can be fixed by $F(g)$, are shown in solid black lines in Fig.~\ref{1-loop} for $I,S_{3,4}$.
Note that the boundary of the 2-cell defined by $S_6$ cannot be fixed by $F(S_6)$, so $S_6$ does not define a third-order topological invariant.
Also, the mass field with a well-defined $v_3(g)$ requires its zeroth homotopy to be trivial.
This physically corresponds to the fact that $1$st,$\cdots$,$(d-1)$th-order topologies should be trivial when considering $d$th-order topology.
We summarize the invariants for each point-group symmetry in Table.~\ref{main text invariant table}.

\section{Gapless modes at zeros of Mass Field}
In the above section, we obtained the classification of the point group symmetry enforced singularities (zeros) of the surface mass field.
However, it is not always true that each symmetry-enforced singularity can support the Majorana zero modes.
After careful and case-by-case inspection, we found that only the singularities with $v_3(F(S_4))=1$ cannot support the Majorana zero modes.
The reason is that the four mass singularities can be moved to the north and south poles while maintaining $S_4$ symmetry, as shown in Fig.~\ref{main text exotic mass pattern I} (c).
Then, the singularities locate at the high-symmetry points with $C_2$ symmetry and host two Majorana Kramers pairs with $S_4$ eigenvalues $e^{\pm i\pi/4}$ and $e^{\pm 3i\pi/4}$ \cite{SM}.
The zero modes with eigenvalues $e^{ i\pi/4}$ and $e^{ -i3\pi/4}$ can hybrid with each other and open the gap while maintaining $C_2$ symmetry.
The same is true for the other two zero modes.
Finally, the $S_4$-enforced surface singularities do not host any Majorana zero modes, and hence the topological invariants $v_{3}(F(S_4))$ classifying the gapless modes reduces to $\emptyset$, as shown in Table.~\ref{main text invariant table}.
Moreover, this $S_4$ protected surface singularity can be viewed as a fragile state, which can be gapped out by placing two 1D TSC on the sphere, of which explicit constructions are shown in Supplementary Materials \cite{SM}.

\section{Algebraic Framework of Classification}
As discussed previously, we assign the set of invariants $v_i(g)$ to each projective co-representation (PCR) satisfying the commutation relation Eq.~(\ref{main text surface S and g}):
\begin{equation}
\bm{F}=\{F(T),F(P),F(S),F(g_1),\cdots \},\ g_i\in G_c.
\end{equation}
It means that two mass fields constrained by the same $\bm{F}$ have the same invariants and possess surface anomalies of the same type.
Enumerating all distinct surface anomalies is then transformed into enumerating all possible projective co-representations.
All PCRs form a \textit{semimodule} $V$, with direct sum as the addition operation.
\begin{table}[t]
\caption{Classification Table of TSCs protected by point groups with a single generator.
}
\label{main text table group extension}
\LTcapwidth=0.5\textwidth
\begin{tabular}{c|c|c|cccc}
    \hline \hline
    PG &Generator& Pairing Symmetry & ${\mathcal{C}}_{\rm{1st}}$    & ${\mathcal{C}}_{\rm{2nd}}$    & ${\mathcal{C}}_{\rm{3rd}}$   & ${\mathcal{C}}_{\rm{Full}}$         \\ \hline 
    \multirow{2}{*}{$C_i$}  &\multirow{2}{*}{$I$}  & $A_u$    & $\mathbbm{Z}$ & ${\mathbbm{Z}}_2$ & ${\mathbbm{Z}}_2$ & $\mathbbm{Z} \otimes \mathbbm{Z}_4$ \\
    &  & $A_g$    & $\emptyset$ & $\emptyset$  & $\emptyset$ & $\emptyset$ \\\hline
    \multirow{2}{*}{$C_2$}&\multirow{2}{*}{$C_{2z}$}  & $A$    & $\mathbbm{Z}$ & $\mathbbm{Z}_2$ & $\mathbbm{Z} $ & $\mathbbm{Z}^2$ \\
          &                 & $B$    & $\emptyset$ & $\mathbbm{Z}_2$  & $\emptyset$ & $\mathbbm{Z}_2$ \\\hline
    $C_3$&$C_{3z}$  & $A$    & $\mathbbm{Z}$ & $\emptyset$ & $\mathbbm{Z} $ & $\mathbbm{Z}^2$ \\\hline
    \multirow{2}{*}{$C_4$}&\multirow{2}{*}{$C_{4z}$}  & $A$    & $\mathbbm{Z}$ & $\mathbbm{Z}_2$ & $\mathbbm{Z}^2 $ & $\mathbbm{Z}^3$ \\
          &                 & $B$    & $\emptyset$ & $\mathbbm{Z}_2$  & $\emptyset$ & $\mathbbm{Z}_2$ \\\hline
    \multirow{2}{*}{$C_6$}&\multirow{2}{*}{$C_{6z}$}  & $A$    & $\mathbbm{Z}$ & $\mathbbm{Z}_2$ & $\mathbbm{Z}^3 $ & $\mathbbm{Z}^4$ \\
          &                 & $B$    & $\emptyset$ & $\mathbbm{Z}_2$  & $\emptyset$ & $\mathbbm{Z}_2$ \\\hline
    \multirow{2}{*}{$C_s$}&\multirow{2}{*}{$M_z$}  & $A^\prime$    & $\emptyset$ & $\emptyset$ & $\emptyset $ & $\emptyset$ \\
          &                 & $A^{\prime\prime}$    & $\mathbbm{Z}$  & $\mathbbm{Z}$& $\emptyset$  & $\mathbbm{Z}^2$ \\\hline
        \multirow{2}{*}{$C_{3i}$}&\multirow{2}{*}{$S_{3z}$}  & $A_u$    & $\mathbbm{Z}$ & $\mathbbm{Z}_2$ & $\mathbbm{Z}\otimes\mathbbm{Z}_2 $ & $\mathbbm{Z}^2\otimes\mathbbm{Z}_4$ \\
          &                 & $A_g$    &$\emptyset$ & $\emptyset$  & $\emptyset$ & $\emptyset$ \\\hline
    \multirow{2}{*}{$S_4$}&\multirow{2}{*}{$S_{4z}$}  & $A$    & $\emptyset$ & $\mathbbm{Z}_2$ & $\emptyset $ & $\mathbbm{Z}_2$ \\
          &                 & $B$    & $\mathbbm{Z}$ & $\mathbbm{Z}_2$  & $\mathbbm{Z}$ & $\mathbbm{Z}^2\otimes\mathbbm{Z}_2$ \\\hline
    \multirow{2}{*}{$C_{3h}$}&\multirow{2}{*}{$S_{6z}$}  & $A^\prime$    & $\emptyset$ & $\emptyset$ & $\emptyset $ & $\emptyset$ \\
          &                 & $A^{\prime\prime}$    & $\mathbbm{Z}$ & $\mathbbm{Z}$  & $\mathbbm{Z}$ & $\mathbbm{Z}^3$ \\\hline\hline
\end{tabular}
\end{table}
\noindent For each order, the set of invariants $ v_n(g)$ collected from all point-group actions defines a $n$-th order homomorphism from a sub-semimodule of $V$ to an abelian group (which is not necessarily surjective) \cite{SM}:
\begin{equation}
   v_n : V_n \mapsto (\overbrace {\mathbbm Z,\cdots,\mathbbm Z}^a ,\overbrace {\mathbbm Z_2,\cdots,\mathbbm Z_2}^b),\ 
  V_n = \begin{cases}
    V & n=1 \\
    \mathrm{Ker}\  v_{n-1} & n>1
  \end{cases}.
\end{equation}
where $a$/$b$ is the number of $\mathbbm Z$/$\mathbbm Z_2$ invariants from different group actions.
The $n$-th order classification group is defined as the image of $ v_n$: ${\mathcal C_{n{\mathrm{th}}}}={\rm Im}\  v_{n}.$
Specifically, the first-order invariant is the winding number
\begin{equation}
  z_w\equiv  v_1[{\bm F}] ={\rm Tr}\ [F(S)].
\end{equation}
For each element in $V$, one can calculate its first-order invariants.
All irreducible PCRs $\{{\bm F}_i\}$ form a set of basis of $V$.
The first-order classification group $\mathcal C_{\rm 1st}$ is defined as the image of $ v_1$, and the first-order invariants of irreducible PCRs $\{ v_1[{\bm F}_i]\}$ can be chosen as a set of generators:
\begin{equation}
\mathcal C_{\rm 1st}\equiv\mathrm{Im}\  v_1=\mathrm{Span}\{ v_1[{\bm F}_1], v_1[{\bm F}_2],\cdots\}.
\end{equation}
The kernel of $ v_1$ defines a sub-semimodule $V_2$ where the second-order invariants $v_2$ are well-defined.
The sub-semimodule $V_2$ has a set of basis $\{\bb_{2,i}\}$, termed as \textit{second-order irreducible building block}, which can be obtained by adding irreducible PCRs while keeping the first-order invariant zero.
That is to say, a PCR is a \textit{second-order irreducible building block} if and only if

  (i) it has vanishing first-order invariant $z_w$,

 (ii) and cannot be decomposed into several smaller PCRs satisfying (i).

\noindent The second-order classification group $\mathcal C_{\rm 2nd}$ is defined as the image of $ v_2$ and the invariants of second-order irreducible building blocks $\{ v_2[\bb_{2,i}]\}$ can be chosen as a set of generators:
\begin{equation}
\mathcal C_{\rm 2nd}\equiv\mathrm{Im}\  v_2=\mathrm{Span}\{ v_2[\bb_{2,1}], v_2[\bb_{2,2}],\cdots\}.
\end{equation}
Similar procedure extends to the third order, where \textit{third-order irreducible building block} is obtained by stacking second-order blocks with zero second-order invariant, and the classification group $\mathcal C_{\rm 3rd}$ is the image of $v_3$, which can be generated by a set of basis $\{\bb_{3,j}\}$:
\begin{equation}
\mathcal C_{\rm 3rd}\equiv\mathrm{Im}\  v_3=\mathrm{Span}\{ v_3[\bb_{3,1}], v_3[\bb_{3,2}],\cdots\}.
\end{equation}

With the classifications for each order in hand, we are only one step away from getting the full classification group $\mathcal{C}_{\rm{full}}$.
We remark that $\mathcal{C}_{\rm{full}}$ may not be the tensor product of first, second, and third-order classifications as there may be linear dependency, i.e., generated by same root states, in which case the classification group undergoes a nontrivial extension.
Since a single first-order root state cannot generate higher-order state, the nontrivial entension happens only between second- and third- order classification groups.
For example, assuming that the classification $\mathcal{C}_{\text{2nd}}$ of second-order root state is $\mathbbm{Z}_2$ and $\mathcal{C}_{\text{3rd}}$ of third-order root state is $\mathbbm{Z}_2$, if two copies of $\mathbbm{Z}_2$ second-order root states are equivalent to the third-order root state, the full classification $\mathcal{C}_{\rm{full}}$ is $\mathcal{C}_{\rm{1st}}\otimes\mathbbm{Z}_{4}$, otherwise is $\mathcal{C}_{\rm{1st}}\otimes\mathbbm{Z}_{2}\otimes\mathbbm{Z}_{2}$.
We summarize the group extension for point groups with single generator in Table.~\ref{main text table group extension} and present the results of 32 point groups in Supplementary Materials \cite{SM}.

\section{Exotic Mass Patterns}
The third-order TSC protected by nonlocal symmetry are diagnosed by the first homotopy group.
However, the nontrivial first homotopy class may result in not only the "point-like defects", but also "loop-like defects", which are also the manifestations of the surface states of third-order TSC.
A representative model is four copies of Majorana-cones with trivial first-order and second-order topology, and odd-parity pairing symmetry.
The representation $F(I)$ of inversion symmetry can be chosen as $\tau_z\rho_0$, where $\bm{\tau},\bm{\rho}$ are Pauli matrices in flavor space.
The symmetric mass terms $M_i$ are chosen to be $M_1=\tau_x\rho_0$, $M_2=\tau_y\rho_y$, $M_3=\tau_x\rho_x$ and $M_4=\tau_x\rho_z$.
The constraint Eq.~(\ref{main mass fields constraints}) imposed by inversion symmetry gives $m^i_{\rr}=-m^i_{-\rr}$.
The third-order invariant $v_3(I)$ is computed as $1$.
We now consider an adiabatic evolution of the mass field, which is given by the family of Hamiltonians
\begin{equation}
\mathcal M(\rr;t)=\sum_{i=1}^{2}\cos(t)m^i_\rr M_i\otimes\sigma_0+\sum_{i=3}^{4}\sin(t) m^i_\rr M_i\otimes\sigma_0,
\end{equation}
where $t$ is adiabatic parameter.
For $t=n\pi/2$, there are two anticommute mass terms $m^i_\rr M_i$.
The gapless region satisfying $(m^1_\rr)^2+(m^2_\rr)^2=0$ (or $(m^3_\rr)^2+(m^4_\rr)^2=0$) appears at the intersection of the zeros of two mass-terms, which are two zero-dimensional points on the sphere, as shown in Fig.~\ref{main text exotic mass pattern I}(a).
While for $t\ne n\pi/2$, four mass-terms do not anticommute.
The gapless region satisfying $(m^1_\rr)^2+(m^2_\rr)^2-(m^3_\rr)^2-(m^4_\rr)^2=0$ always has a one-dimensional solution $\iota_\rr$ on the sphere, as shown in Fig.~\ref{main text exotic mass pattern I} (b).
Therefore, the state with two surface gapless points is topologically equivalent to that with two gapless loops.

\begin{figure}[t]
\centering 
\includegraphics[width=0.43\textwidth]{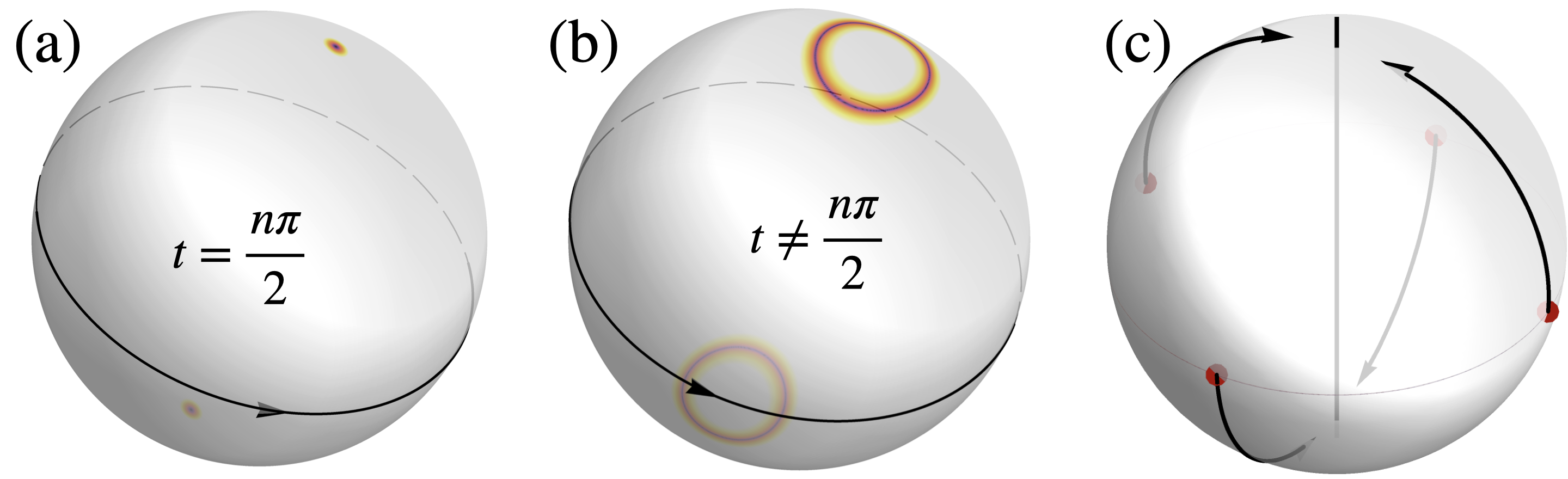}
\caption{(a) and (b) are the mass field introduced in Eq.~\ref{main mass field} at $t$ equal to $n\pi/2$ and not equal to $n\pi/2$.
The dark purple sub-manifolds represents the singularities of mass fields.
The loops with arrows are the 2-cells' boundaries, along which the "winding numbers" $v_3$ are odd.
(c) shows how to move the zeros to the north and south poles while maintaining the $S_4$ symmetry.
The red points indicate the singularities of the surface mass field.
}
\label{main text exotic mass pattern I}
\end{figure}


\section{Conclusion and discussion}
In conclusion, we investigate the classification of intrinsic TSC by a surface approach that uses Majorana cones as the minimal surface anomalous elements to construct topologically distinct surface anomalies.
To achieve this, we first classify the surface mass field restricted by point group symmetry and obtain all the symmetry-enforced singularities.
By using the homotopy group of the mass field, we provide a unified description of all possible surface singularities which are not fixed at the high-symmetry points.
Then we obtained the classification of gapless modes supported by those surface mass singularities.
We find that the classification of the gapless modes is the same as that of the mass singularities, except for those protected by the $S_4$ symmetry.
Finally, we use the n-th order irreducible building blocks to obtain the $n$-th order and full classification groups for TSC protected by all point groups.

There are two other classification approaches: the real-space construction (RSC) \cite{song2019topological,song2020real,PhysRevB.96.205106} and the equivariant $K$-theory \cite{PhysRevB.95.235425,stehouwer2018classification,cornfeld2021tenfold}, based on real-space and momentum-space perspectives of topological crystalline states, respectively.
On the one hand, the RSC method has two main difficulties in solving the classification of TSC: the cross-dimensional bubbles and the group-extension problems \cite{song2020real}.
Our surface construction method can circumvent these difficulties.
On the other hand, the $K$-theory cannot preclude states not having any surface states, considered trivial from an experimental point of view, to enter the classification whereas our method not only shows the anomalous surface states explicitly but also gives the classification results of TSC order-by-order.

\begin{acknowledgments}
Z.Z. is grateful to Shengshan Qin and Jiacheng Gao for fruitful discussions.
C.F. acknowledges funding support by the Chinese Academy of Sciences under grant number XDB33000000.
Y.Q. is supported by the National Natural Science Foundation of China (Grand Nos. 12174068 and 11874115).
\end{acknowledgments}

\bibliography{cite}

\begin{thebibliography}{110}
\bibitem [{foo()}]{footnote}%
  \BibitemOpen
\bibfield  {journal} {  }\href@noop {} {\bibinfo  {journal} {The pairing
  symmetries corresponding to high-dimensional representations imply
  spontaneous symmetry breaking. The symmetry group then reduces to the
  subgroup such that the pairing symmetry is a one-dimensional representation.
  For this reason, the consideration of one-dimensional pairing symmetries is
  complete}\ }\BibitemShut {NoStop}%
\end{thebibliography}

%

\onecolumngrid
\newpage

\begin{center}
\textbf{\large Supplemental Material for ``Topological classification of intrinsic 3D superconductors using anomalous surface construction''}
\end{center}
\appendix

\section{Time-reversal-invariant Superconductor}
In this Appendix, we briefly summarize some basis properties of time-reversal invariant superconductors and pristine-TSC mentioned in main text from the mean-field perspective.

\subsection{Bulk Theory of Generic Time-reversal-invariant Superconductor with Point-group Symmetries}
We start with a generic Hamiltonian of superconductor in the momentum space:
\begin{equation}
\hat{ H}^{\rm{SC}}=\sum_{\kk,\alpha,\beta} H_{0,\alpha\beta}(\kk)c^\dagger_{\kk,\alpha}c_{\kk,\beta}+\sum_{\kk,\kk^\prime,\alpha,\beta,\rho,\tau}V_{\alpha\beta\rho\tau}(\kk,\kk^\prime)c^\dagger_{-\kk,\alpha}c^\dagger_{\kk,\beta}c_{\kk^\prime,\rho}c_{-\kk^\prime,\tau},
\end{equation}
where $c_{\kk,\alpha}$ is the annihilation operator of the electron near the Fermi surface with momentum $\kk$ and orbital/spin $\alpha\equiv\{o,s\}$.
For a superconducting state, the formation of Cooper pairs means $c_{\kk,\alpha}c_{-\kk,\beta}$ has nonzero expectation $\Delta_{\kk}$.
So the mean-field Hamiltonian is expressed as:
\begin{equation}\label{SM mean-field Hamiltonian}
\hat{ H}^{\rm{MF}}=\sum_{\kk,\alpha,\beta} H_{0,\alpha\beta}(\kk)c^\dagger_{\kk,\alpha}c_{\kk,\beta}+\sum_{\kk,\alpha,\beta}\Delta_{\alpha\beta}(\kk)c_{\kk,\alpha}^\dagger c_{-\kk,\beta}^\dagger+\sum_{\kk,\alpha,\beta}\Delta^*_{\alpha\beta}(\kk)c_{-\kk,\beta} c_{\kk,\alpha}.
\end{equation}
The anticommutation relation of fermions gives constraint on $\Delta(\kk)$:
\begin{equation}
\Delta_{\alpha\beta}(\kk)=-\Delta_{\beta\alpha}(-\kk)
\end{equation}

Because of the pairing potential $\Delta_\kk$, the mean-field superconductor Hamiltonian is similar to the insulator from the band-theory perspective, except that the former has particle-hole symmetry.
To fully accommodate the particle-hole symmetry $P$, we rewrite Eq.~(\ref{SM mean-field Hamiltonian}) into matrix form:
\begin{equation}
\hat{ H}^{\rm{MF}}=\sum_{\mathbf{k}}(c_{\kk,\alpha},c^\dagger_{-\kk,\alpha})
\begin{pmatrix} H_{0,\alpha\beta}(\kk)&\Delta_{\alpha\beta}(\kk)\\\Delta_{\alpha\beta}^\dagger(\kk)&- H^*_{0,\alpha\beta}(-\kk)\\ \end{pmatrix}
\left(
\begin{matrix}
c_{\kk,\beta}\\
c_{-\kk,\beta}^\dagger
\end{matrix}
\right)=\sum_\kk \Psi_{\kk}  H^{\rm{BdG}}(\kk)\Psi_{\kk}^\dagger,
\end{equation}
where basis $\Psi_\kk$ is called Nambu-basis, and $ H^{\rm{BdG}}(\kk)$ is Bogoliubov-de Gennes (BdG) Hamiltonian of superconductor.
The BdG Hamiltonian $H^{\rm{BdG}}(\kk)$ has particle-hole symmetry:
\begin{equation}
P=\begin{pmatrix}0&\mathbbm{1}\\\mathbbm{1}&0 \end{pmatrix}\otimes\sigma_0K,\ P H^{\rm{BdG}}(\kk)P^\dagger=- H^{\rm{BdG}}(-\kk),
\end{equation}
$\ssigma$ are Pauli matrices in spin space and $K$ is the conjugate operator.

Assuming that the time-reversal symmetry is present and the square is $-1$, the symmetry class becomes class DIII of the 10-fold Altland-Zirnbauer (AZ) symmetry classification.
The time-reversal symmetry $T$ can be expressed as:
\begin{equation}
T=\begin{pmatrix}\mathbbm{1}&0\\0& \mathbbm{1}\end{pmatrix}\otimes(i\sigma_y)K,\ T H^{\rm{BdG}}(\kk)T^\dagger= H^{\rm{BdG}}(-\kk)
\end{equation}
The state in class DIII is labeled by the three-dimensional winding number $z_w$.
The general form for the integer-valued topological number is
\begin{equation}\label{SM 3d winding number}
z_w =  \int\frac{{{\rm{d}}k}^3}{48\pi^2} \epsilon^{\alpha\beta\gamma}\mathrm{Tr}[S H^{\rm{BdG}}(\kk)^{-1}\partial_{k_\alpha} H^{\rm{BdG}}(\kk)S H^{\rm{BdG}}(\kk)^{-1}\partial_{k_\beta} H^{\rm{BdG}}(\kk)S H^{\rm{BdG}}(\kk)^{-1}\partial_{k_\gamma} H^{\rm{BdG}}(\kk)]
\end{equation}
Qi $et\,al$ \cite{qi2010topological} show that $z_w$ of a time-reversal-invariant superconductor can be completely determined by the Fermi-surface properties when the \textit{weaking pairing limit} is satisfied.
Weak pairing limit means that pairing potential only near the Fermi surface is not negligible.
In three dimensions, $z_w$ is determined by the sign of the pairing potential and the first Chern number $C_{1s}$ of the Berry phase gauge field on the Fermi surfaces:
\begin{equation}\label{SM QI calculate zw}
  z_w=\frac{1}{2}\sum_s \text{sgn}(\delta_s)C_{1s},\ C_{1s}=\frac{1}{2\pi}\int_{FS_s}d\Omega^{ij}[\partial_i a_{sj}(\kk)-\partial_j a_{si}(\kk)],\  \delta_s= T(\langle s,\kk|)\Delta(\kk)|s,\kk\rangle
\end{equation}
with $a_{si}=-i\langle s\kk|\partial_i|s\kk\rangle$ the Berry connection defined for the $s$-band crossing Fermi-surface and $d\Omega^{ij}$ is 2-form of the Fermi surface.
More details are shown in the Ref.~\cite{qi2010topological}.
Note that the definition of 3D winding number has an ambiguity concerning the form of time-reversal symmetry $T$.
In addition, if we redefine $T$ with $-T$, which is equivalent to a gauge transformation, $z_w$ transforms to $-z_w$.
So only the difference of 3D winding number between two TSCs is meaningful \cite{qi2013axion}.
Hereafter, we only guarantee that the relative winding numbers are correct and change their sign if necessary.

Next, we assume that the spatial symmetry $\hat{g}$ is present.
If a spatial operation acts $n$ times equivalent to identity operation $\hat{E}$, we define it as an operation of order $n$.
For operation $\hat{g}$ of order $n$, the fermion annihilation (creation) operator $\hat{c}_{\alpha}^\dagger(\hat{c}_{\alpha})$ of the electron has a $U(1)$ gauge degree of freedom:
\begin{equation}
\hat{ U}=e^{i\frac{m\pi}{n}\hat{Q}},\,\,\hat{ U}\hat{c}_{\alpha}^\dagger\hat{ U}^{-1}=e^{i\frac{m\pi}{n}}\hat{c}_{\alpha}^\dagger,m=0,1,\cdots,n-1,
\end{equation}
where $\hat{Q}$ is the charge operator.
The superconducting system has a spatial symmetry $\hat{g}$, up to a phase transformation $\hat{ U}$, can be expressed as 
\begin{equation}
\hat{g}\hat{ U} \hat{ H}^{\rm{MF}}(\hat{g}\hat{ U})^{-1}=\hat{ H}^{\rm{MF}}.
\end{equation}
Under Nambu-basis, the representation of spatial symmetry can be expressed correspondingly as
\begin{equation}\label{SM U(g) def}
U^{\rm{BdG}}(g)=\left(\begin{matrix}U(g)&0\\0&\chi_gU^*(g) \end{matrix} \right),\ U^{\rm{BdG}}(g) H^{\rm{BdG}}(\kk)(U^{\rm{BdG}}(g))^\dagger= H^{\rm{BdG}}(g\kk)
\end{equation}
The transformation of normal state Hamiltonian and pairing potential under spatial operation $g$ are:
\begin{equation}
U(g) H_0(\kk)U^\dagger(g)= H_0(g\kk),\ U(g)\Delta(\kk)U^T(g)=\chi_g\Delta(g\kk)
\end{equation}
Physically, $\chi_g$ indicates the symmetry of pairing potential and thus is referred to \textit{pairing symmetry}.
For example, if spatial symmetry is $n$-fold rotational symmetry $C_n$, pairing symmetry represents the angular momentum of Cooper pairs.
Pairing symmetry reflects the commutation relation between $P$ and $U^{\rm{BdG}}(g)$:
\begin{equation}\label{SM commutation g P}
P[U^{\rm{BdG}}(g)]P^\dagger=\chi_g^*U^{\rm{BdG}*}(g).
\end{equation}
The time-reversal symmetry always commutes with spatial symmetry:
\begin{equation}\label{SM commutation g T}
T[U^{\rm{BdG}}(g)]T^\dagger=U^{\rm{BdG}*}(g),
\end{equation}
which gives $\chi_g=\chi_g^*=\pm1$ when time-reversal symmetry is present.
A corollary of Eq.~(\ref{SM commutation g P}-\ref{SM commutation g T}) is
\begin{equation}\label{SM commutation g S}
SU^{\rm{BdG}}(g)S^\dagger=\chi_gU^{\rm{BdG}}(g),\ S=T*P
\end{equation}
Moreover, the spatial symmetry restricts the three-dimensional winding number $z_w$ defined in Eq.~(\ref{SM 3d winding number}):
\begin{equation}\label{SM chiarl constraint}
z_w=\chi_g\cdot \det R_g z_w
\end{equation}
where $\det R_g=\pm1$ for proper/improper spatial symmetry.
This constraint can be directly derived by using Eq.~(\ref{SM commutation g S}) and coordinate transformation:
\begin{equation}
  \begin{aligned}
  z_w =&\int\frac{{{\rm{d}}k}^3}{48\pi^2} \epsilon^{\alpha\beta\gamma}\mathrm{Tr}[SH^{-1}(\kk)\partial_{k_\alpha}H(\kk)SH^{-1}(\kk)\partial_{k_\beta}H(\kk)SH(\kk)^{-1}\partial_{k_\gamma}H(\kk)]\\
  =&\int\frac{{{\rm{d}}k}^3}{48\pi^2} \epsilon^{\alpha\beta\gamma}\mathrm{Tr}[SU_g^{-1}H^{-1}(g\kk)U_g\partial_{k_\alpha}U_g^{-1}H(g\kk)U_gSU_g^{-1}H^{-1}(\kk)U_g\partial_{k_\beta}U_g^{-1}H(g\kk)U_gSU_g^{-1}H(g\kk)^{-1}U_g\partial_{k_\gamma}U_g^{-1}H(g\kk)U_g]\\
  =& \int\frac{{{\rm{d}}k}^3}{48\pi^2} \epsilon^{\alpha\beta\gamma}\mathrm{Tr}[(\chi_gS)H^{-1}(g\kk)U_g\partial_{k_\alpha}U_g^{-1}H(g\kk)(\chi_gS)H^{-1}(\kk)U_g\partial_{k_\beta}U_g^{-1}H(g\kk)(\chi_gS)H(g\kk)^{-1}U_g\partial_{k_\gamma}U_g^{-1}H(g\kk)]\\
  =& \int\frac{{{\rm{d}}(g k)}^3}{48\pi^2} \epsilon^{\alpha\beta\gamma}\mathrm{det}R_g \chi_g \mathrm{Tr}[SH^{-1}(g\kk)\partial_{gk_\alpha} H(g\kk)SH^{-1}(\kk) \partial_{gk_\beta} H(g\kk) SH(g\kk)^{-1} \partial_{g k_\gamma} H(g\kk)]\\
  =&\,\,\,\chi_g\cdot\mathrm{det}R_gz_w.
  \end{aligned}
\end{equation}
For simplicity, we replace the symbol $ H^{\rm{BdG}}(\kk)$ with $H(\kk)$ here.

Specifically, we use the single-band spin-triplet pairing TSC with time-reversal symmetry as ``pristine-TSC'', of which surface state is Majorana cones:
\begin{equation}\label{SM pristine-TSC Hamiltonian}
h_\pm(\kk)=
    \left( \begin{matrix}
      \kk^2/2m-\mu+\alpha \kk\cdot\ssigma & \pm\Delta_0\ssigma\cdot \kk\\
       \pm\Delta_0\ssigma\cdot \kk& -\kk^2/2m+\mu-\alpha \kk\cdot\ssigma\\
    \end{matrix} \right )_{\alpha\rightarrow0},\\
\end{equation}
whose basis is chosen to Nambu-basis $(c_{\kk,\uparrow},c_{\kk,\downarrow},-c_{-\kk,\downarrow},c_{-\kk,\uparrow})$ and $\Delta_0>0$.
The representations of time-reversal, particle-hole, chiral, spatial symmetries (denoted as $T,P,S,U_g$) are:
\begin{equation}\label{SM pristine-TSC symmetries}
T=\begin{pmatrix} i\sigma_y & 0\\ 0 & i\sigma_y \\ \end{pmatrix}K,\ P=\begin{pmatrix} 0 & i\sigma_y\\ -i\sigma_y & 0 \\ \end{pmatrix}K,\ S=\begin{pmatrix} 0 & -\sigma_0\\  \sigma_0 & 0 \\ \end{pmatrix},\ U_g=\begin{pmatrix} e^{-i\theta_g\nn_g\ssigma/2} & 0\\ 0 & \chi_ge^{-i\theta_g\nn_g\ssigma/2} \\ \end{pmatrix}
\end{equation}

The normal state Hamiltonian $h_0(\kk)=\kk^2/2m-\mu+\alpha \kk\cdot\ssigma$ has two Fermi-surfaces, and the two Fermi surfaces have opposite Chern number $C_{1s}=\pm1$.
At the same time, the two Fermi surfaces have opposite signs $\delta_s$ of pairing potential such that $z_w=\pm1$ due to Eq.~(\ref{SM QI calculate zw}).

Here, we emphasize that the pristine-TSC is the “minimal” topological superconductor in the class-DIII.
The word “minimal” we use means that they can construct any superconductor with the first-order invariant of any integer.
Since any integer z can be obtained by adding $\pm1$, the “minimal” TSC is chosen as the superconductor, which has winding number 1 or -1.

\subsection{Compatible Spatial Symmetry V.S. Incompatible Spatial Symmetry}

Because of the restriction Eq.(\ref{SM chiarl constraint}) on $z_w$ by spatial symmetry, we artificially define \textit{compatible} and \textit{incompatible}.
In the presence of point group symmetry $G$, $\chi_g$ constitute a one-dimensional real representation of $G$:
\begin{equation}
\chi_{g_i}\chi_{g_j}=\chi_{g_ig_j},\ \ \chi_g=\pm1,\ \ g_i,g_j\in G.
\end{equation}

We say a pairing symmetry is \textit{compatible} with nontrivial first-order TSCs if
\begin{equation}
\chi_g=\det R_g,\ \ \ \ \forall g\in G,
\end{equation}
or \textit{incompatible} if 
\begin{equation}
\chi_{g_0}=-\det R_{g_0},\ \ \ \ \exists g_0 \in G.
\end{equation}
An incompatible pairing symmetry has trivial first-order classification due to Eq.~(\ref{SM chiarl constraint}).
That is to say, the point-group symmetry imposes the constraint:
\begin{equation}
(\chi_g\det R_g-1)z_w=0,
\end{equation}
which holds if and only if $(\chi_g\det R_g-1)=0$ or $z_w=0$.
$(\chi_g\det R_g-1)=0$ means that $z_w$ can take any integer number and thus the point-group symmetry is compatible with any TSC.
$z_w=0$ means that the first-order invariant of TSC must be $0$ such that the first-order topology must be trivial.
It reflects that if there is a nontrivial TSC with first-order invariant $z$, then there must be a TSC with first-order invariant $-z$ such that the net first-order invariant is zero.
It also implies the representation of point-group symmetry is off-diagonal in the flavor space, which will be introduced in Appendix.~\ref{The Representation of Incompatible Point-group Symmetry}.

We also define a group action $g$ as \textit{compatible} symmetry when $\chi_g=\det R_g$ and \textit{incompatible} symmetry when $\chi_g=-\det R_g$.
For example, $C_n$ is compatible when the pairing potential is even under $n$-fold rotation, whereas it is incompatible when the pairing potential is odd under $n$-fold rotation.
Inversion is compatible when the pairing potential is odd under space-inversion, whereas it is incompatible when the pairing potential is even under space-inversion.

\section{Surface Theory of Pristine-TSC}
In this Appendix, we elaborate on the derivation of the surface theory of TCSC starting from the pristine-TSC described in Eqs.~(\ref{SM pristine-TSC Hamiltonian},\ref{SM pristine-TSC symmetries}) in bulk Hilbert space.
The pristine-TSC has the Majorana cone surface state.
We will establish the form of Hamiltonian $h_{\rr,\kk}$ in the surface Hilbert subspace at a certain point $\rr$ and the transformation matrices of the point-group symmetries $g$ that connects surface Hamiltonian in the subspaces defined at two symmetry-related points $\rr$ and $\rr^\prime=g\rr$.

\subsection{The Effective Surface Hamiltonian of Pristine-TSC}\label{SM projection procedure}
We place the pristine-TSC on a three-dimensional manifold which breaks translational symmetry in the $\nn_\rr$ direction.
The chemical potential $\mu$ appearing in normal state Hamiltonian becomes spatially dependent:
\begin{equation}
\mu(r)=\left\{
\begin{aligned}
-\mu_0&\ \ \ r>0\\
0     &\ \ \ r=0\\
\mu_0 &\ \ \ r<0\\
\end{aligned}
\right.
\end{equation}
The bulk momentum $\kk$ can be decomposed into $\kk_\parallel$ and $k_{\perp}\nn_\rr$.
Due to the breaking of translational symmetry in the $\nn_\rr$ direction, $\kk_\perp$ is not a good quantum number and replaced with $-i\partial_r$.
We next linearize the pristine-TSC Hamiltonian with winding number $z_w=\pm1$ by expanding in small momenta around the $\Gamma$ point and consider it close to the surface:
\begin{table}[H]
\renewcommand{\arraystretch}{1.2}
\centering
  \begin{tabular}{c|c}
  \hline\hline
  $z_w=+1$&$z_w=-1$\\
  \hline
  $ h(r;\kk_\parallel)=-\mu(r)t_z+\Delta_0(\kk_\parallel\cdot\ssigma t_x-i\nn_\rr\cdot\ssigma t_x\partial_r)$&$h(r;\kk_\parallel  )=-\mu(r)t_z-\Delta_0(\kk_\parallel\cdot\ssigma t_x-i\nn_\rr\cdot\ssigma t_x\partial_r)$\\
  \hline\hline
  \end{tabular}
  \renewcommand{\arraystretch}{1}
\end{table}
\noindent
where $\bm{t}$ are Pauli matrices in particle-hole space and $\Delta_0>0$ hereafter.
The relation between the surface Hamiltonian and bulk Hamiltonian is
\begin{equation}\label{SM surface DOF and bulk DOF}
 H^{{\rm{bulk}}}= H^{{\rm{surf}}}\otimes {{\rm{Cl}_{1,1}}},
\end{equation}
and thus the dimension of the surface Hamiltonian is half of the bulk Hamiltonian.
Here, $Cl_{1,1}$ is real Clifford algebra.
The surface Hilbert space is defined by the eigenspaces of $\nn_\rr\cdot\ssigma t_y$ with the eigenvalues $\pm1$, of which eigenstates are exponentially localized to the surface region.
The wavefunctions $\psi_\pm(r;\kk)$ of surface modes are
\begin{table}[H]
\renewcommand{\arraystretch}{1.3}
\centering
  \begin{tabular}{c|c}
  $\psi_+(r;\kk_\parallel )=\exp[\int_0^r dr^\prime [\mu(r^\prime)/\Delta_0]]\psi(\kk_\parallel)$&$\psi_-(r;\kk_\parallel )=\exp[\int_0^r dr^\prime [\mu(r^\prime)/\Delta_0]]\psi(\kk_\parallel)$\\
  \end{tabular}
\renewcommand{\arraystretch}{1}
\end{table}
\noindent which gives

\begin{table}[H]
\renewcommand{\arraystretch}{1.3}
\centering
  \begin{tabular}{c|c}
  $h (r;\kk)\psi(r;\kk_\parallel)=[-\mu(r)t_z(1-\nn_\rr\cdot\ssigma t_y)+\Delta_0 \kk_\parallel\cdot\ssigma t_x]\psi(r;\kk_\parallel)$&$h(r;\kk_\parallel)\psi(r;\kk_\parallel )=[-\mu(r)t_z(1+\nn_\rr\cdot\ssigma t_y)-\Delta_0 \kk_\parallel\cdot\ssigma t_x]\psi(r;\kk_\parallel)$\\
  \end{tabular}
\renewcommand{\arraystretch}{1}
\end{table}
\noindent and can be solved by
\begin{table}[H]
\renewcommand{\arraystretch}{1.3}
\centering
  \begin{tabular}{c|c}
  $\left \{ \begin{array}{l} (1-\nn_\rr\cdot\ssigma t_y)\psi_+(r;\kk_\parallel)=0 \\(\Delta_0\kk_\parallel\cdot\ssigma t_x)\psi_+(\kk_\parallel)=E_{\kk_\parallel}\psi_+(\kk_\parallel) \end{array} \right. $&$\left \{ \begin{array}{l} (1+\nn_\rr\cdot\ssigma t_y)\psi_-(r;\kk_\parallel)=0 \\(-\Delta_0\kk_\parallel\cdot\ssigma t_x)\psi_-(\kk_\parallel)=E_{\kk_\parallel}\psi_-(\kk_\parallel) \end{array} \right. .$\\
  \end{tabular}
  \renewcommand{\arraystretch}{1}
\end{table}
\noindent The first equation implies $P_\pm(\rr)\psi_\pm(\kk_\parallel)=\psi_\pm(\kk_\parallel)$, where $P_\pm(\rr)$ is the projection operator $P_\pm(\rr)\equiv\frac{1}{2}(1\pm\nn_\rr\cdot\ssigma t_y)$ satisfying $[P_\pm(\rr),\kk_\parallel\cdot\ssigma t_x]=0$.
The effective surface Hamiltonian can be obtained by projecting $(\pm\Delta_0\kk_\parallel\cdot\ssigma t_x)$ to the Hilbert subspace defined by the eigenspace of $P_\mp(\rr)$.
This can be easily achieved by introducing a basis rotation $V_\rr$ such that the projection operator is diagonal:
\begin{equation}
V_\rr=e^{i(\pi/4)\hat{\bm{n}}_\rr\cdot\ssigma t_x}\rightarrow V_\rr P_\pm(\rr) V_\rr^\dagger=\frac{1}{2}(1\mp t_z).
\end{equation}
We use the $4\times2$ matrix $p=(0,\sigma_0)^T$($(\sigma_0,0)^T$), which acts on a $4\times4$ matrix in the spin and particle-hole space to pick up the degrees of freedom corresponding to the non-zero eigenvalue of the projector $P_+(\rr)(P_-(\rr))$.
The Hamiltonian $h_s(\rr,\kk)$, $T$, $P$, and $S$ in the surface degrees of freedom can be expressed as:
\begin{table}[H]
\renewcommand{\arraystretch}{1.3}
\centering
  \begin{tabular}{c|c} 
  \hline\hline
  $p=(0,\sigma_0)^T$&$p=(\sigma_0,0)^T$\\
  \hline
  $p^TV_\rr (\Delta_0\kk_\parallel\cdot\ssigma t_x)V_\rr^\dagger p=\Delta_0 (\kk\times\nn_\rr)\cdot\ssigma$
 &$p^TV_\rr(-\Delta_0\kk_\parallel\cdot\ssigma t_x)V_\rr^\dagger p=-\Delta_0 (\kk\times\nn_\rr)\cdot\ssigma$\\
  \hline
  $p^TV_\rr TV_\rr^\dagger p=i\sigma_y$&$p^TV_\rr TV_\rr^\dagger p=i\sigma_y$\\
   $p^TV_\rr PV_\rr^\dagger p=\nn_\rr\cdot\ssigma\sigma_yt_z$&$p^TV_\rr PV_\rr^\dagger p=-\nn_\rr\cdot\ssigma\sigma_y$\\
   $p^TV_\rr S V_\rr^\dagger p=i\nn_\rr\cdot\ssigma$ &$p^TV_\rr S V_\rr^\dagger p=-i\nn_\rr\cdot\ssigma$\\
  \hline\hline
  \end{tabular}
\renewcommand{\arraystretch}{1}
\end{table}
\noindent This gives the surface Hamiltonian
\begin{equation}\label{surfaceH}
h_s^{z_w=\pm1}({\rr,\kk})=\pm\Delta_0(\kk\times\nn_\rr)\cdot\ssigma
\end{equation}
with opposite surface chiral operations $\mp i\nn_\rr\cdot\ssigma$.
Note that the surface Hamiltonian $h_s({\rr,\kk})$ projected from pristine-TSCs with opposite winding numbers differ by a relative sign that does not affect surface modes' \textit{helicity} if they have the same chiral operations.
The definition of helicity of a Hamiltonian is how many times the spin-polarization makes full rotations as momentum $\kk$ makes a full counterclockwise rotation enclosing the origin.

\subsection{Point-group Symmetries in Surface Hilbert Space}\label{SM Point-group Symmetries in Surface Hilbert Space}
As shown in Eq.~(\ref{SM chiarl constraint}), a pristine-TSC is not compatible with \textit{incompatible} spatial symmetry whose $\chi_g=-\det R_g$.
In other words, when there is \textit{incompatible} symmetry $g$, it is impossible to write the representation of $g$ on the bulk degrees of freedom of single pristine-TSC.
Therefore, we have to put together the bulk degrees of freedom of the pristine-TSCs with $z_w=1$ and $z_w=-1$, and then write down the representation of \textit{incompatible} symmetry $g$.
We will discuss the surface representation of \textit{compatible} and \textit{incompatible} point-group operations separately.

\subsubsection{The Representation of Compatible Point-group Symmetry}
A natural choice of boundary is the compact manifold $S^2$ because it is invariant under all point group symmetry operations $g$, i.e., for any $\rr$ on the surface, $g\cdot\rr$ is also on the surface, and $\hat{\bm{n}}_{g\cdot\rr}=R_g\hat{\bm{n}}_{\rr}$ relates their surface normals.
The bulk representations of point-group symmetries are described in Eq.~(\ref{SM pristine-TSC symmetries}).
Similar to the process shown above, we can obtain the surface representation at $\rr$ point from the bulk representation by using projection operator $P_\pm(\rr)$:
\begin{equation}
U_s(g;\rr) = p^TV_{g\cdot\rr}  U_g V_\rr^\dagger p
\end{equation}
The subtle thing is that $V_\rr$ after $U_g$ becomes defined at $g\rr$, reflecting that spatial symmetry connects two distinct local Hilbert subspaces defined in two symmetry-related points $\rr$ and $g\rr$.
The detailed derivation process is as follows:
\begin{equation}~\label{SM surfacesymm}
\begin{aligned}
    h_s({\rr,\kk})&=p^T V_\rr h(\kk) V_\rr^\dagger p=p^T V_\rr U^\dagger_g h(g\kk) U_g V_\rr^\dagger p\\
    &=p^T V_\rr U^\dagger_g(V_{g\cdot\rr}^\dagger pp^T V_{g\cdot\rr})h(g\kk)(V^\dagger_{g\cdot\rr} pp^T V_{g\cdot\rr}) U_g V_\rr^\dagger p\\
    &=(p^T V_\rr U^\dagger_g V_{g\cdot\rr}^\dagger p)\, h_s(g\rr,g\kk) \,(p^TV_{g\cdot\rr}  U_g V_\rr^\dagger p)\\
    &\equiv U_s^{\dagger}(g;\rr)h_s({g\rr,g\kk}) U_s(g;\rr)\\
\end{aligned}
\end{equation}
In the above derivation, we used $[V_{g\cdot\rr}  U_g V_\rr^\dagger, pp^T] =0$.
So the explicit form of surface representation of compatible symmetry $g$ can be obtained:
\begin{equation}
\renewcommand{\arraystretch}{1.4}
  \begin{array}{rl}
  &V_{g\cdot\rr}=\exp({i\frac{\pi}{4}\hat{\bm{n}}_{g\cdot\rr}\cdot\ssigma t_x})=\exp({i\frac{\pi}{4}(\det R_g)\hat{\bm{n}}_{\rr}\cdot R_g^{-1}\ssigma t_x})\Rightarrow U_s(g;\rr)=p^TV_{g\cdot\rr}  U_g V_\rr^\dagger p=\begin{pmatrix} \chi_gu_g & 0\\ 0 &  u_g\\ \end{pmatrix},
  \end{array}
\renewcommand{\arraystretch}{1}
\end{equation}
where $u_g=e^{-i\theta_g\nn_g\ssigma/2}$.
Note that the surface representation $U^{\pm}_s(g)$ of compatible symmetry is different for pristine-TSC with $z_w=\pm1$:
\begin{equation}\label{SM compatible symmetry representations}
  U^{+}_s(g;\rr)=\chi_g u_g,\,\,\,\,\,\,U^{-}_s(g;\rr)=u_g.
\end{equation}

\subsubsection{The Representation of Incompatible Point-group Symmetry}\label{The Representation of Incompatible Point-group Symmetry}
In the presence of incompatible symmetry $g$, the minimal bulk Hamiltonian of TSC and corresponding representation are
\begin{equation}
 H(\kk)=h_+(\kk)\oplus h_-(\kk),\ U_g^\dagger  H(\kk) U_g= H(g\kk),\ U_g=\begin{pmatrix}
      &&\eta_1u_g&\\&&&\eta_1\chi_gu_g\\\eta_2u_g&&&\\&\eta_2\chi_gu_g&&\\
    \end{pmatrix},
\end{equation}
where $\eta_i=\pm1$ are the relative phase factor between the representations of different pristine-TSCs.
Following the similar procedure in compatible symmetry case, we can get the surface representation $U_s(g)$:
\begin{equation}\label{SM incompatible symmetry representations}
U_s(g;\rr)=p^TV_{g\cdot\rr}  U_g V_\rr^\dagger p=\begin{pmatrix}
      &\det R_g u_g(i\hat{\bm{n}}_{\rr}\cdot\ssigma)\\-u_g(i\hat{\bm{n}}_{\rr}\cdot\ssigma)&\\
    \end{pmatrix}.
\end{equation}

For convenience, we summarize the results of this Appendix in Table.~\ref{SM Summary of Surface Theory of Pristine-TSC}.

  \begin{table}[t]
  \LTcapwidth=0.7\textwidth
  \renewcommand{\arraystretch}{1.3}
  \caption{Summary of bulk/surface theory of pristine-TSC }
  \label{SM Summary of Surface Theory of Pristine-TSC}
  \centering
  \begin{tabular}{c|c|c}
  \hline\hline
  First-order topological invariant& $z_w=+1$ & $z_w=-1$\\
  \hline
  Effective bulk Hamiltonian& $h_+(\kk)=-\mu t_z+\Delta_0\kk\cdot\ssigma t_x$&$h_-=-\mu t_z-\Delta_0\kk\cdot\ssigma t_x$\\
  Bulk Hilbert-space Basis&\multicolumn{2}{c}{$(c_{\kk\uparrow},c_{\kk\downarrow},-c_{-\kk\downarrow}^{\dagger},c_{-\kk\uparrow}^{\dagger})^T= \begin{pmatrix} \mathbbm{1}_2 & 0\\ 0 & -i\sigma_y \\ \end{pmatrix} (c_{\kk\uparrow},c_{\kk\downarrow},c_{-\kk\uparrow}^{\dagger},c_{-\kk\downarrow}^{\dagger})^T$}\\
  Bulk time-reversal representation &\multicolumn{2}{c}{$T=i\sigma_y\otimes t_0K$} \\
  Bulk particle-Hole representation&\multicolumn{2}{c}{$P=-\sigma_y\otimes t_yK$ }\\
  Bulk chiral representation &\multicolumn{2}{c}{$S=T\cdot P=-i\sigma_0\otimes t_y$ }\\
  Bulk spatial-symmetry representation& \multicolumn{2}{c}{$ \begin{pmatrix}u_g& 0\\ 0 &\chi_g\sigma_yu_g^* \sigma_y\\ \end{pmatrix}=\begin{pmatrix} u_g & 0\\ 0 &\chi_gu_g \\ \end{pmatrix}$}\\
  \hline
  Surface projector operator & $P_+(\rr)=\frac{1}{2}(1-\hat{\bm{n}}_\rr\cdot\ssigma)t_y$&$P_-(\rr)=\frac{1}{2}(1+\hat{\bm{n}}_\rr\cdot\ssigma)t_y$\\
  Basis rotation& \multicolumn{2}{c}{$V_\rr=\exp [i(\pi/4)\hat{\bm{n}}\cdot\ssigma t_x]$}\\
  Projector after basis rotation&$V_\rr P_+(\rr)V_\rr^\dagger=\frac{1}{2}(1-t_z)$ &$V_\rr P_-(\rr)V_\rr^\dagger=\frac{1}{2}(1+t_z)$\\
  \hline
  TRS after basis rotation& \multicolumn{2}{c}{$V_\rr TV_\rr^\dagger=i\sigma_yt_0K$}\\
  PHS after basis rotation& \multicolumn{2}{c}{$V_\rr PV_\rr^\dagger=\hat{\bm{n}}_\rr\cdot\ssigma\sigma_y t_zK$}\\
  Chiral symmetry after basis rotation& \multicolumn{2}{c}{$V_\rr SV_\rr^\dagger=it_z\hat{\bm{n}}_\rr\cdot\ssigma$}\\
  \hline
  Surface Hamiltonian $h_{\rr,\kk}$& $  \Delta_0(\kk\times\hat{\bm{n}}_\rr)\cdot\ssigma$&$  -\Delta_0(\kk\times\hat{\bm{n}}_\rr)\cdot\ssigma$\\
  Surface Time-reversal representation& $ {T}_s= i\sigma_yK$ &$ {T}_s= i\sigma_yK$\\
  Surface chiral representation $ {T}^s {C}_{s}$& $ S_s= -i\hat{\bm{n}}_\rr\cdot\ssigma$ &$ S_s=i\hat{\bm{n}}_\rr\cdot\ssigma$\\
  Redefined surface chiral* representation $i {T}_s {P}_{s}$& $ {S}_s= \hat{\bm{n}}_\rr\cdot\ssigma$ &${S}_s=-\hat{\bm{n}}_\rr\cdot\ssigma$\\
  Surface compatible symmetry representation&$U_s(g;\rr)=\chi_gu_g $&$U_s(g;\rr)=u_g$\\
  Surface incompatible symmetry representation &\multicolumn{2}{c}{$U_s(g;\rr)=\begin{pmatrix}
      &\eta_1\det R_g u_g(i\hat{\bm{n}}_{\rr}\cdot\ssigma)\\-\eta_2u_g(i\hat{\bm{n}}_{\rr}\cdot\ssigma)&\\
    \end{pmatrix}$}\\
  \hline
  \multicolumn{3}{l}{1. Here, $\ssigma$ and $\bm{t}$ are pauli matrices describing spin and particle-hole degrees of freedom, respectively.}\\
  \multicolumn{3}{l}{2.*Redefine chiral operation satisfying $(S_s)^2=1$.}\\
  \hline
  \hline
  \end{tabular}
  \renewcommand{\arraystretch}{1}
  \end{table}

\section{Justification for the Classification Assumption}\label{SM Justification for the Classification Assumption}
In this Appendix, we justify the assumption of establishing the full classification results of TCSCs protected by 32 point groups.
Namely, we stack several copies of Majorana cones and only enumerate all possible linear representations of point group $G$ in flavor space instead of that of double point group $G^D$. The representation of a single Majorana cone is restricted to the faithful representation of ${\rm{SU}}(2)$: $u_g=\exp({-i\theta_g\nn_g\ssigma/2})$, up to a phase factor.

The validity of the assumption can be established in three steps:
\textit{First}, we can classify the distinct bulk phases by their mass fields $ M_\rr$ in surface Hamiltonians.
\textit{Second}, the minimal surface states that occur in DIII-class surface is a single Majorana cone, and all allowed symmetry representations acting on several copies of Majorana cones give rise to all possible distinct mass fields.
\textit{Third}, enumerating all irreducible representation of double point groups $G^D$ is equivalent to enumerating all the tensor product representation $F\otimes u(g)$ of a linear representation $F(g)$ of point group and the fundamental representation $u(g)$ of ${\rm{SU}}(2)$: $\exp({-i\theta_g\nn_g\ssigma/2})$,
as we detail below.

From the perspective of dimensional reduction, a higher-order topological state can be 
adiabatically and symmetrically deformed into a product of several spatially decoupled lower-dimensional topological states[\onlinecite{PhysRevX.7.011020,PhysRevB.92.081304}].
Specifically, the 3D TCSC using one/two-dimensional topological states as ``skeletons'' can be identified as zero/one-dimensional gapless modes on a certain boundary.
On the other hand, the occurrence of edge modes is generally the result of a spatially non-uniform mass field in the surface Hamiltonian.
So it means that two mass fields that cannot be connected by symmetry-allowed perturbations
 corresponds to two topological inequivalent TCSCs.

The minimal surface state in class DIII is a single Majorana cone in the absence of any spatial symmetry.
The class DIII describes the spinful superconductors with time-reversal symmetry, of which bulk Hamiltonian is at least four dimensions.
Because particle-hole and spin degrees of freedom are at least two.
In general, the relation between the surface Hamiltonian and bulk Hamiltonian is
\begin{equation}
 H^{{\rm{bulk}}}= H^{{\rm{surf}}}\otimes {{\rm{Cl}_{1,1}}},
\end{equation}
and thus the dimension of the surface Hamiltonian is half of the bulk Hamiltonian.
So the effective Hamiltonian for gapless edge modes in DIII class is $\kk\cdot \ssigma$ (aka ``Majorana cone''), up to a basis transformation.
We use a concrete BdG Hamiltonian of ``pristine-TSC'' whose surface Hamiltonian is $(\kk\times\hat{\bm{n}}_\rr)\cdot\ssigma$ as our minimal effective surface Hamiltonian.
In order to classify all distinct topological TCSCs in class DIII, especially higher-order states, we need to enumerate all possible mass fields in surface Hamiltonian.
Any surface Hamiltonian is topologically equivalent to several copies of surface Majorana cone of pristine-TSCs, due to the fact any integer winding number $z_w$ can be obtained by adding $\pm1$.
The mass field on the whole surface is restricted by spatial symmetry constraints.
Thus, all topologically inequivalent TCSCs can be constructed by stacking pristine-TSCs with all allowed symmetry constraints.
In surface Hilbert space, this assumption can be translated into that any surface state in class DIII is topological equivalent to several copies of Majorana cones with symmetry-allowed mass terms.

Actually, the symmetry representations of spinful fermions in the crystal are the linear representations of double space groups $G^D$.
But enumerating all TCSCs with the irreducible representation of double point groups is equivalent to enumerating all TCSCs with tensor product representation $F\otimes u(g)$ of an irreducible representation of point groups and faithful representation of ${{\rm SU}(2)}$,
in the sense of constructing the classification group $\mathcal C$ of which elements are topological states.
The reason is that any \textit{double-valued} irreducible representation can be decomposed from $F\otimes u(g)$.
This can be understood by examining the continuum limit of $C_\infty$ symmetry:
the double point groups are subgroups of ${{\rm SU}(2)}$, while the point groups are subgroups of ${{\rm SO}(3)}$.
The irreducible representations $D_j$ of ${{\rm SU}(2)}$ are labeled by half-integer $j$, while they are irreducible representations of ${{\rm SO}(3)}$ when $j$ is restricted to an integer.
Then the tensor product representation $D_i\otimes D_j$ decomposes as follows:
\begin{equation}
D_i(g)\otimes D_j(g)\cong D_{i+j}(g)\oplus D_{i+j-1}(g)\oplus\cdots\oplus D_{i-j+1}(g)\oplus D_{i-j+1}(g),\ i\geqslant j ,\ i,j\in \mathbbm{N}^0/2.
\end{equation}
So all the irreducible representations $D_{k\in \mathbbm{Z}^+/2}$ of ${{\rm SU}(2)}$ can be decomposed from $D_{l\in \mathbbm N^0}\otimes D_{1/2}$.

Further, we argue that states with all the tensor product representation $F\otimes u(g)$ form a complete basis for topological states.
This is equivalent to a well-defined mathematical question:
Any \textit{double-valued} reducible representation can be decomposed to a direct sum of tensor product representations
\begin{equation}
\sum_i c_i F_i\otimes u(g),c_i\in \mathbbm Z^+.
\end{equation}
The definition of \textit{double-valued} representation is the representation whose $\chi(E)=-\chi(\bar{E})$.
$E$ and $\bar{E}$ are identity operation and $2\pi$ rotation operation, respectively.
We observe that it is always true for 32 double point groups.
Let us use the $O(432)$ point group for an explicit example.
The character table of double point group $O^D$ is:

  \begin{center}
  \LTcapwidth=0.5\textwidth
\renewcommand{\arraystretch}{1.3}
  \begin{longtable}{p{1.5cm}<{\centering}p{1.5cm}<{\centering}p{1.5cm}<{\centering}p{1.5cm}<{\centering}p{1.5cm}<{\centering}p{1.5cm}<{\centering}p{1.5cm}<{\centering}p{1.5cm}<{\centering}p{1.5cm}<{\centering}}
  \caption{The character table of double point group $O^D$}
  \label{SM character table of double O}\\
  \hline\hline
   &$E$&$\bar{E}$ & $8C_3$&$8\bar{C_3}$ &$3C_2,3\bar{C_2}$ &$6C_4$ &$6\bar{C_4}$ & $6C_2^\prime,6\bar{C_2}^\prime$\\
     \hline
    $\Gamma_1$& 1& 1&1 &1 &1 &1 &1 &1 \\
      \hline
    $\Gamma_2$& 1& 1&1 &1 &1 &$-1$ &$-1$ &$-1$\\
      \hline
    $\Gamma_3$&2& 2&$-1$ &$-1$ &2 &0 &0 &0 \\
       \hline
    $\Gamma_4$&3 &3 &0 &0 &$-1$ &1 &1 &$-1$\\
       \hline
    $\Gamma_5$&3 &3 &0 &0 &$-1$ &$-1$ &$-1$ &1\\
       \hline
    $\Gamma_6$&2 &$-2$ &1 &$-1$ &0 &$\sqrt{2}$ &$-\sqrt{2}$ &0\\
       \hline
    $\Gamma_7$&2 &$-2$ &1 &$-1$ &0 &$-\sqrt{2}$ &$\sqrt{2}$ &0\\
       \hline
    $\Gamma_8$&4 &$-4$ &1 &0 &0 &0 &0 &0\\
  \hline\hline
  \end{longtable}
  \renewcommand{\arraystretch}{1}
  \end{center}
\noindent The restricted representations of $\Gamma_{i\leqslant 5}$ are linear representations of point group $O$.
The $\Gamma_6$ representation is the restricted representation of faithful representation $D_{1/2}$ of ${{\rm{SU}}(2)}$.
By using orthogonality relationship for characters
\begin{equation}
  \sum_{t=1}^{|{\textup{G}}|}\chi_i{*}(g_t)\chi_{j}(g_t)=\frac{1}{|{\textup{G}}|}\delta_{ij}
\end{equation}
and ${\rm{Tr}}(A\otimes B)={\rm{Tr}}(A)\ {\rm{Tr}}(B)$, we can reduce the tensor product representation $\Gamma_{1-5}\otimes \Gamma_6$ to irreducible double-valued representations $\Gamma_{6-8}$ of $O^D$:
\begin{equation}
\begin{aligned}
&\Gamma_1\otimes\Gamma_6\cong\Gamma_6\\
&\Gamma_2\otimes\Gamma_6\cong\Gamma_7\\
&\Gamma_3\otimes\Gamma_6\cong\Gamma_8\\
&\Gamma_4\otimes\Gamma_6\cong\Gamma_6\oplus\Gamma_8\\
&\Gamma_5\otimes\Gamma_6\cong\Gamma_7\oplus\Gamma_8\\
\end{aligned}
\end{equation}
It means we can use the TCSCs with representations $\Gamma_{i\leqslant 5}\otimes u_g$ as $n$th-order irreducible building blocks to form classification group $\mathcal C$, of which generators determine the classification results.
All allowed TCSCs with double-valued representations $\Gamma_{6-8}$ can be found in the classification group, in the sense that the direct sum of representations corresponds to the addition of their invariant vectors.


\section{General Theory in the Surface Flavor Degrees of Freedom}
In this Appendix, we establish the general theory of TCSCs in the surface flavor Hilbert space.
In order to construct $n$th-order TCSCs, especially higher-order TCSCs, one has to consider the surface theory of $2L$-copies of minimal surface theory and add mass field that respects the protecting point-group symmetries:
\begin{equation}\label{SM general surface theory}
 H_F(\kk,\rr)=\begin{pmatrix}\mathbbm{1}_L&0\\0&-\mathbbm{1}_L\\\end{pmatrix}\otimes (\kk\times\hat{\bm{n}}_\rr)\cdot\ssigma,\ \  T_s=\begin{pmatrix}\mathbbm{1}_L&0\\0&\mathbbm{1}_L\\\end{pmatrix}\otimes i\sigma_y K,\  S_{s}=\begin{pmatrix}\mathbbm{1}_L&0\\0&-\mathbbm{1}_L\\\end{pmatrix}\otimes (\nn_\rr\cdot\ssigma),
\end{equation}
where $\mathbbm{1}_L$ is a $L$-dimensional identity matrix that represents $L$ identity copies.
Recall that, in Appendix~\ref{SM Justification for the Classification Assumption}, we showed that we could just consider the surface spatial symmetry representations to tensor product representations $F\otimes u(g)$ and restrict $u(g)$ to $u_g=\exp({-i\theta_g\nn_g\ssigma/2})$.
This means that we can consider the symmetry constraints only in flavor space with Hamiltonian and symmetry representations:
\begin{equation}\label{SM Hamiltonian in surface flavor space}
 H_{s}(\kk,\rr)=\begin{pmatrix}\mathbbm{1}_L&0\\0&-\mathbbm{1}_L\\\end{pmatrix},\ \ F_s(T)=\begin{pmatrix}\mathbbm{1}_L&0\\0&\mathbbm{1}_L\\\end{pmatrix}K,\ F_s(S)=\begin{pmatrix}\mathbbm{1}_L&0\\0&-\mathbbm{1}_L\\\end{pmatrix}.
\end{equation}
Time-reversal and chiral symmetries impose constraints on possible mass-terms $ M_\rr$ in surface Hamiltonian:
\begin{equation}
 M_\rr= M_\rr^*, \ \ \ F_s(  S) M_\rr \ F^{-1}_s(  S)= M_\rr,\ \  M_\rr= M_\rr^\dagger.
\end{equation}
The general form of $ T_s$-, $ S_s$-symmetrical mass field in the flavor space can be expressed as
\begin{equation}\label{SM form mass field}
 M_\rr=\begin{pmatrix}
      0&m_\rr\\
      m_\rr^T &0\\
    \end{pmatrix},m_\rr = m_\rr^*.
\end{equation}

\subsection{Group Structure in Flavor Space}

The TCSCs has symmetry group $G=G_c\times Z_2^T\times Z_2^S$ containing antiunitary time-reversal symmetry $T$ and unitary chiral symmetry $S$ in flavor space.
In the bulk flavor space, the representation $F(g)$ of $G$ forms a projective representation, and the factor system is 
\begin{equation}\label{SM factor system}
\begin{split}
\omega(S,T)=\omega(g_1,g_2)=\omega(T,g)=\omega(g,T)=1,\\
\omega(S,g)=\chi_g,\omega(g,S)=1, \forall g,g_i\in G_c,
\end{split}
\end{equation}
which is the equivalent statement of Eqs.(\ref{SM U(g) def}-\ref{SM commutation g S}).
Here the definition of factor system $\omega$ for a projective representation is 
\begin{equation}
F(g_1)F(g_2)=\omega(g_1,g_2)F(g_1g_2),\ \omega(g_1,g_2)\in \mathrm{U(1)}.
\end{equation}

Due to the fact that $\nn_\rr\cdot\ssigma$ in the chiral surface representation (Eq.~\ref{SM general surface theory}) is a pseudoscalar, whose transformation property under faithful representation of ${{\rm{SU}}}(2)$ is
\begin{equation}
e^{i\theta_g\nn_g\ssigma/2} (\nn_\rr\cdot\ssigma)e^{-i\theta_g\nn_g\ssigma/2}=\det R_g((R_g\nn_\rr)\cdot\ssigma)=\det R_g(\nn_{g\rr}\cdot\ssigma),
\end{equation}
the commutation relation between \textit{surface} flavor representation of $S$ and point group symmetry $g$ becomes:
\begin{equation}\label{SM factor system surface S and g}
F_s(g) F_s(S)-\omega_s(S,g) F_s(S)F_s(g)=0,\omega_s(S,g)=\chi_g\det R_g 
\end{equation}
The other projective factors $\omega_s(g_1,g_2)$ of projective \textit{surface} flavor space representation are the same as in Eq.(\ref{SM factor system}), because they are all inherited from the bulk Hilbert space representation.
For simplicity, we omit the subscript $s$ hereafter.
Next, we will discuss only the surface theory unless Specifically stated.

\subsection{The Surface Hilbert Space $\mathcal V^s$}
In Eq.~(\ref{SM general surface theory}), the Hamiltonian has two parameters $\kk$ and $\rr$, which means we treat each point on the sphere $S^2$ as a microscopically large but macroscopically small system such that $[\kk,\rr]=0$.
Any open boundary condition in the $\nn_\rr$ direction which may appear in a crystal can be described as a tangent plane of $S^2$ whose normal vector is $\nn_\rr$, and in-plane momentum is $\kk=(k_1,k_2)$, as shown in Fig.~\ref{SM_sur_momentu}.
\begin{figure}[H]
\centering 
\includegraphics[width=0.18\textwidth]{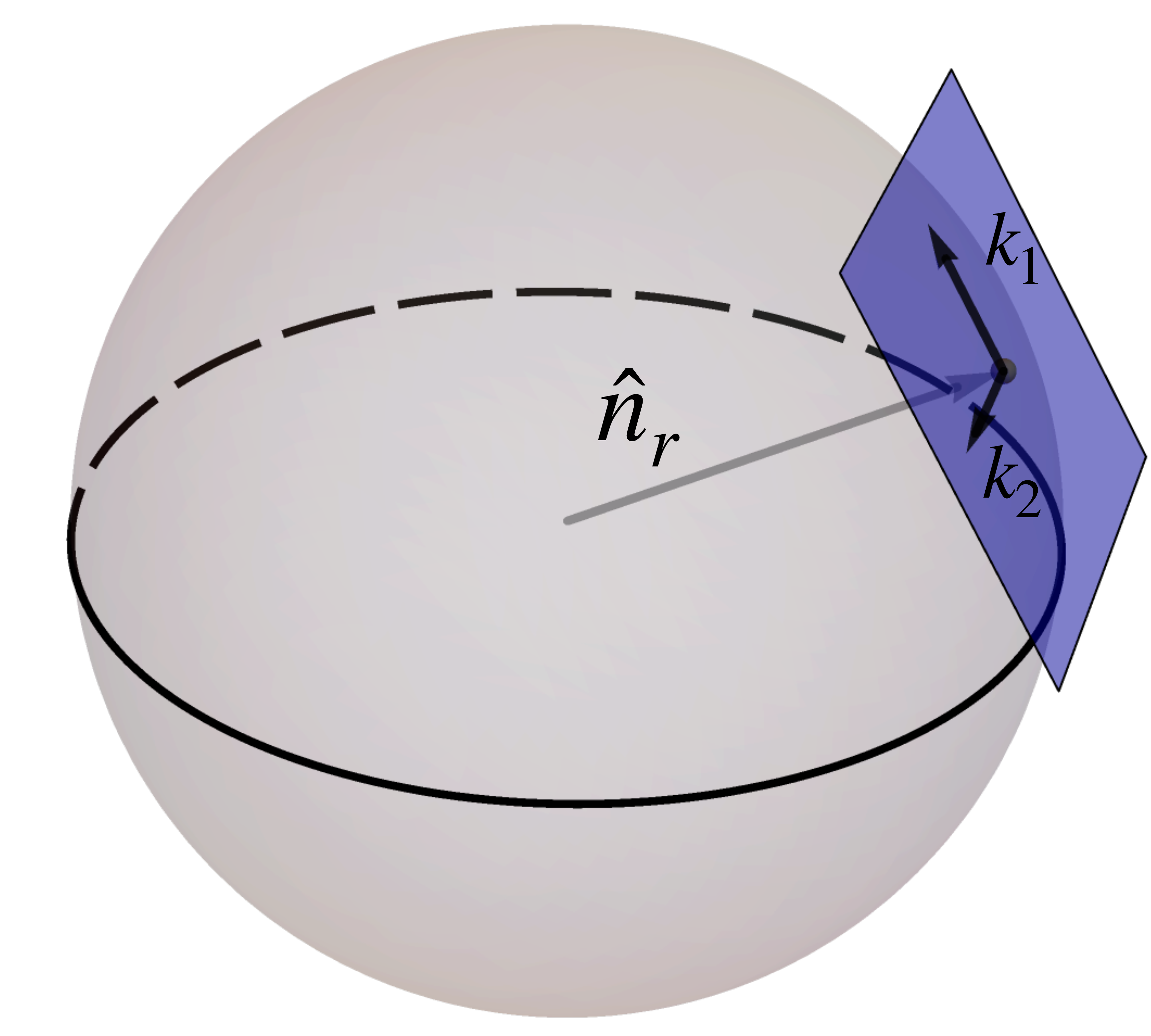}
\caption{Any point $\rr$ on the sphere describes a crystal with open boundary condition in $\nn_\rr$ direction.
$(k_1,k_2)$ are in-plane momentums which correspond to surface momentum.}
\label{SM_sur_momentu}
\end{figure}
\noindent The surface Hamiltonians at different $\rr$ points are independent of each other, except that the spatial symmetry imposes some connections on them.
Next, we will derive the symmetry constraints on surface mass field $ M_\rr$.

Recall that in the surface degrees of freedom, the transformation of surface Hamiltonian (Eq.~\ref{SM surfacesymm}) is
\begin{equation}
h_s(\kk,\rr)=U_s^{\dagger}(g;\rr)h_s({g\rr,g\kk}) U_s(g;\rr),\ U_s(g;\rr) = p^TV_{g\cdot\rr}  U_g V_\rr^\dagger p,
\end{equation}
where $U_s(g;\rr)$ is the representation in the local surface Hilbert subspace at $\rr$ point.
The transformation of point-group symmetry satisfies group associative law when it acts on the whole surface Hilbert space
\begin{equation}
\mathcal V^s= \bigoplus_{\rr\in S^2} \mathcal V_\rr^s
\end{equation}
 instead of Hilbert surface $\mathcal V_\rr^s$ at a single point $\rr$.
The point-group symmetry representations $U_s(g)$ in $\mathcal V^s$ are 
\begin{equation}
U_s(g) = \hat{U}_{\rr\rightarrow g\rr}*U_s(g;\rr),\ \hat{U}_{\rr\rightarrow g \rr}\equiv
\begin{array}{lc}
\mbox{}&
\begin{array}{cc} \ &\mathcal V_\rr^s \end{array}\\
\begin{array}{c} \\ \\\mathcal V_{g\rr}^s\end{array}&
\begin{bmatrix}
0      & \cdots &\cdots & 0      \\
\vdots & \ddots &\cdots & \cdots \\
\vdots & \cdots &U_s(g;\rr) & \vdots \\
0      & \cdots &\cdots & 0
\end{bmatrix}
\end{array},
\end{equation}
where $\hat{U}_{\rr\rightarrow g\rr}$ is real space coordinates transformation from $\rr$ to $g\rr$.
So strictly speaking, the transformation of surface Hamiltonian $ H_s(\kk,\rr)$ should be
\begin{equation}
g\cdot  H_s(\kk,\rr)\equiv (\hat{U}_{\rr\rightarrow g\cdot \rr}\otimes U_s(g;\rr))  H_s(\kk,\rr)(U_s(g;\rr)^\dagger \otimes \hat{U}^\dagger_{\rr\rightarrow g\cdot \rr}) =   H_s(g\kk,g\rr).
\end{equation}
It means point group symmetry transformation in local Hilbert subspace $ V^s_\rr$ cannot constitute a representation satisfying group associative law:
\begin{equation}
U_s(g_2,\rr)\ U_s(g_1,\rr)\neq U_s(g_1g_2,\rr),
\end{equation}
while point group symmetry transformation in whole Hilbert space $ V^s$ constitutes a representation satisfying group associative law:
\begin{equation}
[\hat{U}_{g_1\rr\rightarrow g_1g_2\rr}*U_s(g_2,g_1\rr)]\ [\hat{U}_{\rr\rightarrow g_1\rr} U_s(g_1,\rr)]= U_s(g_1g_2,\rr),
\end{equation}
For simplicity, we omit the real-space transformation $\hat{U}_{\rr\rightarrow g\rr}$ operation hereafter.

\subsection{The Point-group Symmetry Constraints $F_s^\prime(g)$ on Surface Mass-field ${{ M}}_{r}$}
Recall arguments in Appendix~\ref{SM Justification for the Classification Assumption}, the surface symmetry representations $U_s(g)$ of point-group operations can be expressed as tensor product representation of $F(g)$ and $u(g)$.
But the symmetry transformation $F^\prime(g)$ in flavor space of mass field $ M_\rr$ can be obtained by multiplying representation $F(g)$ defined in local flavor space by a matrix resulting from real space coordinates transformation $U_{\rr\rightarrow g\rr}$:
\begin{equation}\label{SM mass fields constraints}
g\cdot M_\rr\equiv F^\prime(g)  M_\rr F^{\prime-1}(g) =  M_{g\cdot\rr}, F^\prime(g)=F(g)\times \mathrm{Matrix}.
\end{equation}

To derive the concrete form of the Matrix, we consider the case of two copies of pristine-TSCs with winding number $z_w=\pm1$, and the case of more copies is just a simple generalization of two copies.
As shown in Appendix~\ref{SM Point-group Symmetries in Surface Hilbert Space}, the \textit{compatible} and \textit{incompatible} symmetry representations at point $\nn_\rr$ have forms:
\begin{equation}
U(g,\rr)=\begin{pmatrix} \eta_1\chi_gu_g & 0\\ 0 &\eta_2u_g \\ \end{pmatrix}
,\ \ \begin{pmatrix}
      &\eta_1\mathrm{det}R_g u_g(i\hat{\bm{n}}_{\rr}\cdot\ssigma)\\-\eta_2u_g(i\hat{\bm{n}}_{\rr}\cdot\ssigma)&\\
    \end{pmatrix},
\end{equation}
respectively.
$\eta_{i}$ defines the flavor space representations of compatible/incompatible symmetries:
\begin{equation}
F(g)=\begin{pmatrix} \eta_1  & 0\\ 0 &\eta_2  \\ \end{pmatrix}
,\ \ \begin{pmatrix}
      &\eta_1 \\\eta_2&\\
    \end{pmatrix},
\end{equation}
which form a projective representation of $G$ with factor system described in Eq.~(\ref{SM factor system surface S and g}).
But the symmetry constraints acting on mass field in surface flavor space is
\begin{equation}\label{SM the transformation matrix of point-group symmetry}
F^\prime(g)=\begin{pmatrix} \eta_1\chi_g  & 0\\ 0 &\eta_2  \\ \end{pmatrix}
,\ \ \begin{pmatrix}
      &\eta_1 \det R_g\\-\eta_2&\\
    \end{pmatrix},
\end{equation}
when we fix the form of representations in single pristine-TSC degrees of freedom:
\begin{equation}
u(g)=e^{-i\theta_g\nn_g\ssigma/2},\ \ e^{-i\theta_g\nn_g\ssigma/2}(i\nn_\rr\cdot\ssigma).
\end{equation}
More specifically,
\begin{equation}
\begin{aligned}
&\text{compatible symmetry}&
U_g&\equiv (F(g)\times \mathrm{Matrix})\otimes  u_g =\begin{pmatrix} \eta_1\chi_g  & 0\\ 0 &\eta_2 \\ \end{pmatrix}\otimes u_g\\
&&&=\left[\begin{pmatrix} \eta_1  & 0\\ 0 &\eta_2 \\ \end{pmatrix}\times\begin{pmatrix} \chi_g& 0\\ 0 &1 \\ \end{pmatrix}\right]\otimes u_g\\
&\text{incompatible symmetry}&
U_g&\equiv (F(g)\times \mathrm{Matrix})\otimes  u_g =\begin{pmatrix}
      &\eta_1\mathrm{det}R_g \\-\eta_2&\\
    \end{pmatrix}\otimes[u_g(i\hat{\bm{n}}_{\rr}\cdot\ssigma)]\\
&&&=\left[\begin{pmatrix}
      &\eta_1 \\\eta_2&\\
    \end{pmatrix}\times\begin{pmatrix} -1& 0\\ 0 &\det R_g \\ \end{pmatrix}\right]\otimes [u_g(i\hat{\bm{n}}_{\rr}\cdot\ssigma)]\\
\end{aligned}
\end{equation}
\begin{equation}\label{bulk to surface transformation}
\Longrightarrow\ \ \ \ 
\begin{aligned}
&\text{compatible proper symmetry}&      F^\prime(g)=F(g)\times \ \ \tau_0  \\
&\text{compatible improper symmetry}&    F^\prime(g)=F(g)\times -\tau_z  \\
&\text{incompatible proper symmetry}&    F^\prime(g)=F(g)\times -\tau_z \\
&\text{incompatible improper symmetry}&  F^\prime(g)=F(g)\times -\tau_0 \\
\end{aligned}
\end{equation}
Thus, the symmetry constraint on the mass field written in compact form is:
\begin{equation}
[F(g)\tau_z^{(\chi_g-1)/2}]  M_\rr [\tau_z^{(1-\chi_g)/2}F^{-1}(g)] =  M_{g\cdot \rr}.
\end{equation}

\section{Classification of Gapless Modes on the Sphere and Homotopy Group Theory}\label{Classification of Gapless Modes on the Sphere and Homotopy Group Theory}
In this Appendix, we classify the distinct mass field on the sphere under symmetry constraints Eq.~(\ref{SM mass fields constraints}) by homotopy group theory and give the method to calculate the invariants.
We also give some concrete examples which calculate topological invariants of TCSCs under certain surface symmetry representations.

For $2L$-copies pristine-TSCs with Hamiltonian Eq.~(\ref{SM Hamiltonian in surface flavor space}), the mass field $ M_\rr$ can be determined by a $\rr$-dependent $L$-dimensional invertible matrix $m_\rr$ satisfying Eq.(\ref{SM mass fields constraints}) except for some gapless points.
The invertible matrix defines a set of continuous maps
\begin{equation}
\mathcal F: S^2\mapsto{\rm{GL}}(L,\mathbbm{R}).
\end{equation}
The gapless regions (lines/points) can be regarded as singularities on this map because the gapless region means it has at least one zero eigenvalues such that the determinant
\begin{equation}
\det m_\rr =\prod_{i=1}^{2L}\lambda_i=0,
\end{equation}
where $\lambda_i$ is $i$-th eigenvalue of $m_\rr$.
A matrix with zero determinant has not full rank, so it is not element in ${\rm{GL}}(L,\mathbbm{R})$.
So the classification of distinct mass field on the sphere can be transformed to homotopy group $\pi_n$ of ${\rm{GL}}(L,\mathbbm{R})$.

\subsection{Zeroth Homotopy Group of Mass-field}

The mass field $ M_\rr$ has gapless points on the path connecting a gapped general point $\rr$ and its image under $g$ if $\det (m_\rr m_{g\rr})=-1$, while there is a gapped line connecting the two points if $\det (m_\rr m_{g\rr})=1$.
So any mass field $ M_\rr$ having negative determinant has a globally irremovable gapless line, where $\det  M_\rr=0$, hosting a helical Majorana mode.
The state with one-dimensional gapless lines corresponds to a second-order state labeled by zeroth homotopy group $\pi_0({\rm{GL}}(L,\mathbbm{R}))\cong\mathbbm{Z}_2$.
The two points are chosen as an arbitrary gapped point on the sphere and its image under $g$.
The point-group symmetry representations can enforce the negative determinant.
Consider a concrete two-copies of the pristine-TSCs model, of which surface Hamiltonian is:
\begin{equation}
 H_s(\rr,\kk) = \tau_z\otimes (\kk\times\nn_\rr)\cdot\ssigma,
\end{equation}
where $\ttau$ and $\ssigma$ are the Pauli matrices on flavor and single pristine-TSC surface degrees of freedom, respectively.
The time-reversal and chiral symmetries are 
\begin{equation}
 T_s=\tau_0\otimes i\sigma_y K, \ \ {S}_s =\tau_z\otimes\nn_\rr\cdot\ssigma.
\end{equation}
The only possible $ T_s$-, $ S_s$-symmetric mass term is:
\begin{equation}
 M_\rr=\begin{pmatrix} 0&m_\rr\\m_\rr&0  \end{pmatrix}\otimes \sigma_0=m_\rr\tau_x\otimes\sigma_0,\ \ m_\rr\in\mathbbm R
\end{equation}
So the TCSCs is a $\mathbbm{Z}_2=1$ second-order state when Eq.(\ref{SM mass fields constraints}) imposes the constraint
\begin{equation}\label{SM 2nd symmetry constraints}
F^\prime(g)\begin{pmatrix} 0&m_\rr\\m_\rr&0  \end{pmatrix}\ \ F^\prime(g)=\begin{pmatrix} 0&m_{g\cdot\rr}\\m_{g\cdot\rr}&0  \end{pmatrix} \Rightarrow m_\rr=-m_{g\cdot\rr}.
\end{equation}
We translate Eq.~(\ref{SM 2nd symmetry constraints}) into a more physical language:
The mass-term can change slowly with respect to the scale of correlation length because each point $\rr$ on the sphere is a microscopically large but macroscopically small system.
The effective surface Hamiltonian becomes
\begin{equation}
 H_s= \iint_{S^2}\dd\rr\ \tau_z\otimes (\kk\times\nn_\rr)\cdot\ssigma+\iint_{S^2}\dd\rr\ \hat{\Psi}_\rr^\dagger M_\rr\hat{\Psi}_\rr
\end{equation}
In the whole surface Hilbert space $\mathcal V^s$, there is point group symmetry $g$ due to
\begin{equation}
[ H_s,U_s(g)]=0.
\end{equation}
It requires $m_\rr=-m_{g\cdot\rr}$, which necessitates the presence of domain walls on the sphere.
\begin{figure}[H]
\centering 
\includegraphics[width=0.17\textwidth]{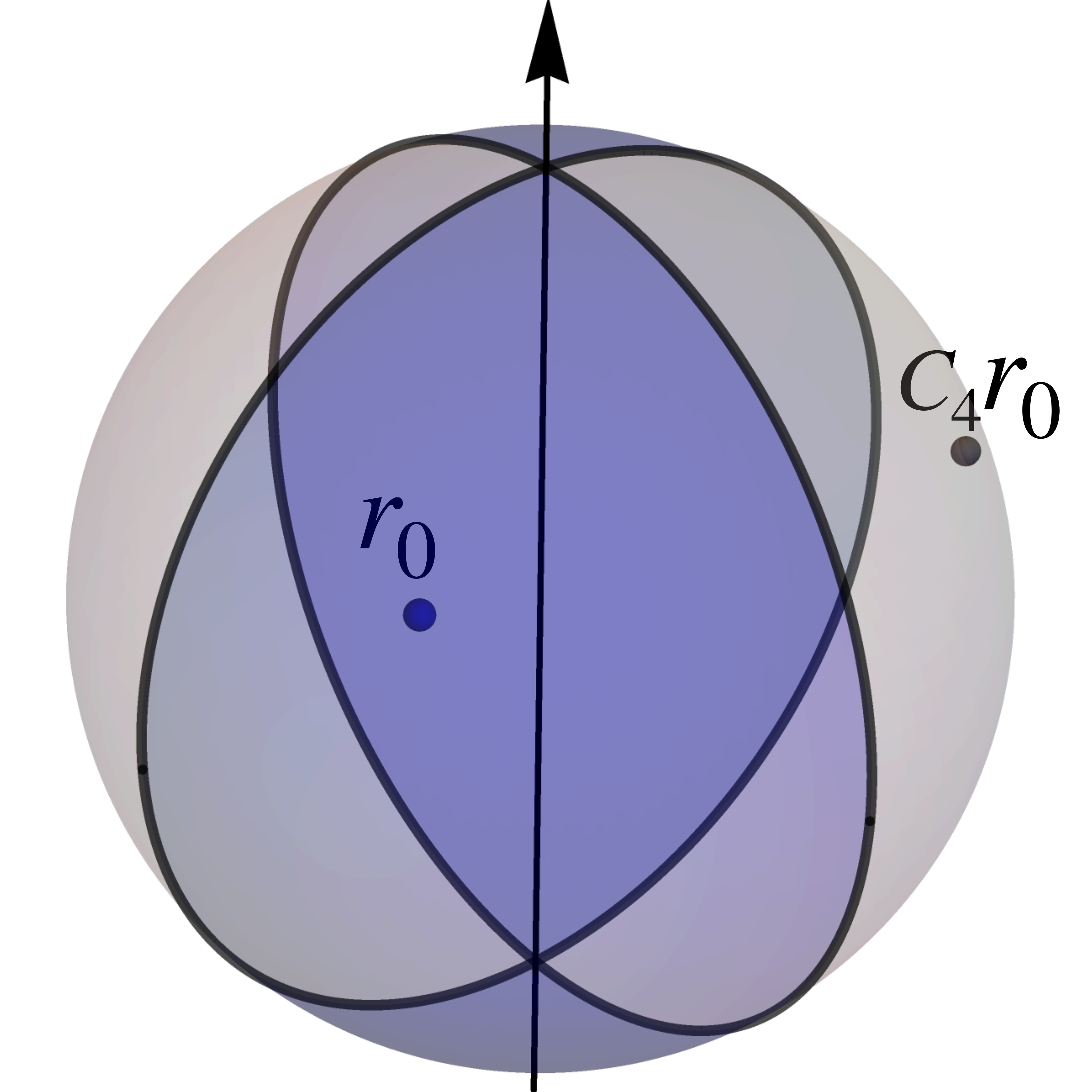}
\caption{Schematics of second-order surface states protected by incompatible $C_4$ symmetry, which is labeled by $\mathbbm{Z}_2=1$.
Two blue points are arbitrary points $\rr_0$ and its image under $C_4$, at which mass field $m_\rr$ have opposite sign.
Red lines represent gapless regions on the sphere due to spatial symmetry constraints.
}
\label{SM_2nd_example}
\end{figure}
Let us use a TCSCs with $C_4$-rotational odd pairing potential for an explicit example.
The only nontrivial factor of projective representation is
\begin{equation}
\omega(S,C_4)=-1.
\end{equation}
A reasonable choice of flavor space representation $F(C_4)$ and corresponding transformation matrix $F^\prime(C_4)$ is
\begin{equation}
F(C_4) = \begin{pmatrix}0&1\\1&0 \end{pmatrix},\ \ F^\prime(C_4) = \begin{pmatrix}0&1\\1&0\end{pmatrix}\begin{pmatrix}-1&0\\0&\det C_4 \end{pmatrix}=\begin{pmatrix}0&1\\-1&0 \end{pmatrix}.
\end{equation}
The mass field $m_\rr$ changes sign under $C_4$ operation due to
\begin{equation}
\{F^\prime(C_4),\tau_x\}=0\Rightarrow m_\rr=-m_{C_4\rr}.
\end{equation}
It requires that there is at least one gapless point on the path connecting $\rr$ and $C_4\rr$ as shown in Fig.\ref{SM_2nd_example} and corresponds to a second-order nontrivial TCSC with $\mathbbm Z_2=1$.
Another reasonable choice of flavor space representation $F(C_4)$ and corresponding transformation matrix $F^\prime(C_4)$ is
\begin{equation}
F_s(C_4) = \begin{pmatrix}0&-1\\1&0 \end{pmatrix},\ \ F_s^\prime(C_4) = \begin{pmatrix}0&-1\\1&0\end{pmatrix}\begin{pmatrix}-1&0\\0&\det C_4 \end{pmatrix}=\begin{pmatrix}0&-1\\-1&0 \end{pmatrix}.
\end{equation}
The mass field $m_\rr$ does not changes sign under $C_4$ operation due to
\begin{equation}
[F^\prime(C_4),\tau_x]=0\Rightarrow m_\rr=m_{C_4\rr}.
\end{equation}
It means that there is a uniform mass field on the sphere to fully gap the surface Hamiltoniann, which corresponds to a second-order trivial TCSC with $\mathbbm Z_2=0$.

Let us take a TCSC with even-parity pairing potential as another example. The only nontrivial factor of projective representation is
\begin{equation}
\omega(S,I)=-1.
\end{equation}
There is only one reasonable choice of flavor space representation $F(I)$ and corresponding transformation matrix $F^\prime(I)$:
\begin{equation}
F(I) = \begin{pmatrix}0&1\\1&0 \end{pmatrix},\ \ F^\prime(I) = \begin{pmatrix}0&1\\1&0\end{pmatrix}\begin{pmatrix}-1&0\\0&\det I \end{pmatrix}=\begin{pmatrix}0&-1\\-1&0 \end{pmatrix}.
\end{equation}
It means that there is a uniform mass field on the sphere to fully gap the surface Hamiltonian, which corresponds to a second-order trivial TCSC with $\mathbbm Z_2=0$.
So the classification of second-order with even-parity pairing potential reduces to $\emptyset$.

\subsection{Fundamental Group of Mass-field}
If the mass field $m_\rr$ fall into the trivial zeroth homotopy class, we can find a loop, $\iota$ on which the mass field belongs to one of two connected components ${{\rm{GL}}^\pm(L,\mathbbm R)}$, to calculate the invariant of the fundamental group (aka first homotopy group), which indicates the existence of stable defects in the area surrounded by $\iota$.
The group ${{\rm{GL}}^+(L,\mathbbm R)}$ has a fundamental group isomorphic to $\mathbbm Z_2$ for $L>2$.

This is consistent with the classification of topological defects: the time-reversal and chiral symmetry can stabilize $\mathbbm Z_2$ zero-dimensional defects.
But in general, the stability of zero-dimensional defects is a necessary but not sufficient condition for the occurrence of the surface state of third-order TCSC.
The reason is that a pair of zero-dimensional defects can be moved to the same point and annihilated together.
So some spatial symmetry is needed to keep them in different positions on the surface.

The spatial symmetries in 32 point groups can be divided into two types, one is the symmetry with fixed points, such as $C_n,M$, and the other is the symmetry without any fixed point, such as $S_{3,4,6},I$.
For mass field protected by the symmetry with fixed points will changes on-site symmetries at fixed points such that the mass field may not be an element in ${{\rm{GL}}(L,\mathbbm R)}$ on the sub-manifold, which will be discussed in Appendix \ref{SM Gapless Modes at Fixed Points on the Sphere}.

We now consider fundamental groups of mass field protected by spatial symmetry without fixed points.
For a third-order TCSCs whose mass field fall into trivial zeroth homotopy class, there are some locally stable but globally unfixed gapless points/circles on the sphere.
It can be moved on the sphere by adding symmetry-preserving terms.
But in each \textit{2-cell}, gapless region must exist if it is a third-order nontrivial state.
\begin{figure}[H]
\centering 
\includegraphics[width=0.5\textwidth]{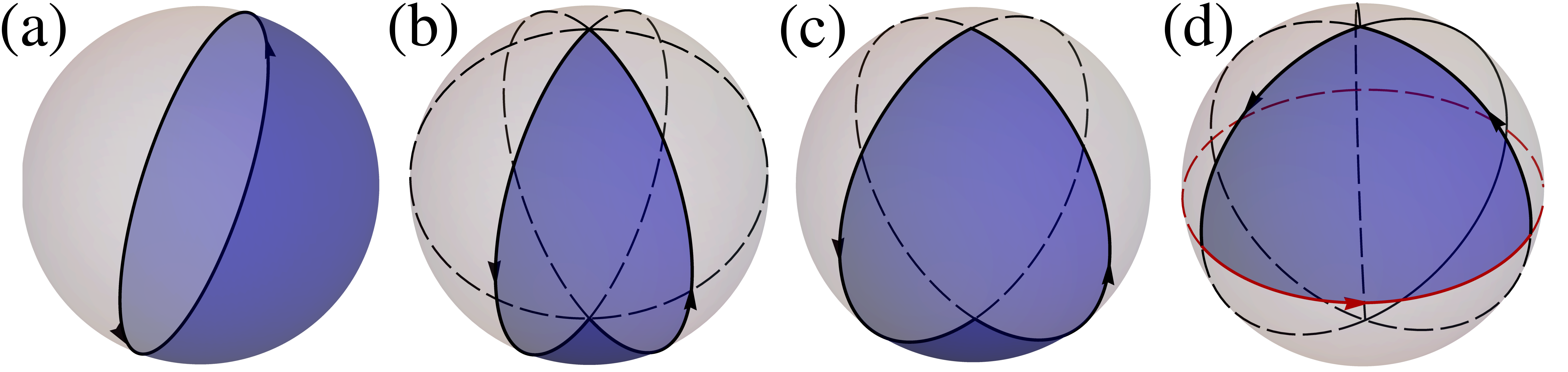}
\caption{The colored regions in (a)-(d) are possible choices of 2-cells when $I$, $S_3$, $S_4$, $S_6$ symmetries exist, respectively.
The $\mathbbm{Z}_2$ invariants of the first homotopy groups are defined on the solid black 1-loops on the sphere.
The solid red line in (d) is a path that is not related to another path on the 1-loop by symmetries.
}
\label{SM_1-loop}
\end{figure}
\noindent Here the definition of \textit{2-cell} is two-dimensional region on sphere in which no two distinct points in the same $2$-cell are related under point-group symmetry \cite{song2019topological}, as shown in Fig.~\ref{SM_1-loop}.
The existence of the gapless region in 2-cell can be labeled by the first homotopy group
\begin{equation}
\pi_1({\rm{GL}}(L,\mathbbm{R}))\cong \mathbbm{Z}_2.
\end{equation}
The $\mathbbm{Z}_2$ invariant can be calculated by parity of the "winding number" along the boundary of the $2$-cell.

The point-group symmetry representation enforces the parity of this winding number when we choose some certain loops:
If $\iota$ can be divided into several connected paths
\begin{equation}
\iota=\iota_1*\iota_2*\cdots*\iota_n,
\end{equation}
where the definition of the loop is
\begin{equation}
\iota_i :[0,1]\rightarrow {\rm{GL}}(L,\mathbbm{R}),\ \ \iota_i=g^i\iota_1,\ \  \iota_{i-1}(1)=\iota_i(0),\ \ \iota_n(1)=\iota_1(0),
\end{equation}
and $*$ represents the product of two paths:
\begin{equation}
\iota_i *\iota_j(s)=\left\{ \begin{array}{l}\iota_i(2s)\ \ \ \ \ \ \ \ \ \ \ 0\leqslant s\leqslant\frac{1}{2}\\\iota_j(2s-1)\ \ \ \ \frac{1}{2}\leqslant s\leqslant 1\end{array}\right. .
\end{equation}
The mass field $ M(\iota_1)$ is related to $ M(\iota_i)$ by relationship
\begin{equation}\label{SM loop constraints}
 M(\iota_i) = F^\prime(g^i) M(\iota_1)F^{\prime\dagger}(g^i).
\end{equation}
Thus the invariant of $\pi_1({\rm{GL}}(L,\mathbbm{R}))$ can be fixed by $F^\prime(g)$ so that it can be used to label TCSCs.
The boundaries of 2-cell satisfying the condition Eq.~(\ref{SM loop constraints}) are shown in solid black lines in Fig.~\ref{SM_1-loop}(a)-(c), while the boundary of 2-cell defined by $S_6$ does not satisfy the condition.

We now prove an important theorem that simplifies the enumerating procedure.
\begin{theorem}\label{thm1}
  The fundamental group of $\mathrm{GL}(L,\mathbb R)$ is isomorphic to the fundamental group of $\mathrm{SO}(L)$.
\end{theorem}
\begin{proof}
For a given loop $f: t\in [0,1] \rightarrow M\in \mathrm{GL}(L,\mathbbm R)$, for each $t \in [0,1]$, the columns of $M(t)$ can be regarded as $L$ linearly independent vectors. 
Using the standard QR decomposition (or Gram-Schmidt orthogonalization algorithm), we get
\begin{equation}
  M(t) = Q(t) R(t),
\end{equation}
where $Q(t) \in \mathrm{SO}(L)$ and $R(t)$ is upper-triangle matrix. 
We can then define a deformation:
\begin{equation}
  g(t,s) \equiv Q(t) R(t,s), \quad g(t,0) = M(t), \quad g(t,1) = Q(t).
\end{equation}
where $R(t,s)$ is
\begin{equation}
  R_{mn}(t,s) = \begin{cases}
    R_{mm}(t) & m=n \\
    (1-s)R_{mn}(t) & m < n
  \end{cases}.
\end{equation}
In this way we proved that for any loop in $\mathrm{GL}(L,\mathbb R)$, there is a loop in $\mathrm{SO}(L)$ that is homotopic to it. Since $\mathrm{SO}(L) \subset \mathrm{GL}(L,\mathbb R)$, the homotopy class of loops in $\mathrm{GL}(L,\mathbb R)$ and $\mathrm{SO}(L)$ is one-to-one correspondent,
\begin{equation}
  \pi_1(\mathrm{GL}(L,\mathbb R)) \cong \pi_1(\mathrm{SO}(L)) \cong \begin{cases}
    \mathbbm Z & L = 2 \\
    \mathbbm Z_2 & L>2
  \end{cases}.
\end{equation}
\end{proof}
We remark here that Theorem \ref{thm1} essentially tells us that for the purpose of the third-order classification of the mass field, it is sufficient to consider only the mass term in $\mathrm{SO}(L)$.

In the following, we will discuss how the invariant is computed, and why it is a well-defined topological invariant.
We first prove a crucial property of the fundamental groups of Lie groups.
\begin{lemma}\label{lemma1}
  For any Lie group $G$, the homotopy class of the loop $fg(t) \equiv f(t)g(t)$ is $f * g$. 
\end{lemma}
\begin{proof}
Without loss of generality, we set the base point of any loop on $G$ to be the identity element, i.e. $f(0)=f(1)=e$.
We then write down the explicit homotopy between $fg(t)$ and $(f*g)(t)$:
\begin{equation}
  F(t,s) = \begin{cases}
    f(2t) g(2ts) & 0\le t\le \frac{1}{2} \\
    g(2(1-s)t+2s-1) & \frac{1}{2} \le t \le 1
  \end{cases}.
\end{equation}
Note that $F(t,0)=(f*g)(t)$ and $F(t,1)=fg(t)$.
\end{proof}

\begin{corollary}\label{cor1}
For two loops $f, g$ defined on $\mathrm{SO}(2 N)$ and $\mathrm{SO}(2 M)$ respectively, the direct sum
\begin{equation}
  (f \oplus g)(t) \equiv f(t) \oplus g(t) \in \mathrm{SO}(2N+2M)
\end{equation}
defines a loop on $\mathrm{SO}(2N+2M)$.
The homotopy class of $f \oplus g$ has the $\mathbbm Z_2$ invariant $\xi=\xi_{f}+\xi_{g}$.
\end{corollary}
\begin{proof}
We assume the loop $f$ and $g$ is labeled by the $\mathbbm Z_2$ invariants $\xi_f$ and $\xi_g$ respectively.
Define two embeddings from $\mathrm{SO}(2N)$ and $\mathrm{SO}(2M)$ to $\mathrm{SO}(2N+2M)$:
\begin{equation}
\begin{aligned}
  \tilde{f}(t) &\equiv f(t)\oplus \mathbbm I_{2M \times 2M}, \\
  \tilde{g}(t) &\equiv \mathbbm I_{2N \times 2N} \oplus g(t).
\end{aligned}
\end{equation}
Since the first homotopy class of $\mathrm{GL}(L>2,\mathbbm R)$ is stable: $\pi_1(\mathrm{GL}(L>2,\mathbbm R)) = \mathbbm Z_2$.
The homotopy classes of the embedded loops are labeled by the same invariants, i.e., $\xi_{f'} = \xi_f$, $\xi_{g'} = \xi_g$.
Using Lemma \ref{lemma1}, 
\begin{equation}
  \xi_{\tilde{f}\tilde{g}} 
  = \xi_{\tilde{f}*\tilde{g}} 
  = \xi_{\tilde{f}} + \xi_{\tilde{g}}
  = \xi_{f} + \xi_{g}.
\end{equation}
Note that $f\oplus g = \tilde{f}\tilde{g}$, we thus proved the corollary.
\end{proof}

\begin{corollary}
Any loop $f(t)$ on the orthogonal group $\mathrm{SO}(2N)$ is homotopic to a path with block-diagonal form:
\begin{equation}
  f(t) \sim f_{1}(t) \oplus f_{2}(t) \oplus \cdots \oplus f_{N}(t),
\end{equation}
where each $f_{i}(t)$ is a loop on $\mathrm{SO}(2)$, and the homotopy class of $f$ (labelled by a $\mathbb{Z}_{2}$ invariant $\xi$ ) is
\begin{equation}
  \xi = \exp \left(i\pi\sum_{i=1}^{N} W[f_i]\right) \equiv \exp\left(i\pi W_f\right) ,
\end{equation}
where $W[f_i]$ is the winding number for $f_{i}$ on $S O(2)$, and we have defined $W_f$ to be the sum of $W[f_i]$'s, which is only well-defined up to addition of an arbitrary even integer.
\end{corollary}
\begin{proof}
We first proof that any of such loop $f(t)$ can be continuously deformed to the block-diagonal form. 
For any $t$, $f(t)$ can be diagonalized by a unitary matrix:
\begin{equation}
  f(t) = U(t) \mathrm{diag}(\lambda_1,\lambda_1^*,\cdots,\lambda_N,\lambda_N^*) U^\dagger(t),
\end{equation}
Note that the eigenvalues comes in pairs because of the reality condition.
That is, for a eigenvector
\begin{equation}
  f(t) V_i = \lambda_i V_i \quad \Longrightarrow \quad f(t) V_i^* = \lambda_i^* V_i^*.
\end{equation}
We can make a basis transformation in this two-dimensional subspace:
\begin{equation}
  V_{i}^{\pm} \propto V_i \pm V_i^*.
\end{equation}
Under sub choice of basis, we have
\begin{equation}
  f(t) = V(t) \left[f_{1}(t) \oplus f_{2}(t) \oplus \cdots \oplus f_{N}(t)\right] V^T(t),
\end{equation}
where $V(t)$ is real orthogonal matrix and $f_{1}(t) \oplus f_{2}(t) \oplus \cdots \oplus f_{N}(t)$ is in block-diagonal form.
We can then use the homotopy to prove the homotopic relation:
\begin{equation}
  F(t,s) = V(t,s) \left[f_{1}(t) \oplus f_{2}(t) \oplus \cdots \oplus f_{N}(t)\right] V^T(t,s),
\end{equation}
where $V(t,s)$ is a continuous loop from $V(t)$ to identity.
We remark here that (1) we can use the degree of freedom in defining $V_i^\pm$ to make $V(t)\in \mathrm{GL}^+(L,\mathbbm R)$, and thus is loop connected to the identity, and (2) the contraction from $V(t)$ to identity is always possible, since the obstruction to do so implies the nonzero Chern number of the system, which is forbidden by the time-reversal symmetry.

The computation of the topological invariant follows the same strategy as Corollary \ref{cor1}.
The only subtlety is that $\pi_1(\mathrm{SO}(2))$ is $\mathbbm Z$ instead of $\mathbbm Z_2$.
Now consider the embedding
\begin{equation}
  \tilde{f}_i(t) \equiv \mathbbm{I}_{(2i-2)\times(2i-2)}\oplus f_i(t) \oplus \mathbbm{I}_{(2N-2i)\times(2N-2i)}.
\end{equation}
First consider the case $\xi_{f_i} = \pm 1$, we argue that its $\mathbbm Z_2$ invariant should be $1$.
It follows from the fact that any loops is homotopic to the direct sum of $\mathrm{SO}(2)$ loops, and any $\mathrm{SO}(2)$ loop is equivalent to multiple $\xi=\pm 1$ loops.
If such basic element has zero $\mathbbm Z_2$ invariant, the homotopy class of $\mathrm{SO}(L)$ would be trivial.
On the other hand, if the $\mathbbm Z_2$ invariant $\xi_f = 1$, the invariant of any loop can be computed by the above procedure, and the result is the party of the total winding number.
\end{proof}

The homotopy class of the mass discussed above helps us attach a $\mathbbm Z_2$ number to a state.
To do so, we first decomposed the representation into irreducible representations of $S_3$/$S_4$, each irreducible representation determines a $\mathbbm Z_2$ invariant by its own, the total invariant is the sum.
One subtlety, however, may come from the representation for which the mass field must vanished on the north and south poles as the loop (along with the mass should not be singular) can not be constructed.
We argue that such representations containing singularities gives trivial invariant and can be safely removed.
To bypass the massless point, consider the special loop in Fig.~\ref{SM loop singularity}, where we made an infinitesimal detour avoiding the north/south pole.
\begin{figure}[H]
\centering 
\includegraphics[width=0.35\textwidth]{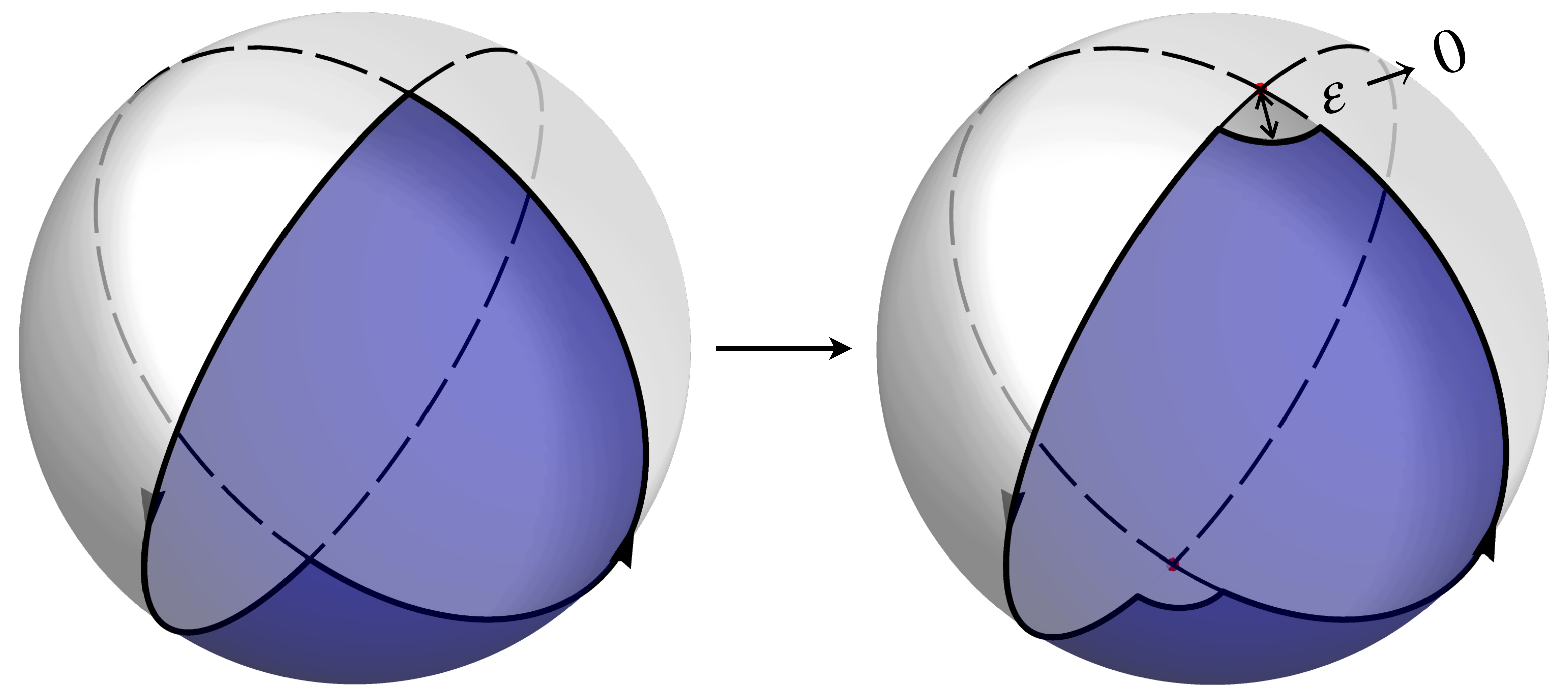} 
\caption{Symmetric loop of $S_4$ for the singular case, where the detour from the north/south pole is a quarter-arc of radius $\varepsilon$. The final result is regarded as the limit $\varepsilon\rightarrow 0$ being taken.}
\label{SM loop singularity}
\end{figure}

Note that the loop is not symmetric any more, and the mass field along two arcs will not be constraint.
The parity of the winding will be arbitrary in this case, and the statement holds in the limit $\varepsilon\rightarrow 0$.
Thus the invariant is regarded to be zero.
On the other hand, for the case where there is not symmetry-enforced singularity, although we can make the same detour, in the limit $\varepsilon\rightarrow 0$ the contribution from two arcs is negligible, and total winding will approach the same result in this limit.

\subsection{Gapless Modes at Fixed Points on the Sphere}\label{SM Gapless Modes at Fixed Points on the Sphere}

Some spatial symmetries in 32 point groups have fixed points $\rr_0$ on the sphere.
For example, the fixed points of $C_n$ are the intersection of the axis of rotation and the sphere.
At the fixed points, the spatial symmetry becomes local symmetry such that there are some Hilbert subspaces that are invariant under a point group symmetry.
Thus, the surface Hamiltonian can be block-diagonalized into subspaces span by orthogonal eigenvectors:
\begin{equation}
[F^\prime(g), H_s(\kk,\rr_0)]=0 \rightarrow  H = H_{e_1}\oplus\cdots\oplus  H_{e_n}
\end{equation}
For each subspace labeled by eigenvalue $e_i$ of $g$, we can classify $ H_{e_i}$.
From the perspective of real-space, the zero-dimensional gapless modes at fixed points are locally identical to gapless edge modes of one-dimensional system in AZ-class (two-dimensional for mirror), according to the presence or absence of $T,P,S$ in each eigenspace \cite{fang2017topological}.

Generally speaking, the non-identity surface flavor transformation matrix $F^\prime(g)$ enforces several pairs of zero-dimensional gapless modes at fixed points.
It can be seen from the fact \cite{PhysRevB.82.115120,PhysRevResearch.2.043300} that the Dirac massive Hamiltonian
\begin{equation}
 H(\kk,\phi)= \sum_i k_{i} \Gamma_{i+2}+m_0[\cos(\mathcal N \phi)\Gamma_1+\sin(\mathcal N \phi)\Gamma_2]
\end{equation}
traps $\mathcal N$ pairs of Majorana zero modes, where $\mathcal N$ is the winding number of mass field $m(\phi)$.
\subsubsection{TCSC with rotation-even pairing potential}

\begin{figure}[H]
\centering 
\includegraphics[width=0.95\textwidth]{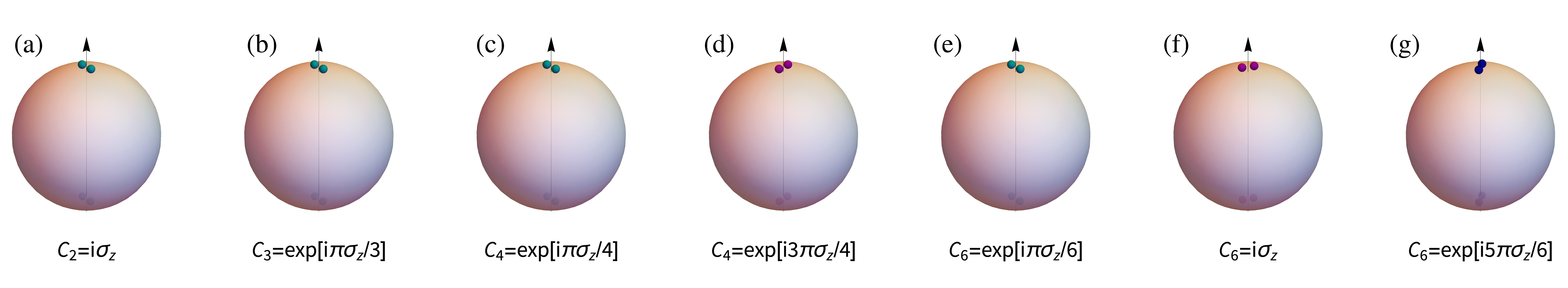} 
\caption{The zero-dimensional gapless modes at fixed points protected by $C_n$ symmetry. $C_n=\exp[im\pi\sigma_z/n]$ is the representation furnished by Majorana zero modes.}
\label{SM Cn even-parity}
\end{figure}

Due to rotation-even pairing potential, the chiral symmetry commute with $n$-fold rotation symmetry
\begin{equation}
[U_s(S),U_s(C_n)]=0
\end{equation}
Thus, there is chiral symmetry in each eigenspace labeled by the eigenvalue of $C_n$.
Each pair of Majorana zero modes can be labeled by eigenvalues of chiral symmetry and rotational symmetry.
Two pairs of Majorana zero modes with the same chiral symmetry eigenvalues cannot annihilate with each other.
The chiral operation representation in the two Majorana basis is $\tau_0$ such that there is no mass term that anticommutes with chiral symmetry to gap out the Majorana zero modes.
It implies that the classification of zero-dimensional gapless modes at fixed points in each eigenspace with complex eigenvalues of spatial symmetry is $\mathbbm Z$.
On the other hand, two pairs of Majorana zero modes with different spatial symmetry eigenvalues also cannot annihilate with each other.
Because the hybridization between two eigenspaces is forbidden by spatial symmetry.
It implies that TSCs whose Majorana zero modes have different angular momentum are topologically inequivalent to each other, and the classification is the direct product of several $\mathbbm Z^{e_i}$:
\begin{equation}
\mathcal C_{\rm{3rd}}(g)=\bigotimes_i \mathbbm Z^{e_i},
\end{equation}
where $e_i$ represents of eigenvalues of $g$.
Because of the presence of time-reversal symmetry and $\{T,S\}=0$, $z^{e_i}$ Majorana zero modes with chiral eigenvalues $+1$ in $e_i$ sub-eigenspace means there are $z^{e_i}$ Majorana zero modes with chiral eigenvalues $-1$ in $e_i$ sub-eigenspace:
\begin{equation}\label{SM complex eigenvalues classification Cn}
z^{e_i}=-z^{e_i^*}.
\end{equation}
Thus, the independent $\mathbbm Z$ reduces to half of the number of complex eigenvalues of $C_n$.

When $e_i\in\mathbbm R$, the invariant $\mathbbm Z^{e_i}$ must be zero due to Eq.(\ref{SM complex eigenvalues classification Cn}).
The case is the eigenspaces of $C_3$ with eigenvalue $e_i=-1$.
But we can define a $\mathbbm Z_2$ invariants because there is time-reversal symmetry in this eigenspace.
As shown in Appendix \ref{SM Bubble Equivalence for Incompatible $n$-fold Rotation and Compatible Three-fold Rotation}, this $\mathbbm Z_2=1$ TCSC is ``spurious'' and can be annihilated under ``bubbling''.
For classification purposes, we also need to determine the value of $\mathbbm Z$ invariants, as we show in Appendix \ref{SM At Fixed Points of $n$-fold Rotation Symmetry}.

Finally, we summarize the classification of gapless modes at fixed points when pairing potential is $n$-fold rotation-even.
\begin{equation}
\begin{array}{l}
\mathcal C_{\rm{3rd}}(C_2)= \mathbbm Z^{\pm\frac{1}{2}}\\
\mathcal C_{\rm{3rd}}(C_3)= \mathbbm Z^{\pm\frac{1}{2}}\\
\mathcal C_{\rm{3rd}}(C_4)= \mathbbm Z^{\pm\frac{1}{2}}\otimes \mathbbm Z^{\pm\frac{3}{2}}\\
\mathcal C_{\rm{3rd}}(C_6)= \mathbbm Z^{\pm\frac{1}{2}}\otimes \mathbbm Z^{\pm\frac{3}{2}}\otimes \mathbbm Z^{\pm\frac{5}{2}}\\
\end{array}
\end{equation}
where $\pm\frac{m}{2}$ represent the time-reversal partner eigenvalues $\exp({i\pm \frac{m}{2} \frac{2\pi}{n}})$ of $C_n$.

\subsubsection{TCSC with rotation-odd pairing potential}\label{SM TCSC with rotation-odd pairing potential}
Let us consider a minimal stacking model of third-order TCSC:
\begin{equation}
 H_s(\kk,\rr)=\gamma_z\tau_0\otimes(\kk\times\nn_\rr)\cdot\ssigma,
\end{equation}
where $\ttau,\ggamma$ and $\ssigma$ are the Pauli matrices on flavor and single pristine-TSC surface degrees of freedom, respectively.
The time-reversal and chiral symmetries are 
\begin{equation}
 T_s=\gamma_0\tau_0\otimes i\sigma_y K, \ \ {S}_s =\gamma_z\tau_0\otimes\nn_\rr\cdot\ssigma.
\end{equation}
There are two pairs of possible $ T_s$-, $ S_s$-symmetric mass terms that commute with each other:
\begin{equation}
\begin{array}{cl}
 M_{1,\rr}=m_{1,\rr}\begin{pmatrix} 0&\tau_0\\\tau_0&0  \end{pmatrix}\otimes \sigma_0=m_{1,\rr}\gamma_x\tau_0\otimes\sigma_0,\ \ \ & M_{2,\rr}=m_{2,\rr}\begin{pmatrix} 0&-i\tau_y\\i\tau_y&0  \end{pmatrix}\otimes \sigma_0=m_{2,\rr}\gamma_y\tau_y\otimes\sigma_0\\
 M_{3,\rr}=m_{3,\rr}\begin{pmatrix} 0&\tau_x\\\tau_x&0  \end{pmatrix}\otimes \sigma_0=m_{3,\rr}\gamma_x\tau_x\otimes\sigma_0,\ \ \ & M_{4,\rr}=m_{4,\rr}\begin{pmatrix} 0&\tau_z\\ \tau_z&0  \end{pmatrix}\otimes \sigma_0=m_{4,\rr}\gamma_x\tau_z\otimes\sigma_0\\
\end{array}
\end{equation}
Because the pairing potential is rotation-odd, the only nontrivial factor of projective representation is
\begin{equation}\label{SM incompatible Cn S commutation relation}
\omega(S,C_n)=-1\rightarrow\{F(S),F(C_n)\}=0.
\end{equation}
We observe that there is at least one mass term $ M_{i,\rr}$ remaining unchanged under $C_n$.
So it implies the classification of incompatible $C_n$ symmetry is $\emptyset$.
For example, a reasonable choice satisfying Eq.(\ref{SM incompatible Cn S commutation relation}) of flavor space representation $F_s(C_n)$ and corresponding transformation matrix $F_s^\prime(C_n)$ is
\begin{equation}
F(C_n) = \begin{pmatrix}0&\tau_0\\\tau_0&0 \end{pmatrix},\ \ F^\prime(C_n) = \begin{pmatrix}0&\tau_0\\\tau_0&0\end{pmatrix}\begin{pmatrix}-\tau_0&0\\0&\det C_n \tau_0 \end{pmatrix}=\begin{pmatrix}0&\tau_0\\-\tau_0&0 \end{pmatrix}.
\end{equation}
The mass field $ M_{2}$ remains unchanged under $n$-fold rotation symmetry $F_s^\prime(C_n)$, and we can add a uniform mass field, e.g.
\begin{equation}
 M_\rr=m_0 \begin{pmatrix} 0&-i\tau_y\\i\tau_y&0  \end{pmatrix}\otimes \sigma_0, m_0\in\mathbbm R,
\end{equation}
to fully gap the sphere.
So the classification of TCSCs protected by incompatible $n$-fold rotational symmetry is
\begin{equation}
\mathcal C_{{\rm 3rd}}(C_n)=\emptyset.
\end{equation}

This can also be obtained by analyzing the presence and absence of $T,P,S$ in each eigenspace.
We choose a natural gauge between $T,P,S$, which is different from that in Ref.\cite{fang2017topological} such that
\begin{equation}
\begin{aligned}
  T:e_m\rightarrow e_m^*\\
  P:e_m\rightarrow -e_m^*\\
  S:e_m\rightarrow -e_m\\
\end{aligned}
\end{equation}
where $e_m=\exp(i m\pi/n),m=1,\cdots,2n-1$ are eigenvalues of $C_n$.

For two-fold rotation symmetry, the surface Hilbert space at fixed points can be block-diagonalized into subspaces labeled by $C_2$ eigenvalues:
\begin{equation}
 H_s(\kk,\rr_0)= H_{+i}\oplus  H_{-i}.
\end{equation}
$ H_{+i}$ and $ H_{-i}$ are mapped to each other under $ T$ and $ S$.
So $ H_{+i}\oplus H_{-i}$ belongs to one-dimensional class DIII, of which invariant is $\mathbbm Z_2$.
As shown in Appendix \ref{SM Bubble Equivalence for Incompatible $n$-fold Rotation and Compatible Three-fold Rotation}, this $\mathbbm Z_2=1$ TCSC is ``spurious'' and can be annihilated under ``bubbling''.
Hence, the classification of TCSC protected by incompatible $C_2$ is $\emptyset$.
We note that even if we use 1-D anomalous Hamiltonian classification $\mathbbm Z_2$, the final classification result remains unchanged. Because in the allowed surface representation, this bubble does not appear and is automatically excluded by our method.
The same situation also occurs in the following bubble example ($C_{4,6}$), so we will not repeat it.
\vspace{0.3cm}

For three-fold rotation symmetry, the pairing potential must be rotation-even otherwise TSC breaks $C_3$.
\vspace{0.3cm}

For four-fold rotation symmetry, the surface Hilbert space at fixed points can be block-diagonalized into subspaces labeled by $C_4$ eigenvalues:
\begin{equation}
 H_s(\kk,\rr_0)= H_{-\frac{3}{2}}\oplus  H_{-\frac{1}{2}}\oplus  H_{\frac{1}{2}}\oplus  H_{\frac{3}{2}}.
\end{equation}
$ H_{\frac{1}{2}}$ and $ H_{\frac{3}{2}}$ are mapped to each other under chiral symmetry $ S$.
So $ H_{\frac{1}{2}}\oplus H_{\frac{3}{2}}$ belongs to one-dimensional class AIII, of which invariant is $\mathbbm Z$.
But this invariant must be zero.
Consider a contradiction: if the invariant $z$ is nonzero, there are $z$ Majorana zero modes at fixed points, and the representation of the chiral symmetry is identity matrix.
The $C_4$ rotation symmetry should anticommute with chiral symmetry, but none of the matrices anticommute with the identity matrix.
So the invariant must be zero, and the classification reduces to $\emptyset$.
Based on similar arguments, the $\mathbbm Z$ invariant of $ H_{-\frac{1}{2}}\oplus H_{-\frac{3}{2}}$ must vanish.
Hence, the classification of TCSC protected by incompatible $C_4$ is $\emptyset$.

\vspace{0.3cm}

For six-fold rotation symmetry, the surface Hilbert space at fixed points can be block-diagonalized into subspaces labeled by $C_6$ eigenvalues:
\begin{equation}
 H_s(\kk,\rr_0)= H_{-\frac{5}{2}}\oplus \ H_{-\frac{3}{2}}\oplus H_{-\frac{1}{2}}\oplus H_{\frac{1}{2}}\oplus  H_{\frac{3}{2}}\oplus  H_{\frac{5}{2}}.
\end{equation}
$ H_{\frac{1}{2}}$ and $ H_{\frac{5}{2}}$ are mapped to each other under $ S$.
So $ H_{\frac{1}{2}}\oplus H_{\frac{5}{2}}$ belongs to one-dimensional class AIII, of which invariant is $\mathbbm Z$.
Based on similar arguments in the $C_4$ case, the invariant must be zero, and the classification reduces to $\emptyset$.
The same is true for $ H_{-\frac{1}{2}}\oplus H_{-\frac{5}{2}}$.
On the other hand, $ H_{\frac{3}{2}}$ and $ H_{-\frac{3}{2}}$ are mapped to each other under $ T$ and $ S$.
So $ H_{\frac{3}{2}}$ and $ H_{-\frac{3}{2}}$ belong to one-dimensional class DIII, of which invariant is $\mathbbm Z_2$.
As shown in Appendix~\ref{SM Bubble Equivalence for Incompatible $n$-fold Rotation and Compatible Three-fold Rotation}, this $\mathbbm Z_2=1$ TCSC is ``spurious'' and can be annihilated under ``bubbling''.
Finally, the classification of TCSCs protected by incompatible $C_6$ is $\emptyset$.

\subsubsection{TCSC with Mirror-even pairing potential}
Because the pairing potential is Mirror-even, the only nontrivial factor of projective representation is
\begin{equation}\label{SM incompatible Cn S commutation relation}
\omega(S,M)=-1\rightarrow\{F(S),F(M)\}=0.
\end{equation}
We observe that there is at least one mass term $ M_{i,\rr}$ remaining unchanged under $M$.
So it implies the classification of incompatible $M$ symmetry is $\emptyset$.
For example, a reasonable choice satisfying Eq.(\ref{SM incompatible Cn S commutation relation}) of flavor space representation $F_s(M)$ and corresponding transformation matrix $F_s^\prime(M)$ is
\begin{equation}
F(M) = \begin{pmatrix}0&\tau_0\\\tau_0&0 \end{pmatrix},\ \ F^\prime(M) = \begin{pmatrix}0&\tau_0\\\tau_0&0\end{pmatrix}\begin{pmatrix}-\tau_0&0\\0&\det M \tau_0 \end{pmatrix}=\begin{pmatrix}0&-\tau_0\\-\tau_0&0 \end{pmatrix}.
\end{equation}
The mass field $ M_{1,3,4}$ remains unchanged under Mirror symmetry $F_s^\prime(M)$, and we can add a uniform mass field, e.g.
\begin{equation}
 M_\rr=\left[m_0 \begin{pmatrix} 0&\tau_z\\\tau_0&z  \end{pmatrix}+m_1 \begin{pmatrix} 0&\tau_x\\\tau_x&0  \end{pmatrix}\right]\otimes \sigma_0, m_0,m_1\in\mathbbm R,
\end{equation}
to fully gap the sphere.
So the classification of TCSCs protected by incompatible Mirror symmetry is
\begin{equation}
\mathcal C_{{\rm 3rd}}(M)=\emptyset.
\end{equation}
This can also be obtained by analyzing the presence and absence of $T,P,S$ in each eigenspace.
We choose a natural gauge between $T,P,S$ such that
\begin{equation}
\begin{aligned}
  T:e_m\rightarrow e_m^*\\
  P:e_m\rightarrow e_m^*\\
  S:e_m\rightarrow e_m\\
\end{aligned}
\end{equation}
where $e_m=\pm i$ are eigenvalues of $M$.
$ H_{+i}$ is invariant under chiral symmetry $ S$.
Therefore, $ H_{+i}$ belongs to the two-dimensional AIII class, which is trivial.
Based on similar arguments, $ H_{-i}$ is trivial.
The question that remains to be answered is can $ H_{+i}\oplus  H_{-i}$, i.e., full Hamiltonian, be $\mathbbm Z_2$ nontrivial?
We argue that it is impossible because we cannot write down an effective surface theory on the basis of Majorana edge modes:
\begin{equation}
 H(k)=k\sigma_z,\ \ T=i\sigma_yK,\ \ P=K,\ \ S=i\sigma_y.
\end{equation}
There is not an allowed representation $U(M)$ of mirror satisfying
\begin{equation}
U(M)  H(k)U(M)=  H(k),\ \ TU(M)T^{-1}=U(M),\ \ PU(M)P^{-1}=U(M),\ \ SU(M)S^{-1}=U(M).
\end{equation}
Finally, the classification of TCSCs protected by incompatible mirror is $\emptyset$.

\subsubsection{TCSC with Mirror-odd pairing potential}
Because the pairing potential is mirror-odd, the factor system of projective representation is
\begin{equation}
\omega(g_i,g_j)=1,\ \ \forall g_i,g_j\in G.
\end{equation}
That is to say, all symmetry representation matrices commute with other symmetry representation matrices.
A reasonable choice of flavor space representation $F(M)$ and corresponding transformation matrix $F^\prime(M)$ is
\begin{equation}
F(M) = \begin{pmatrix}\tau_0&0\\0&\tau_0 \end{pmatrix},\ \ F^\prime(M) = \begin{pmatrix}\tau_0&0\\0&\tau_0 \end{pmatrix}\begin{pmatrix}\chi_M\tau_0&0\\0&\tau_0 \end{pmatrix}=\begin{pmatrix}-\tau_0&0\\0&\tau_0 \end{pmatrix}.
\end{equation}
The mass field $m_{i\rr}$ changes sign under $M$ operation due to
\begin{equation}
\{F^\prime(M), M_{i}\}=0\Rightarrow m_{i,\rr}=-m_{i,M\rr}=-m_{i,\rr}.
\end{equation}
where $M_i$ is the matrix part of $ M_{i,\rr}$ except $m_{i,\rr}$, It requires that the mirror-invariant line is gapless.
Two pairs of mirror-protected helical edge modes cannot annihilate with each other.
Because in each mirror eigenspace, there are particle-hole symmetry $P$ such that $\mathbbm Z$ invariant of D class is well-defined.
Finally, the classification of TCSCs protected by compatible mirror symmetry is $\mathbbm Z$.

\section{Algebraic Framework of Classification}
In this Appendix, we establish the algebraic framework of classification starting from the pairing symmetry and point-group symmetry of TCSC, as shown in Fig.~\ref{SM_flow_chart}.

\subsection{Flow Chart For Classification Algorithm}
\begin{figure}[t]
\centering 
\includegraphics[width=1\textwidth]{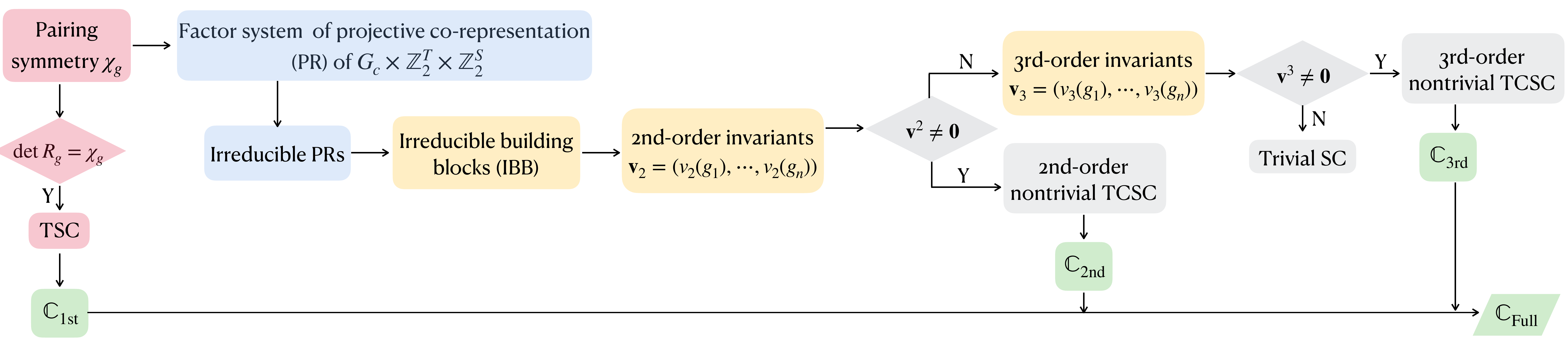}
\caption{Flow chart for our classification algorithm.
For a TCSC with given pairing symmetry $\chi_g$ and point group, we first determine the factor system of projective co-representations of $G_c\times\mathbbm Z_2^T\times\mathbbm Z_2^S$ and then give all inequivalent irreducible projective co-representations by central extension of $G_c\times\mathbbm Z_2^S$ with kernel $\mathbbm Z_2$.
All irreducible building blocks are the direct sum of irreducible projective co-representations.
We calculate the second-order invariants ${\bf{v}}^2$ protected by all point-group symmetries of irreducible building blocks to see whether it is not a zero vector.
If yes, it is a second-order TCSC and constitutes an element in $\mathcal C_{\rm{2nd}}$.
If not, it is either a third-order topological superconductor or a trivial superconductor.
To see whether it is a third-order state, we calculate the third-order invariants ${\bf{v}}^3$ protected by all point-group symmetries of irreducible building blocks to see whether it is not a zero vector.
If yes, it is a third-order TCSC and constitutes an element in $\mathcal C_{\rm{3rd}}$.
If not, it is a trivial superconductor.
Finally, we consider the group-extension problems of classification $\mathcal C_{i\rm{th}}$ for each order to obtain full classification $\mathcal C_{\rm{full}}$.
}
\label{SM_flow_chart}
\end{figure}

\paragraph{Red Region ---}In this part of the $\hyperref[SM_flow_chart]{\mathrm{Flow\ Chart}}$, we determine whether the first order invariant is zero or not by pairing symmetry $\chi_g$.
The point-group symmetries impose constraints on first-order invariants, as shown in Eq.(\ref{SM chiarl constraint}):
\begin{equation}
z_w=\chi_g\cdot \det R_g z_w
\end{equation}
So if $\chi_g=\det R_g,\ \forall g \in G_c$, the first-order invariants can take any integer, so the first-order classification is $\mathbbm Z$.
If there are some point-group symmetries satisfying $\chi_g=-\det R_g$, the first-order topology must be trivial, so the first-order classification reduces to $\emptyset$.

\paragraph{Blue Region ---}For this part of the $\hyperref[SM_flow_chart]{\mathrm{Flow\ Chart}}$, we give a general recipe \cite{bradley1968magnetic,bradley2010mathematical} for obtaining the projective co-representations belonging to a given factor system.
Recall that the factor system of (surface flavor) projective co-representation in the case where the pairing symmetry is $\chi_g$:
\begin{equation}\label{SM factor system in appendix Algebraic Framework of Classification}
\omega(S,g) = \det R_g\chi_g, \ \ \omega(g,S)=\omega(g,T)=\omega(T,g)=\omega(S,T)=\omega(T,S)=1,\ \ \forall g\in G_c,
\end{equation}
and in surface flavor space the representation of time-reversal symmetry is $T=K$.

Here, the factor system $\omega(g_i,g_j),\forall g_i,g_j\in G_c\times\mathbbm Z_2^T\times\mathbbm Z_2^S$ is either $+1$ or $-1$ due to time-reversal symmetry.
We define the $a(g_i,g_j)$ by
\begin{equation}\label{SM projective representation calculate start}
\omega(g_i,g_j)=\exp (i\pi a(g_i,g_j)),\ \ a(g_i,g_j)=0,1.
\end{equation}
Also, since $\omega(g_i,g_I^{-1})=\omega(g_i^{-1},g_I)$ and the constraint imposed by associativity of matrix multiplication:
\begin{equation}
\omega(g_i,g_jg_l)\omega(g_j,g_l)=\omega(g_ig_j,g_l)\omega(g_i,g_j),
\end{equation}
we have
\begin{equation}
a(g_i,g_I^{-1})=a(g_i^{-1},g_I), \ \ a(g_i,g_jg_l)a(g_j,g_l)=a(g_ig_j,g_l)a(g_i,g_j) \mod2 
\end{equation}
Let $Z_2$ be the cyclic group of integers $0,1$ with group product defined as addition modulo g.

One can construct $G_S^d$ consisting of the $4|G_c|$ elements that are pairs of element $[g_i,\alpha]$, one from $G_c\times\mathbbm Z_2^S$ and one from $Z_2$.\
Define a multiplication rule by the equation
\begin{equation}
[g_j,\alpha][g_j,\beta]=[g_ig_j,\alpha+\beta+a(g_i,g_j)\mod 2]
\end{equation}
With respect to this multiplication rule, $G_S^d$ is a group.
From the multiplication rule and Eq.(\ref{SM factor system in appendix Algebraic Framework of Classification}), it follows that
\begin{equation}
[E,\alpha][g_i,\beta]=[g_i,\beta][E,\alpha]=[g_i,\alpha+\beta],\ \ \forall g_i\in G_c\times \mathbbm Z_2^S.
\end{equation}
It means that the subgroup $\{[E,\alpha]\}$ commute with all elements in $G_S^d$ and lies in the centre of $G_S^d$.
From Schur's lemma, it follows that if $F$ is a linear representation of $G_S^d$ then $F([E,\alpha])$ is a scalar multiple of the identity matrix.
Suppose now that there exists a linear representation whose matrix forms of $(E,\alpha)$ are
\begin{equation}
F^d([E,\alpha])=\exp(i\pi\alpha),\ \ \alpha=0,1.
\end{equation}
Then the matrix form of $(g_i,\alpha)$ is
\begin{equation}\label{SM central extension group representation}
F^d([g_i,\alpha])=F^d([g_i,0])F^d([E,\alpha])=F^d([g_i,0])\exp(i\pi\alpha).
\end{equation}
Writing $F^d([g_i,0])=F(g_i),\forall g_i\in G_c\times\mathbbm Z_2^S$, we have
\begin{equation}
\begin{aligned}
F(g_i)F(g_j)&=F^d([g_i,\alpha])F^d([g_j,\beta])\exp(-i\pi(\alpha+\beta))\\
&=F^d([g_ig_j,\alpha+\beta+a(g_i,g_j)])\exp(-i\pi(\alpha+\beta))\\
&=F(g_ig_j)\exp(-i\pi a(g_i,g_j))
\end{aligned}
\end{equation}
It can be checked that $F(g_i)$ is an irreducible projective representation of $G_c\times \mathbbm Z_2^S$.
Conversely, all irreducible projective representations can be found using this procedure:
If $Q$ is an irreducible projective representation of $G_c\times \mathbbm Z_2^S$ and we set
\begin{equation}
F^d([g_i,\alpha])=Q(g_i)\exp(i\pi\alpha),\ \ \forall g_i \in G_c\times \mathbbm Z_2^S,
\end{equation}
$\bm{F}^d$ is a linear representation of $G^d_S$.
In the language of the mathematician, all the irreducible projective representations of $G$ can be lifted into the central extension of $G$ with kernel $M$, which is the number of classes of factor systems.

Finally, the relation between the dimension of projective representation and group order is
\begin{equation}
|G|=\sum_i d_i^2,
\end{equation}
where $d_i$ is $i$-th irreducible projective representation with given factor system and $i$ runs over all irreducible representation.
However, the number of irreducible projective representations is not equal to the number of classes of $G$.

Next, if the group is $G_c\times \mathbbm Z_2^S\times \mathbbm Z_2^T$, the representation $\bm{F}^d$ in Eq.(\ref{SM central extension group representation}) becomes co-representation $\bm{F}^r$:
\begin{equation}
    F^r([u,\alpha])=\begin{pmatrix}
      F^d(u)&0\\
      0&F^d([T^{-1}uT,\alpha])^*
    \end{pmatrix},\ \ 
    F^r([A,\alpha])=\begin{pmatrix}
        0&F^d([AT,\alpha])\\
        F^d([T^{-1}A,\alpha])^*&0\\
    \end{pmatrix},
\end{equation}
where $u$ and $A$ are unitary operations and antiunitary operations in $G_c\times \mathbbm Z_2^S\times \mathbbm Z_2^T$, respectively.
Note that $\bm{F}^r$ may be reducible or irreducible, depending on the value of
\begin{equation}\label{SM projective representation calculate end}
\sum_{A,\alpha}\chi([A,\alpha][A,\alpha]).
\end{equation}

\paragraph{Yellow Region}
In this part of the $\hyperref[SM_flow_chart]{\mathrm{Flow\ Chart}}$, we use the irreducible projective co-representations to construct $n$th-order irreducible building blocks and calculate the corresponding $n$th-order invariants.
A $n$-th order irreducible building block cannot be decomposed into several smaller symmetry-invariant subsectors with zero $(n-1)$-th order invariants.
If there are incompatible point group symmetries, then the representation is a projective representation, and if there are no incompatible point group symmetries, then the factor system of the representation is trivial, that is, a linear representation.

\textit{Compatible Pairing Symmetry --} For the pairing symmetry in which all symmetries are compatible $(\chi_g=\det R_g)$, the representation of chiral symmetry is $\mathbbm{1}_{L}$, where $L$ is the dimension of linear representation.
All possible dimensions of irreducible linear representations are $1,2,3$.
The possible $2$-th order irreducible building blocks $\bb_{2,i}$ can be obtain from $1$-th order irreducible building blocks $\bb_{1,i}$, i.e., the irreducible representations:
\begin{equation}\label{SM possible construction of irreducible building blocks}
\begin{array}{rl}
                        (1d\ \ {\rm{ rep}})&\oplus\ \ \ (1d\ \ {\rm{ rep}})\\
(1d\ \ {\rm{ rep}}\oplus 1d\ \ {\rm{ rep}})&\oplus\ \ \ ( 2d\ \ {\rm{ rep}})\\
                        (2d\ \ {\rm{ rep}})&\oplus\ \ \  (2d\ \ {\rm{ rep}})\\
(1d\ \ {\rm{ rep}}\oplus 1d\ \ {\rm{ rep}}\oplus 1d\ \ {\rm{ rep}})&\oplus\ \ \ ( 3d\ \ {\rm{ rep}})\\
(1d\ \ {\rm{ rep}}\oplus 2d\ \ {\rm{ rep}})&\oplus\ \ \ ( 3d\ \ {\rm{ rep}})\\
                        (3d\ \ {\rm{ rep}})&\oplus \ \ \ (3d\ \ {\rm{ rep}})\\
(1d\ \ {\rm{ rep}}\oplus 3d\ \ {\rm{ rep}})&\oplus\ \ \ ( 2d\ \ {\rm{ rep}}\oplus 2d\ \ {\rm{ rep}})\\
(3d\ \ {\rm{ rep}}\oplus 3d\ \ {\rm{ rep}})&\oplus\ \ \ ( 2d\ \ {\rm{ rep}}\oplus2d\ \ {\rm{ rep}}\oplus 2d\ \ {\rm{ rep}})\\
\end{array}
\end{equation}
The chiral symmetry is implemented as $\mathbbm{1}\oplus-\mathbbm{1}$, where $\mathbbm{1}$ is the identity matrix.
For each irreducible building block with given representation of point-group symmetry $g\in G_c$, we can calculate the topological invariants ${\bf{v}}_{2/3}$ following the procedures in Appendix.~\ref{Classification of Gapless Modes on the Sphere and Homotopy Group Theory}.
Next, we show that all possible constructions of irreducible building blocks are exhausted by the construction described in Eq.(\ref{SM possible construction of irreducible building blocks}):
For an irreducible block, the number $d_i$ appearing in $\mathbbm{1}$ chiral sector must be all different from that in $-\mathbbm{1}$ chiral sector except that in both sectors it is an irreducible representation of the same dimension of the same dimension, i.e., $1\oplus1$, $2\oplus2$, $3\oplus3$.

(i) If the chiral symmetry is $\mathbbm{1}_{3N}\oplus-\mathbbm{1}_{3N}$, then in $-\mathbbm{1}$ chiral sector $d_i$ can only be $1$ or $2$.
If both $1$ and $2$ appear, then $3\oplus 1+2$ already appears in Eq.(\ref{SM possible construction of irreducible building blocks}).
If only $1$ appears, then $3\oplus 1+1+1$ already appears in Eq.(\ref{SM possible construction of irreducible building blocks}).
If only $2$ appears, then $3+3\oplus 2+2+2$ already appears in Eq.(\ref{SM possible construction of irreducible building blocks}).

(ii) If the chiral symmetry is $\mathbbm{1}_{3N+1}\oplus-\mathbbm{1}_{3N+1}$.
If only $1$ appears in $-\mathbbm{1}$ chiral sector, then $1\oplus1$ already appears in Eq.(\ref{SM possible construction of irreducible building blocks}).
If only $2$ appears, then $3+1\oplus2+2$ already appears in Eq.(\ref{SM possible construction of irreducible building blocks}).
If both $1$ and $2$ appear, then there is not possible irreducible building blocks due to the fact that the highest dimension of the representation where the chiral representation is $\mathbbm 1$ is 3.

(iii) If the chiral symmetry is $\mathbbm{1}_{3N+2}\oplus-\mathbbm{1}_{3N+2}$.
If only $1$ appears in $-\mathbbm{1}$ chiral sector, then $2\oplus1+1$ already appears in Eq.(\ref{SM possible construction of irreducible building blocks}).
If only $2$ appears, then $2\oplus2$ already appears in Eq.(\ref{SM possible construction of irreducible building blocks}).

Similar logic can be generalized to get the basis $\bb_{3,i}$ of $V_3$ from basis $\bb_{2,i}$ of $V_2$.

\textit{Incompatible Pairing Symmetry --} For the pairing symmetry in which there exists incompatible symmetries $(\chi_g=-\det R_g)$, the representation of chiral symmetry is automatically $\tau_z$.
So the possible 2nd-order irreducible building blocks $\bb_{2,i}$ are
\begin{equation}
\begin{aligned}
2d\ \mathrm{rep}\\
4d\ \mathrm{rep}\\
6d\ \mathrm{rep}
\end{aligned}
\end{equation}
All possible 3rd-order irreducible building blocks $\bb_{3,i}$ are the part of 
\begin{equation}\label{SM possible construction of irreducible building blocks}
\begin{array}{rl}
2d\ & \mathrm{rep}\\
4d\ & \mathrm{rep}\\
6d\ & \mathrm{rep}\\
(2d\ \ {\rm{ rep}})&\oplus\ \ \ (2d\ \ {\rm{ rep}})\\
(2d\ \ {\rm{ rep}})&\oplus\ \ \ (4d\ \ {\rm{ rep}})\\
(2d\ \ {\rm{ rep}})&\oplus\ \ \ (6d\ \ {\rm{ rep}})\\
(4d\ \ {\rm{ rep}})&\oplus\ \ \ (4d\ \ {\rm{ rep}})\\
(4d\ \ {\rm{ rep}})&\oplus\ \ \ (6d\ \ {\rm{ rep}})\\
(6d\ \ {\rm{ rep}})&\oplus\ \ \ (6d\ \ {\rm{ rep}})\\
\end{array}
\end{equation}
whose second-order invariants are zero.

\paragraph{Green Region}
The classification group is a finitely generated abelian group.
In this part of the flow chart, we get the group structure of classification $\mathcal C_{\rm{2nd/3rd}}$ for each order from topological invariants ${\bf{v}}^{2/3}$.
We first stack the row vectors ${{\bf{v}}^{2/3}}$ into a matrix, of which elements may be defined on different rings $\mathbbm{Z}/\mathbbm{Z}_2$.
Next, we calculate its row echelon form using integer version of Gaussian elimination.
Each nonzero row vectors of the reduced matrix is a generator of the classification group, as shown in Appendix.~\ref{Finding a finite set of generators for an Abelian group}.

\paragraph{Example} Let us use a TCSCs with $C_4$-rotational odd pairing potential for an explicit example.
Due to the pairing potential being odd under action $C_4$, the first-order invariant $z_w$ must be zero:
\begin{equation}
z_w=\det C_4\chi_{C_4} z_w=(+1) \times(-1)z_w=-z_w\rightarrow z_w=0\rightarrow \mathcal C_{\rm{1st}}\cong\emptyset.
\end{equation}
Next, the factor system of the surface projective representation is 
\begin{equation}
\omega(S,C_4)=-1,\ \ \omega(C_4,S)=\omega(C_4,T)=\omega(T,C_4)=\omega(S,T)=\omega(T,S)=\omega(g,T)=\omega(T,g)=1.
\end{equation}
According to Eqs.(\ref{SM projective representation calculate start}-\ref{SM projective representation calculate end}), we get all inequivalent irreducible co-representations of $C_4\times Z_2^T\times Z_2^S$:
\begin{equation}
\begin{aligned}
F^{(1)}(C_4)=\begin{pmatrix}0&1\\1&0 \end{pmatrix},\ 
F^{(1)}(C_2)=\begin{pmatrix}1&0\\0&1 \end{pmatrix},\ 
F^{(1)}(S)=\begin{pmatrix}1&0\\0&-1 \end{pmatrix},\ 
F^{(1)}(T)=K.\ \\
F^{(2)}(C_4)=\begin{pmatrix}0&1\\-1&0 \end{pmatrix},\ 
F^{(2)}(C_2)=\begin{pmatrix}-1&0\\0&-1 \end{pmatrix},\ 
F^{(2)}(S)=\begin{pmatrix}1&0\\0&-1 \end{pmatrix},\ 
F^{(2)}(T)=K.\ \\
\end{aligned}
\end{equation}
As shown in Eq.(\ref{SM the transformation matrix of point-group symmetry}), the symmetry constraints acting on mass field are:
\begin{equation}
\begin{aligned}
F^{\prime(1)}(C_4)=\begin{pmatrix}0&1\\-1&0 \end{pmatrix},\ 
F^{\prime(1)}(C_2)=\begin{pmatrix}1&0\\0&1 \end{pmatrix},\ 
F^{\prime(1)}(S)=\begin{pmatrix}1&0\\0&-1 \end{pmatrix},\ 
F^{\prime(1)}(T)=K.\ \\
F^{\prime(2)}(C_4)=\begin{pmatrix}0&1\\1&0 \end{pmatrix},\ 
F^{\prime(2)}(C_2)=\begin{pmatrix}-1&0\\0&-1 \end{pmatrix},\ 
F^{\prime(2)}(S)=\begin{pmatrix}1&0\\0&-1 \end{pmatrix},\ 
F^{\prime(2)}(T)=K.\ \\
\end{aligned}
\end{equation}
The $2$nd-order irreducible building block of which Hamiltonian is $H=h_+\oplus h_-$ with representation $F^{(1)}(C_4)$ has nonzero 2nd-order invariants:
\begin{equation}
{\bf{u}}=(v_2({C_4}),v_2({C_2}),v_2({C_4^3}))=(1,0,1)\mod2.
\end{equation}
The $2$nd-order irreducible building block of which Hamiltonian is $H=h_+\oplus h_-$ with representation $F^{(2)}(C_4)$ has zero 2nd-order invariants:
\begin{equation}
{\bf{v}}=(v_2({C_4}),v_2({C_2}),v_2({C_4^3}))=(0,0,0)\mod  2.
\end{equation}
${\bf{u}}$ and ${\bf{v}}$ generate an abelian group $\mathcal C_{\rm 2nd}$.
Thus, the classification of second-order TCSCs with $B$-pairing symmetry protected by $C_4$ is $\mathbbm Z_2$.

All $3$rd-order irreducible building blocks of which Hamiltonian is $H=h_+\oplus h_-\oplus h_+\oplus h_-$ with representations
\begin{equation}
F^{(1)}\oplus F^{(1)},\ \ F^{(2)}\oplus F^{(2)}
\end{equation}
have zero third-order invariants.
So, the third-order classifications $\mathcal C_{\rm{3rd}}=\emptyset$.
Finally, the classification of TCSCs protected by incompatible $C_4$ is $\mathcal C_{\rm{full}}=C_{\rm{2nd}}=\mathbbm Z_2$.

\subsection{The Classification Table of TCSCs Protected by Crystallographic Point groups}
\begin{table}[H]
\caption{The classification table of TCSCs for each point group and all one-dimensional pairing symmetries.
Here, the Irrep is the pairing symmetry.
$\mathcal{C}_{\mathrm{1st}}$, $\mathcal{C}_{\mathrm{2nd}}$, and $\mathcal{C}_{\mathrm{3rd}}$ are the classification group of first, second, third order TCSC, respectively.
$\mathcal{C}_{\mathrm{full}}$ is the full classification of TCSC after considering group-extension problems.
2nd/3rd Generators are the point-group elements under which the mass field of corresponding classification group's generators (root state) change sign.
For example, in the $C_{2h}$ and $A_u$ irrep case, one of the 2nd Generators
is $[M_z,I]$, which means one of the generators of $\mathcal C_{\rm{2nd}}$ is a TCSC (root state) whose mass field changes sign under $M_z$ and inversion symmetries.
}
\centering
\vspace{0.6cm}
\begin{tabular}{|p{1.5cm}<{\centering}|p{1.5cm}<{\centering}|p{1.5cm}<{\centering}|p{1.5cm}<{\centering}|p{1.5cm}<{\centering}|p{1.5cm}<{\centering}|p{1.5cm}<{\centering}|p{1.5cm}<{\centering}|}
\hline
$\hyperref[Group C1]{C_{1}}$&
$\hyperref[Group Ci]{C_{i}}$&
$\hyperref[Group C2]{C_{2}}$&
$\hyperref[Group Cs]{C_{s}}$&
$\hyperref[Group C2h]{C_{2h}}$&
$\hyperref[Group D2]{D_{2}}$&
$\hyperref[Group C2v]{C_{2v}}$&
$\hyperref[Group D2h]{D_{2h}}$\\\hline
$\hyperref[Group C4]{C_{4}}$&
$\hyperref[Group S4]{S_{4}}$&
$\hyperref[Group C4h]{C_{4h}}$&
$\hyperref[Group D4]{D_{4}}$&
$\hyperref[Group C4v]{C_{4v}}$&
$\hyperref[Group D2d]{D_{2d}}$&
$\hyperref[Group D4h]{D_{4h}}$&
$\hyperref[Group C3]{C_{3}}$\\\hline
$\hyperref[Group C3i]{C_{3i}}$&
$\hyperref[Group D3]{D_{3}}$&
$\hyperref[Group C3v]{C_{3v}}$&
$\hyperref[Group D3d]{D_{3d}}$&
$\hyperref[Group C6]{C_{6}}$&
$\hyperref[Group C3h]{C_{3h}}$&
$\hyperref[Group C6h]{C_{6h}}$&
$\hyperref[Group D6]{D_{6}}$\\\hline
$\hyperref[Group C6v]{C_{6v}}$&
$\hyperref[Group D3h]{D_{3h}}$&
$\hyperref[Group D6h]{D_{6h}}$&
$\hyperref[Group T]{T_{}}$&
$\hyperref[Group Th]{T_{h}}$&
$\hyperref[Group O]{O_{}}$&
$\hyperref[Group Td]{T_{d}}$&
$\hyperref[Group Oh]{O_{h}}$\\\hline
\end{tabular}
\end{table}


\begin{table}[H]
\caption{Group $C_1$}
\label{Group C1}
\centering
\begin{tabular}{c|cccc|c|c}
\hline \hline
 Irrep & $\mathcal{C}_{\mathrm{1st}}$ & $\mathcal{C}_{\mathrm{2nd}}$ & $\mathcal{C}_{\mathrm{3rd}}$ & $\mathcal{C}_{\mathrm{full}}$ & 2nd Generators & 3rd Generators\\    
\hline
$A_{}$ & $\emptyset$ & $\emptyset$ & $\emptyset$ & $\emptyset$ & $ \begin{aligned} \emptyset\end{aligned}$ & $ \begin{aligned} \emptyset\end{aligned}$ \\
\hline
\hline
\end{tabular}
\end{table}
\begin{table}[H]
\caption{Group $C_{i}$}
\label{Group Ci}
\centering
\begin{tabular}{c|cccc|c|c}
\hline \hline
 Irrep & $\mathcal{C}_{\mathrm{1st}}$ & $\mathcal{C}_{\mathrm{2nd}}$ & $\mathcal{C}_{\mathrm{3rd}}$ & $\mathcal{C}_{\mathrm{full}}$ & 2nd Generators & 3rd Generators\\    
\hline
$A_{u}$ & $\mathbbm{Z}$ & $\mathbbm{Z}_2$ & $\mathbbm{Z}_2$ & $\mathbbm{Z}\otimes\mathbbm{Z}_4$ & $ \begin{aligned} &[I] \\\end{aligned}$ & $ \begin{aligned} &[I] \\\end{aligned}$ \\
\hline
$A_{g}$ & $\emptyset$ & $\emptyset$ & $\emptyset$ & $\emptyset$ & $ \begin{aligned} \emptyset\end{aligned}$ & $ \begin{aligned} \emptyset\end{aligned}$ \\
\hline
\hline
\end{tabular}
\end{table}
\begin{table}[H]
\caption{Group $C_{2}$}
\label{Group C2}
\centering
\begin{tabular}{c|cccc|c|c}
\hline \hline
 Irrep & $\mathcal{C}_{\mathrm{1st}}$ & $\mathcal{C}_{\mathrm{2nd}}$ & $\mathcal{C}_{\mathrm{3rd}}$ & $\mathcal{C}_{\mathrm{full}}$ & 2nd Generators & 3rd Generators\\    
\hline
$A_{}$ & $\mathbbm{Z}$ & $\mathbbm{Z}_2$ & $\mathbbm{Z}$ & $\mathbbm{Z}^2$ & $ \begin{aligned} &[C_{2z}] \\\end{aligned}$ & $ \begin{aligned} &[C_{2y}] \\\end{aligned}$ \\
\hline
$B_{}$ & $\emptyset$ & $\mathbbm{Z}_2$ & $\emptyset$ & $\mathbbm{Z}_2$ & $ \begin{aligned} &[C_{2z}] \\\end{aligned}$ & $ \begin{aligned} \emptyset\end{aligned}$ \\
\hline
\hline
\end{tabular}
\end{table}
\begin{table}[H]
\caption{Group $C_{s}$}
\label{Group Cs}
\centering
\begin{tabular}{c|cccc|c|c}
\hline \hline
 Irrep & $\mathcal{C}_{\mathrm{1st}}$ & $\mathcal{C}_{\mathrm{2nd}}$ & $\mathcal{C}_{\mathrm{3rd}}$ & $\mathcal{C}_{\mathrm{full}}$ & 2nd Generators & 3rd Generators\\    
\hline
$A_{''}$ & $\mathbbm{Z}$ & $\mathbbm{Z}$ & $\emptyset$ & $\mathbbm{Z}^2$ & $ \begin{aligned} &[M_{z} ] \\\end{aligned}$ & $ \begin{aligned} \emptyset\end{aligned}$ \\
\hline
$A_{'}$ & $\emptyset$ & $\emptyset$ & $\emptyset$ & $\emptyset$ & $ \begin{aligned} \emptyset\end{aligned}$ & $ \begin{aligned} \emptyset\end{aligned}$ \\
\hline
\hline
\end{tabular}
\end{table}
\begin{table}[H]
\caption{Group $C_{2h}$}
\label{Group C2h}
\centering
\begin{tabular}{c|cccc|c|c}
\hline \hline
 Irrep & $\mathcal{C}_{\mathrm{1st}}$ & $\mathcal{C}_{\mathrm{2nd}}$ & $\mathcal{C}_{\mathrm{3rd}}$ & $\mathcal{C}_{\mathrm{full}}$ & 2nd Generators & 3rd Generators\\    
\hline
$A_{u}$ & $\mathbbm{Z}$ & $\mathbbm{Z}\otimes\mathbbm{Z}_2$ & $\mathbbm{Z}\otimes\mathbbm{Z}_2$ & $\mathbbm{Z}^3\otimes\mathbbm{Z}_2$ & $ \begin{aligned} &[M_{z} ,I] \\&[C_{2z},I] \\\end{aligned}$ & $ \begin{aligned} &[C_{2y}] \\&[I] \\\end{aligned}$ \\
\hline
$A_{g}$ & $\emptyset$ & $\emptyset$ & $\emptyset$ & $\emptyset$ & $ \begin{aligned} \emptyset\end{aligned}$ & $ \begin{aligned} \emptyset\end{aligned}$ \\
\hline
$B_{u}$ & $\emptyset$ & $\mathbbm{Z}_2$ & $\mathbbm{Z}_2$ & $\mathbbm{Z}_4$ & $ \begin{aligned} &[C_{2z},I] \\\end{aligned}$ & $ \begin{aligned} &[I] \\\end{aligned}$ \\
\hline
$B_{g}$ & $\emptyset$ & $\mathbbm{Z}$ & $\emptyset$ & $\mathbbm{Z}$ & $ \begin{aligned} &[C_{2z},M_{z} ] \\\end{aligned}$ & $ \begin{aligned} \emptyset\end{aligned}$ \\
\hline
\hline
\end{tabular}
\end{table}
\begin{table}[H]
\caption{Group $D_{2}$}
\label{Group D2}
\centering
\begin{tabular}{c|cccc|c|c}
\hline \hline
 Irrep & $\mathcal{C}_{\mathrm{1st}}$ & $\mathcal{C}_{\mathrm{2nd}}$ & $\mathcal{C}_{\mathrm{3rd}}$ & $\mathcal{C}_{\mathrm{full}}$ & 2nd Generators & 3rd Generators\\    
\hline
$A_{1}$ & $\mathbbm{Z}$ & $\mathbbm{Z}_2^2$ & $\mathbbm{Z}^3$ & $\mathbbm{Z}^4$ & $ \begin{aligned} &[C_{2z},C_{2y}] \\&[C_{2x},C_{2y}] \\\end{aligned}$ & $ \begin{aligned} &[C_{2z}] \\&[C_{2y}] \\&[C_{2x}] \\\end{aligned}$ \\
\hline
$B_{1}$ & $\emptyset$ & $\mathbbm{Z}_2$ & $\emptyset$ & $\mathbbm{Z}_2$ & $ \begin{aligned} &[C_{2x},C_{2y}] \\\end{aligned}$ & $ \begin{aligned} \emptyset\end{aligned}$ \\
\hline
$B_{2}$ & $\emptyset$ & $\mathbbm{Z}_2$ & $\emptyset$ & $\mathbbm{Z}_2$ & $ \begin{aligned} &[C_{2z},C_{2x}] \\\end{aligned}$ & $ \begin{aligned} \emptyset\end{aligned}$ \\
\hline
$B_{3}$ & $\emptyset$ & $\mathbbm{Z}_2$ & $\emptyset$ & $\mathbbm{Z}_2$ & $ \begin{aligned} &[C_{2z},C_{2y}] \\\end{aligned}$ & $ \begin{aligned} \emptyset\end{aligned}$ \\
\hline
\hline
\end{tabular}
\end{table}
\begin{table}[H]
\caption{Group $C_{2v}$}
\label{Group C2v}
\centering
\begin{tabular}{c|cccc|c|c}
\hline \hline
 Irrep & $\mathcal{C}_{\mathrm{1st}}$ & $\mathcal{C}_{\mathrm{2nd}}$ & $\mathcal{C}_{\mathrm{3rd}}$ & $\mathcal{C}_{\mathrm{full}}$ & 2nd Generators & 3rd Generators\\    
\hline
$A_{2}$ & $\mathbbm{Z}$ & $\mathbbm{Z}^2$ & $\mathbbm{Z}$ & $\mathbbm{Z}^4$ & $ \begin{aligned} &[ M_{y},C_{2z}] \\&[C_{2z},M_{x} ] \\\end{aligned}$ & $ \begin{aligned} &[C_{2z}] \\\end{aligned}$ \\
\hline
$A_{1}$ & $\emptyset$ & $\emptyset$ & $\emptyset$ & $\emptyset$ & $ \begin{aligned} \emptyset\end{aligned}$ & $ \begin{aligned} \emptyset\end{aligned}$ \\
\hline
$B_{2}$ & $\emptyset$ & $\mathbbm{Z}$ & $\emptyset$ & $\mathbbm{Z}$ & $ \begin{aligned} &[ M_{y},C_{2z}] \\\end{aligned}$ & $ \begin{aligned} \emptyset\end{aligned}$ \\
\hline
$B_{1}$ & $\emptyset$ & $\mathbbm{Z}$ & $\emptyset$ & $\mathbbm{Z}$ & $ \begin{aligned} &[C_{2z},M_{x} ] \\\end{aligned}$ & $ \begin{aligned} \emptyset\end{aligned}$ \\
\hline
\hline
\end{tabular}
\end{table}
\begin{table}[H]
\caption{Group $D_{2h}$}
\label{Group D2h}
\centering
\begin{tabular}{c|cccc|c|c}
\hline \hline
 Irrep & $\mathcal{C}_{\mathrm{1st}}$ & $\mathcal{C}_{\mathrm{2nd}}$ & $\mathcal{C}_{\mathrm{3rd}}$ & $\mathcal{C}_{\mathrm{full}}$ & 2nd Generators & 3rd Generators\\    
\hline
$A_{u}$ & $\mathbbm{Z}$ & $\mathbbm{Z}^3$ & $\mathbbm{Z}^3$ & $\mathbbm{Z}^7$ & $ \begin{aligned} &[ M_{y},C_{2z},C_{2x},I] \\&[C_{2z},M_{x} ,I,C_{2y}] \\&[M_{z} ,C_{2x},I,C_{2y}] \\\end{aligned}$ & $ \begin{aligned} &[C_{2z},I] \\&[C_{2y},I] \\&[C_{2x},I] \\\end{aligned}$ \\
\hline
$A_{g}$ & $\emptyset$ & $\emptyset$ & $\emptyset$ & $\emptyset$ & $ \begin{aligned} \emptyset\end{aligned}$ & $ \begin{aligned} \emptyset\end{aligned}$ \\
\hline
$B_{1u}$ & $\emptyset$ & $\mathbbm{Z}$ & $\mathbbm{Z}_2$ & $\mathbbm{Z}\otimes\mathbbm{Z}_2$ & $ \begin{aligned} &[M_{z} ,C_{2x},I,C_{2y}] \\\end{aligned}$ & $ \begin{aligned} &[I] \\\end{aligned}$ \\
\hline
$B_{1g}$ & $\emptyset$ & $\mathbbm{Z}^2$ & $\emptyset$ & $\mathbbm{Z}^2$ & $ \begin{aligned} &[ M_{y},M_{x} ,C_{2x},C_{2y}] \\&[M_{x} ] \\\end{aligned}$ & $ \begin{aligned} \emptyset\end{aligned}$ \\
\hline
$B_{2u}$ & $\emptyset$ & $\mathbbm{Z}$ & $\mathbbm{Z}_2$ & $\mathbbm{Z}\otimes\mathbbm{Z}_2$ & $ \begin{aligned} &[ M_{y},C_{2z},C_{2x},I] \\\end{aligned}$ & $ \begin{aligned} &[I] \\\end{aligned}$ \\
\hline
$B_{2g}$ & $\emptyset$ & $\mathbbm{Z}^2$ & $\emptyset$ & $\mathbbm{Z}^2$ & $ \begin{aligned} &[C_{2z},M_{x} ,M_{z} ,C_{2x}] \\&[M_{z} ] \\\end{aligned}$ & $ \begin{aligned} \emptyset\end{aligned}$ \\
\hline
$B_{3u}$ & $\emptyset$ & $\mathbbm{Z}$ & $\mathbbm{Z}_2$ & $\mathbbm{Z}\otimes\mathbbm{Z}_2$ & $ \begin{aligned} &[C_{2z},M_{x} ,I,C_{2y}] \\\end{aligned}$ & $ \begin{aligned} &[I] \\\end{aligned}$ \\
\hline
$B_{3g}$ & $\emptyset$ & $\mathbbm{Z}^2$ & $\emptyset$ & $\mathbbm{Z}^2$ & $ \begin{aligned} &[ M_{y},C_{2z},M_{z} ,C_{2y}] \\&[M_{z} ] \\\end{aligned}$ & $ \begin{aligned} \emptyset\end{aligned}$ \\
\hline
\hline
\end{tabular}
\end{table}
\begin{table}[H]
\caption{Group $C_{4}$}
\label{Group C4}
\centering
\begin{tabular}{c|cccc|c|c}
\hline \hline
 Irrep & $\mathcal{C}_{\mathrm{1st}}$ & $\mathcal{C}_{\mathrm{2nd}}$ & $\mathcal{C}_{\mathrm{3rd}}$ & $\mathcal{C}_{\mathrm{full}}$ & 2nd Generators & 3rd Generators\\    
\hline
$A_{}$ & $\mathbbm{Z}$ & $\mathbbm{Z}_2$ & $\mathbbm{Z}^2$ & $\mathbbm{Z}^3$ & $ \begin{aligned} &[C_{4z},C_{4z}^3] \\\end{aligned}$ & $ \begin{aligned} &[C_{2z}] \\&[C_{4z},C_{4z},C_{4z}^3,C_{4z}^3] \\\end{aligned}$ \\
\hline
$B_{}$ & $\emptyset$ & $\mathbbm{Z}_2$ & $\emptyset$ & $\mathbbm{Z}_2$ & $ \begin{aligned} &[C_{4z},C_{4z}^3] \\\end{aligned}$ & $ \begin{aligned} \emptyset\end{aligned}$ \\
\hline
\hline
\end{tabular}
\end{table}
\begin{table}[H]
\caption{Group $S_{4}$}
\label{Group S4}
\centering
\begin{tabular}{c|cccc|c|c}
\hline \hline
 Irrep & $\mathcal{C}_{\mathrm{1st}}$ & $\mathcal{C}_{\mathrm{2nd}}$ & $\mathcal{C}_{\mathrm{3rd}}$ & $\mathcal{C}_{\mathrm{full}}$ & 2nd Generators & 3rd Generators\\    
\hline
$B_{}$ & $\mathbbm{Z}$ & $\mathbbm{Z}_2$ & $\mathbbm{Z}$ & $\mathbbm{Z}^2\otimes\mathbbm{Z}_2$ & $ \begin{aligned} &[S_{4z}^3,S_{4z}] \\\end{aligned}$ & $ \begin{aligned} &[C_{2z}] \\\end{aligned}$ \\
\hline
$A_{}$ & $\emptyset$ & $\mathbbm{Z}_2$ & $\emptyset$ & $\mathbbm{Z}_2$ & $ \begin{aligned} &[S_{4z}^3,S_{4z}] \\\end{aligned}$ & $ \begin{aligned} \emptyset\end{aligned}$ \\
\hline
\hline
\end{tabular}
\end{table}
\begin{table}[H]
\caption{Group $C_{4h}$}
\label{Group C4h}
\centering
\begin{tabular}{c|cccc|c|c}
\hline \hline
 Irrep & $\mathcal{C}_{\mathrm{1st}}$ & $\mathcal{C}_{\mathrm{2nd}}$ & $\mathcal{C}_{\mathrm{3rd}}$ & $\mathcal{C}_{\mathrm{full}}$ & 2nd Generators & 3rd Generators\\    
\hline
$A_{u}$ & $\mathbbm{Z}$ & $\mathbbm{Z}\otimes\mathbbm{Z}_2$ & $\mathbbm{Z}^2$ & $\mathbbm{Z}^4$ & $ \begin{aligned} &[S_{4z}^3,S_{4z},M_{z} ,I] \\&[C_{4z},C_{4z}^3,S_{4z}^3,S_{4z}] \\\end{aligned}$ & $ \begin{aligned} &[C_{2z},I] \\&[C_{4z},C_{4z},C_{4z}^3,C_{4z}^3] \\\end{aligned}$ \\
\hline
$A_{g}$ & $\emptyset$ & $\emptyset$ & $\emptyset$ & $\emptyset$ & $ \begin{aligned} \emptyset\end{aligned}$ & $ \begin{aligned} \emptyset\end{aligned}$ \\
\hline
$B_{u}$ & $\emptyset$ & $\mathbbm{Z}\otimes\mathbbm{Z}_2$ & $\emptyset$ & $\mathbbm{Z}\otimes\mathbbm{Z}_2$ & $ \begin{aligned} &[S_{4z}^3,S_{4z},M_{z} ,I] \\&[C_{4z},C_{4z}^3,S_{4z}^3,S_{4z}] \\\end{aligned}$ & $ \begin{aligned} \emptyset\end{aligned}$ \\
\hline
$B_{g}$ & $\emptyset$ & $\mathbbm{Z}_2$ & $\emptyset$ & $\mathbbm{Z}_2$ & $ \begin{aligned} &[C_{4z},C_{4z}^3,S_{4z}^3,S_{4z}] \\\end{aligned}$ & $ \begin{aligned} \emptyset\end{aligned}$ \\
\hline
\hline
\end{tabular}
\end{table}
\begin{table}[H]
\caption{Group $D_{4}$}
\label{Group D4}
\centering
\begin{tabular}{c|cccc|c|c}
\hline \hline
 Irrep & $\mathcal{C}_{\mathrm{1st}}$ & $\mathcal{C}_{\mathrm{2nd}}$ & $\mathcal{C}_{\mathrm{3rd}}$ & $\mathcal{C}_{\mathrm{full}}$ & 2nd Generators & 3rd Generators\\    
\hline
$A_{1}$ & $\mathbbm{Z}$ & $\mathbbm{Z}_2^2$ & $\mathbbm{Z}^4$ & $\mathbbm{Z}^5$ & $ \begin{aligned} &[C_{4z},C_{2x}] \\&[C_{2\bar{x}y} ,C_{2x}] \\\end{aligned}$ & $ \begin{aligned} &[C_{2z}] \\&[C_{4z},C_{4z},C_{4z}^3,C_{4z}^3] \\&[C_{2y},C_{2x}] \\&[C_{2xy} ,C_{2\bar{x}y} ] \\\end{aligned}$ \\
\hline
$A_{2}$ & $\emptyset$ & $\mathbbm{Z}_2$ & $\emptyset$ & $\mathbbm{Z}_2$ & $ \begin{aligned} &[C_{2\bar{x}y} ,C_{2x}] \\\end{aligned}$ & $ \begin{aligned} \emptyset\end{aligned}$ \\
\hline
$B_{1}$ & $\emptyset$ & $\mathbbm{Z}_2^2$ & $\mathbbm{Z}$ & $\mathbbm{Z}\otimes\mathbbm{Z}_2$ & $ \begin{aligned} &[C_{2\bar{x}y} ,C_{2x}] \\&[C_{4z},C_{2\bar{x}y} ] \\\end{aligned}$ & $ \begin{aligned} &[C_{2y}] \\\end{aligned}$ \\
\hline
$B_{2}$ & $\emptyset$ & $\mathbbm{Z}_2^2$ & $\mathbbm{Z}$ & $\mathbbm{Z}\otimes\mathbbm{Z}_2$ & $ \begin{aligned} &[C_{2\bar{x}y} ,C_{2x}] \\&[C_{4z},C_{2x}] \\\end{aligned}$ & $ \begin{aligned} &[C_{2xy} ] \\\end{aligned}$ \\
\hline
\hline
\end{tabular}
\end{table}
\begin{table}[H]
\caption{Group $C_{4v}$}
\label{Group C4v}
\centering
\begin{tabular}{c|cccc|c|c}
\hline \hline
 Irrep & $\mathcal{C}_{\mathrm{1st}}$ & $\mathcal{C}_{\mathrm{2nd}}$ & $\mathcal{C}_{\mathrm{3rd}}$ & $\mathcal{C}_{\mathrm{full}}$ & 2nd Generators & 3rd Generators\\    
\hline
$A_{2}$ & $\mathbbm{Z}$ & $\mathbbm{Z}^2$ & $\mathbbm{Z}^2$ & $\mathbbm{Z}^5$ & $ \begin{aligned} &[ M_{\bar{x}y},C_{4z}] \\&[C_{4z}, M_{y}] \\\end{aligned}$ & $ \begin{aligned} &[C_{2z}] \\&[C_{4z},C_{4z},C_{4z}^3,C_{4z}^3] \\\end{aligned}$ \\
\hline
$A_{1}$ & $\emptyset$ & $\emptyset$ & $\emptyset$ & $\emptyset$ & $ \begin{aligned} \emptyset\end{aligned}$ & $ \begin{aligned} \emptyset\end{aligned}$ \\
\hline
$B_{2}$ & $\emptyset$ & $\mathbbm{Z}$ & $\emptyset$ & $\mathbbm{Z}$ & $ \begin{aligned} &[C_{4z}, M_{y}] \\\end{aligned}$ & $ \begin{aligned} \emptyset\end{aligned}$ \\
\hline
$B_{1}$ & $\emptyset$ & $\mathbbm{Z}$ & $\emptyset$ & $\mathbbm{Z}$ & $ \begin{aligned} &[ M_{\bar{x}y},C_{4z}] \\\end{aligned}$ & $ \begin{aligned} \emptyset\end{aligned}$ \\
\hline
\hline
\end{tabular}
\end{table}
\begin{table}[H]
\caption{Group $D_{2d}$}
\label{Group D2d}
\centering
\begin{tabular}{c|cccc|c|c}
\hline \hline
 Irrep & $\mathcal{C}_{\mathrm{1st}}$ & $\mathcal{C}_{\mathrm{2nd}}$ & $\mathcal{C}_{\mathrm{3rd}}$ & $\mathcal{C}_{\mathrm{full}}$ & 2nd Generators & 3rd Generators\\    
\hline
$B_{1}$ & $\mathbbm{Z}$ & $\mathbbm{Z}\otimes\mathbbm{Z}_2$ & $\mathbbm{Z}^2$ & $\mathbbm{Z}^4$ & $ \begin{aligned} &[ M_{\bar{x}y},S_{4z}^3] \\&[S_{4z}^3,C_{2x}] \\\end{aligned}$ & $ \begin{aligned} &[C_{2z}] \\&[C_{2y},C_{2x}] \\\end{aligned}$ \\
\hline
$B_{2}$ & $\emptyset$ & $\mathbbm{Z}_2$ & $\emptyset$ & $\mathbbm{Z}_2$ & $ \begin{aligned} &[S_{4z}^3,C_{2x}] \\\end{aligned}$ & $ \begin{aligned} \emptyset\end{aligned}$ \\
\hline
$A_{1}$ & $\emptyset$ & $\mathbbm{Z}_2$ & $\mathbbm{Z}$ & $\mathbbm{Z}$ & $ \begin{aligned} &[S_{4z}^3,C_{2x}] \\\end{aligned}$ & $ \begin{aligned} &[C_{2y}] \\\end{aligned}$ \\
\hline
$A_{2}$ & $\emptyset$ & $\mathbbm{Z}\otimes\mathbbm{Z}_2$ & $\emptyset$ & $\mathbbm{Z}\otimes\mathbbm{Z}_2$ & $ \begin{aligned} &[ M_{\bar{x}y},S_{4z}^3] \\&[S_{4z}^3,C_{2x}] \\\end{aligned}$ & $ \begin{aligned} \emptyset\end{aligned}$ \\
\hline
\hline
\end{tabular}
\end{table}
\begin{table}[H]
\caption{Group $D_{4h}$}
\label{Group D4h}
\centering
\begin{tabular}{c|cccc|c|c}
\hline \hline
 Irrep & $\mathcal{C}_{\mathrm{1st}}$ & $\mathcal{C}_{\mathrm{2nd}}$ & $\mathcal{C}_{\mathrm{3rd}}$ & $\mathcal{C}_{\mathrm{full}}$ & 2nd Generators & 3rd Generators\\    
\hline
$A_{1u}$ & $\mathbbm{Z}$ & $\mathbbm{Z}^3$ & $\mathbbm{Z}^4$ & $\mathbbm{Z}^8$ & $ \begin{aligned} &[ M_{\bar{x}y},C_{4z},S_{4z}^3,C_{2\bar{x}y} ] \\&[C_{4z},S_{4z}^3, M_{y},C_{2x}] \\&[S_{4z}^3,M_{z} ,C_{2\bar{x}y} ,C_{2x},I] \\\end{aligned}$ & $ \begin{aligned} &[C_{2z},I] \\&[C_{4z},C_{4z},C_{4z}^3,C_{4z}^3] \\&[C_{2y},C_{2x}] \\&[C_{2xy} ,C_{2\bar{x}y} ] \\\end{aligned}$ \\
\hline
$A_{1g}$ & $\emptyset$ & $\emptyset$ & $\emptyset$ & $\emptyset$ & $ \begin{aligned} \emptyset\end{aligned}$ & $ \begin{aligned} \emptyset\end{aligned}$ \\
\hline
$A_{2u}$ & $\emptyset$ & $\mathbbm{Z}$ & $\emptyset$ & $\mathbbm{Z}$ & $ \begin{aligned} &[S_{4z}^3,M_{z} ,C_{2\bar{x}y} ,C_{2x},I] \\\end{aligned}$ & $ \begin{aligned} \emptyset\end{aligned}$ \\
\hline
$A_{2g}$ & $\emptyset$ & $\mathbbm{Z}^2$ & $\emptyset$ & $\mathbbm{Z}^2$ & $ \begin{aligned} &[ M_{\bar{x}y}, M_{y},C_{2\bar{x}y} ,C_{2x}] \\&[ M_{y}] \\\end{aligned}$ & $ \begin{aligned} \emptyset\end{aligned}$ \\
\hline
$B_{1u}$ & $\emptyset$ & $\mathbbm{Z}^2$ & $\mathbbm{Z}$ & $\mathbbm{Z}^3$ & $ \begin{aligned} &[C_{4z},S_{4z}^3, M_{y},C_{2x}] \\&[S_{4z}^3,M_{z} ,C_{2\bar{x}y} ,C_{2x},I] \\\end{aligned}$ & $ \begin{aligned} &[C_{2y}] \\\end{aligned}$ \\
\hline
$B_{1g}$ & $\emptyset$ & $\mathbbm{Z}$ & $\emptyset$ & $\mathbbm{Z}$ & $ \begin{aligned} &[ M_{\bar{x}y},C_{4z},S_{4z}^3,C_{2\bar{x}y} ] \\\end{aligned}$ & $ \begin{aligned} \emptyset\end{aligned}$ \\
\hline
$B_{2u}$ & $\emptyset$ & $\mathbbm{Z}^2$ & $\mathbbm{Z}$ & $\mathbbm{Z}^3$ & $ \begin{aligned} &[ M_{\bar{x}y},C_{4z},S_{4z}^3,C_{2\bar{x}y} ] \\&[S_{4z}^3,M_{z} ,C_{2\bar{x}y} ,C_{2x},I] \\\end{aligned}$ & $ \begin{aligned} &[C_{2xy} ] \\\end{aligned}$ \\
\hline
$B_{2g}$ & $\emptyset$ & $\mathbbm{Z}$ & $\emptyset$ & $\mathbbm{Z}$ & $ \begin{aligned} &[C_{4z},S_{4z}^3, M_{y},C_{2x}] \\\end{aligned}$ & $ \begin{aligned} \emptyset\end{aligned}$ \\
\hline
\hline
\end{tabular}
\end{table}
\begin{table}[H]
\caption{Group $C_{3}$}
\label{Group C3}
\centering
\begin{tabular}{c|cccc|c|c}
\hline \hline
 Irrep & $\mathcal{C}_{\mathrm{1st}}$ & $\mathcal{C}_{\mathrm{2nd}}$ & $\mathcal{C}_{\mathrm{3rd}}$ & $\mathcal{C}_{\mathrm{full}}$ & 2nd Generators & 3rd Generators\\    
\hline
$A_{}$ & $\mathbbm{Z}$ & $\emptyset$ & $\mathbbm{Z}$ & $\mathbbm{Z}^2$ & $ \begin{aligned} \emptyset\end{aligned}$ & $ \begin{aligned} &[C_{3z} ,C_{3z}^2 ] \\\end{aligned}$ \\
\hline
\hline
\end{tabular}
\end{table}
\begin{table}[H]
\caption{Group $C_{3i}$}
\label{Group C3i}
\centering
\begin{tabular}{c|cccc|c|c}
\hline \hline
 Irrep & $\mathcal{C}_{\mathrm{1st}}$ & $\mathcal{C}_{\mathrm{2nd}}$ & $\mathcal{C}_{\mathrm{3rd}}$ & $\mathcal{C}_{\mathrm{full}}$ & 2nd Generators & 3rd Generators\\    
\hline
$A_{u}$ & $\mathbbm{Z}$ & $\mathbbm{Z}_2$ & $\mathbbm{Z}\otimes\mathbbm{Z}_2$ & $\mathbbm{Z}^2\otimes\mathbbm{Z}_4$ & $ \begin{aligned} &[S_{3z}^5 ,S_{3z} ,I] \\\end{aligned}$ & $ \begin{aligned} &[C_{3z} ,C_{3z}^2 ] \\&[I] \\\end{aligned}$ \\
\hline
$A_{g}$ & $\emptyset$ & $\emptyset$ & $\emptyset$ & $\emptyset$ & $ \begin{aligned} \emptyset\end{aligned}$ & $ \begin{aligned} \emptyset\end{aligned}$ \\
\hline
\hline
\end{tabular}
\end{table}
\begin{table}[H]
\caption{Group $D_{3}$}
\label{Group D3}
\centering
\begin{tabular}{c|cccc|c|c}
\hline \hline
 Irrep & $\mathcal{C}_{\mathrm{1st}}$ & $\mathcal{C}_{\mathrm{2nd}}$ & $\mathcal{C}_{\mathrm{3rd}}$ & $\mathcal{C}_{\mathrm{full}}$ & 2nd Generators & 3rd Generators\\    
\hline
$A_{1}$ & $\mathbbm{Z}$ & $\mathbbm{Z}_2$ & $\mathbbm{Z}^2$ & $\mathbbm{Z}^3$ & $ \begin{aligned} &[C_{2x}] \\\end{aligned}$ & $ \begin{aligned} &[C_{3z} ,C_{3z}^2 ] \\&[C_{2\bar{x}y} ,C_{2,x2y},C_{2,2xy} ] \\\end{aligned}$ \\
\hline
$A_{2}$ & $\emptyset$ & $\mathbbm{Z}_2$ & $\emptyset$ & $\mathbbm{Z}_2$ & $ \begin{aligned} &[C_{2x}] \\\end{aligned}$ & $ \begin{aligned} \emptyset\end{aligned}$ \\
\hline
\hline
\end{tabular}
\end{table}
\begin{table}[H]
\caption{Group $C_{3v}$}
\label{Group C3v}
\centering
\begin{tabular}{c|cccc|c|c}
\hline \hline
 Irrep & $\mathcal{C}_{\mathrm{1st}}$ & $\mathcal{C}_{\mathrm{2nd}}$ & $\mathcal{C}_{\mathrm{3rd}}$ & $\mathcal{C}_{\mathrm{full}}$ & 2nd Generators & 3rd Generators\\    
\hline
$A_{2}$ & $\mathbbm{Z}$ & $\mathbbm{Z}$ & $\mathbbm{Z}$ & $\mathbbm{Z}^3$ & $ \begin{aligned} &[M_{x2y} ] \\\end{aligned}$ & $ \begin{aligned} &[C_{3z} ,C_{3z}^2 ] \\\end{aligned}$ \\
\hline
$A_{1}$ & $\emptyset$ & $\emptyset$ & $\emptyset$ & $\emptyset$ & $ \begin{aligned} \emptyset\end{aligned}$ & $ \begin{aligned} \emptyset\end{aligned}$ \\
\hline
\hline
\end{tabular}
\end{table}
\begin{table}[H]
\caption{Group $D_{3d}$}
\label{Group D3d}
\centering
\begin{tabular}{c|cccc|c|c}
\hline \hline
 Irrep & $\mathcal{C}_{\mathrm{1st}}$ & $\mathcal{C}_{\mathrm{2nd}}$ & $\mathcal{C}_{\mathrm{3rd}}$ & $\mathcal{C}_{\mathrm{full}}$ & 2nd Generators & 3rd Generators\\    
\hline
$A_{1u}$ & $\mathbbm{Z}$ & $\mathbbm{Z}\otimes\mathbbm{Z}_2$ & $\mathbbm{Z}^2\otimes\mathbbm{Z}_2$ & $\mathbbm{Z}^4\otimes\mathbbm{Z}_2$ & $ \begin{aligned} &[S_{3z}^5 , M_{y},I] \\&[S_{3z}^5 ,C_{2x},I] \\\end{aligned}$ & $ \begin{aligned} &[C_{3z} ,C_{3z}^2 ] \\&[C_{2\bar{x}y} ,C_{2,x2y},C_{2,2xy} ] \\&[I] \\\end{aligned}$ \\
\hline
$A_{1g}$ & $\emptyset$ & $\emptyset$ & $\emptyset$ & $\emptyset$ & $ \begin{aligned} \emptyset\end{aligned}$ & $ \begin{aligned} \emptyset\end{aligned}$ \\
\hline
$A_{2u}$ & $\emptyset$ & $\mathbbm{Z}_2$ & $\mathbbm{Z}_2$ & $\mathbbm{Z}_4$ & $ \begin{aligned} &[S_{3z}^5 ,C_{2x},I] \\\end{aligned}$ & $ \begin{aligned} &[I] \\\end{aligned}$ \\
\hline
$A_{2g}$ & $\emptyset$ & $\mathbbm{Z}$ & $\emptyset$ & $\mathbbm{Z}$ & $ \begin{aligned} &[ M_{y},C_{2x}] \\\end{aligned}$ & $ \begin{aligned} \emptyset\end{aligned}$ \\
\hline
\hline
\end{tabular}
\end{table}
\begin{table}[H]
\caption{Group $C_{6}$}
\label{Group C6}
\centering
\begin{tabular}{c|cccc|c|c}
\hline \hline
 Irrep & $\mathcal{C}_{\mathrm{1st}}$ & $\mathcal{C}_{\mathrm{2nd}}$ & $\mathcal{C}_{\mathrm{3rd}}$ & $\mathcal{C}_{\mathrm{full}}$ & 2nd Generators & 3rd Generators\\    
\hline
$A_{}$ & $\mathbbm{Z}$ & $\mathbbm{Z}_2$ & $\mathbbm{Z}^3$ & $\mathbbm{Z}^4$ & $ \begin{aligned} &[C_{6z} ,C_{6z}^5 ,C_{2z}] \\\end{aligned}$ & $ \begin{aligned} &[C_{3z} ,C_{3z}^2 ] \\&[C_{2z}] \\&[C_{6z}^5 ,C_{6z}^5 ,C_{6z}^5 ,C_{6z} ,C_{6z} ,C_{6z} ] \\\end{aligned}$ \\
\hline
$B_{}$ & $\emptyset$ & $\mathbbm{Z}_2$ & $\emptyset$ & $\mathbbm{Z}_2$ & $ \begin{aligned} &[C_{6z} ,C_{6z}^5 ,C_{2z}] \\\end{aligned}$ & $ \begin{aligned} \emptyset\end{aligned}$ \\
\hline
\hline
\end{tabular}
\end{table}
\begin{table}[H]
\caption{Group $C_{3h}$}
\label{Group C3h}
\centering
\begin{tabular}{c|cccc|c|c}
\hline \hline
 Irrep & $\mathcal{C}_{\mathrm{1st}}$ & $\mathcal{C}_{\mathrm{2nd}}$ & $\mathcal{C}_{\mathrm{3rd}}$ & $\mathcal{C}_{\mathrm{full}}$ & 2nd Generators & 3rd Generators\\    
\hline
$A_{''}$ & $\mathbbm{Z}$ & $\mathbbm{Z}$ & $\mathbbm{Z}$ & $\mathbbm{Z}^3$ & $ \begin{aligned} &[S_{6z}^5 , S_{6z},M_{z} ] \\\end{aligned}$ & $ \begin{aligned} &[C_{3z} ,C_{3z}^2 ] \\\end{aligned}$ \\
\hline
$A_{'}$ & $\emptyset$ & $\emptyset$ & $\emptyset$ & $\emptyset$ & $ \begin{aligned} \emptyset\end{aligned}$ & $ \begin{aligned} \emptyset\end{aligned}$ \\
\hline
\hline
\end{tabular}
\end{table}
\begin{table}[H]
\caption{Group $C_{6h}$}
\label{Group C6h}
\centering
\begin{tabular}{c|cccc|c|c}
\hline \hline
 Irrep & $\mathcal{C}_{\mathrm{1st}}$ & $\mathcal{C}_{\mathrm{2nd}}$ & $\mathcal{C}_{\mathrm{3rd}}$ & $\mathcal{C}_{\mathrm{full}}$ & 2nd Generators & 3rd Generators\\    
\hline
$A_{u}$ & $\mathbbm{Z}$ & $\mathbbm{Z}\otimes\mathbbm{Z}_2$ & $\mathbbm{Z}^3\otimes\mathbbm{Z}_2$ & $\mathbbm{Z}^5\otimes\mathbbm{Z}_2$ & $ \begin{aligned} &[S_{3z}^5 ,S_{6z}^5 , S_{6z},S_{3z} ,M_{z} ,I] \\&[C_{6z} ,C_{6z}^5 ,S_{3z}^5 ,S_{3z} ,C_{2z},I] \\\end{aligned}$ & $ \begin{aligned} &[C_{3z} ,C_{3z}^2 ] \\&[C_{2z}] \\&[C_{6z}^5 ,C_{6z}^5 ,C_{6z}^5 ,C_{6z} ,C_{6z} ,C_{6z} ] \\&[I] \\\end{aligned}$ \\
\hline
$A_{g}$ & $\emptyset$ & $\emptyset$ & $\emptyset$ & $\emptyset$ & $ \begin{aligned} \emptyset\end{aligned}$ & $ \begin{aligned} \emptyset\end{aligned}$ \\
\hline
$B_{u}$ & $\emptyset$ & $\mathbbm{Z}_2$ & $\mathbbm{Z}_2$ & $\mathbbm{Z}_4$ & $ \begin{aligned} &[C_{6z} ,C_{6z}^5 ,S_{3z}^5 ,S_{3z} ,C_{2z},I] \\\end{aligned}$ & $ \begin{aligned} &[I] \\\end{aligned}$ \\
\hline
$B_{g}$ & $\emptyset$ & $\mathbbm{Z}$ & $\emptyset$ & $\mathbbm{Z}$ & $ \begin{aligned} &[C_{6z} ,C_{6z}^5 ,S_{6z}^5 , S_{6z},C_{2z},M_{z} ] \\\end{aligned}$ & $ \begin{aligned} \emptyset\end{aligned}$ \\
\hline
\hline
\end{tabular}
\end{table}
\begin{table}[H]
\caption{Group $D_{6}$}
\label{Group D6}
\centering
\begin{tabular}{c|cccc|c|c}
\hline \hline
 Irrep & $\mathcal{C}_{\mathrm{1st}}$ & $\mathcal{C}_{\mathrm{2nd}}$ & $\mathcal{C}_{\mathrm{3rd}}$ & $\mathcal{C}_{\mathrm{full}}$ & 2nd Generators & 3rd Generators\\    
\hline
$A_{1}$ & $\mathbbm{Z}$ & $\mathbbm{Z}_2^2$ & $\mathbbm{Z}^5$ & $\mathbbm{Z}^6$ & $ \begin{aligned} &[C_{6z} ,C_{2z},C_{2x}] \\&[C_{2\bar{x}y} ,C_{2x}] \\\end{aligned}$ & $ \begin{aligned} &[C_{3z} ,C_{3z}^2 ] \\&[C_{2z}] \\&[C_{6z}^5 ,C_{6z}^5 ,C_{6z}^5 ,C_{6z} ,C_{6z} ,C_{6z} ] \\&[C_{2xy} ,C_{2x},C_{2y}] \\&[C_{2\bar{x}y} ,C_{2,x2y},C_{2,2xy} ] \\\end{aligned}$ \\
\hline
$A_{2}$ & $\emptyset$ & $\mathbbm{Z}_2$ & $\emptyset$ & $\mathbbm{Z}_2$ & $ \begin{aligned} &[C_{2\bar{x}y} ,C_{2x}] \\\end{aligned}$ & $ \begin{aligned} \emptyset\end{aligned}$ \\
\hline
$B_{2}$ & $\emptyset$ & $\mathbbm{Z}_2$ & $\emptyset$ & $\mathbbm{Z}_2$ & $ \begin{aligned} &[C_{6z} ,C_{2z},C_{2x}] \\\end{aligned}$ & $ \begin{aligned} \emptyset\end{aligned}$ \\
\hline
$B_{1}$ & $\emptyset$ & $\mathbbm{Z}_2$ & $\emptyset$ & $\mathbbm{Z}_2$ & $ \begin{aligned} &[C_{6z} ,C_{2z},C_{2\bar{x}y} ] \\\end{aligned}$ & $ \begin{aligned} \emptyset\end{aligned}$ \\
\hline
\hline
\end{tabular}
\end{table}
\begin{table}[H]
\caption{Group $C_{6v}$}
\label{Group C6v}
\centering
\begin{tabular}{c|cccc|c|c}
\hline \hline
 Irrep & $\mathcal{C}_{\mathrm{1st}}$ & $\mathcal{C}_{\mathrm{2nd}}$ & $\mathcal{C}_{\mathrm{3rd}}$ & $\mathcal{C}_{\mathrm{full}}$ & 2nd Generators & 3rd Generators\\    
\hline
$A_{2}$ & $\mathbbm{Z}$ & $\mathbbm{Z}^2$ & $\mathbbm{Z}^3$ & $\mathbbm{Z}^6$ & $ \begin{aligned} &[C_{6z} , M_{y},C_{2z}] \\&[C_{6z} ,M_{x2y} ,C_{2z}] \\\end{aligned}$ & $ \begin{aligned} &[C_{3z} ,C_{3z}^2 ] \\&[C_{2z}] \\&[C_{6z}^5 ,C_{6z}^5 ,C_{6z}^5 ,C_{6z} ,C_{6z} ,C_{6z} ] \\\end{aligned}$ \\
\hline
$A_{1}$ & $\emptyset$ & $\emptyset$ & $\emptyset$ & $\emptyset$ & $ \begin{aligned} \emptyset\end{aligned}$ & $ \begin{aligned} \emptyset\end{aligned}$ \\
\hline
$B_{2}$ & $\emptyset$ & $\mathbbm{Z}$ & $\emptyset$ & $\mathbbm{Z}$ & $ \begin{aligned} &[C_{6z} ,M_{x2y} ,C_{2z}] \\\end{aligned}$ & $ \begin{aligned} \emptyset\end{aligned}$ \\
\hline
$B_{1}$ & $\emptyset$ & $\mathbbm{Z}$ & $\emptyset$ & $\mathbbm{Z}$ & $ \begin{aligned} &[C_{6z} , M_{y},C_{2z}] \\\end{aligned}$ & $ \begin{aligned} \emptyset\end{aligned}$ \\
\hline
\hline
\end{tabular}
\end{table}
\begin{table}[H]
\caption{Group $D_{3h}$}
\label{Group D3h}
\centering
\begin{tabular}{c|cccc|c|c}
\hline \hline
 Irrep & $\mathcal{C}_{\mathrm{1st}}$ & $\mathcal{C}_{\mathrm{2nd}}$ & $\mathcal{C}_{\mathrm{3rd}}$ & $\mathcal{C}_{\mathrm{full}}$ & 2nd Generators & 3rd Generators\\    
\hline
$A_{1''}$ & $\mathbbm{Z}$ & $\mathbbm{Z}^2$ & $\mathbbm{Z}^2$ & $\mathbbm{Z}^5$ & $ \begin{aligned} &[M_{x2y} ,C_{2x}] \\&[S_{6z}^5 ,M_{z} ,C_{2x}] \\\end{aligned}$ & $ \begin{aligned} &[C_{3z} ,C_{3z}^2 ] \\&[C_{2xy} ,C_{2x},C_{2y}] \\\end{aligned}$ \\
\hline
$A_{2''}$ & $\emptyset$ & $\mathbbm{Z}$ & $\emptyset$ & $\mathbbm{Z}$ & $ \begin{aligned} &[S_{6z}^5 ,M_{z} ,C_{2x}] \\\end{aligned}$ & $ \begin{aligned} \emptyset\end{aligned}$ \\
\hline
$A_{1'}$ & $\emptyset$ & $\emptyset$ & $\emptyset$ & $\emptyset$ & $ \begin{aligned} \emptyset\end{aligned}$ & $ \begin{aligned} \emptyset\end{aligned}$ \\
\hline
$A_{2'}$ & $\emptyset$ & $\mathbbm{Z}$ & $\emptyset$ & $\mathbbm{Z}$ & $ \begin{aligned} &[M_{x2y} ,C_{2x}] \\\end{aligned}$ & $ \begin{aligned} \emptyset\end{aligned}$ \\
\hline
\hline
\end{tabular}
\end{table}
\begin{table}[H]
\caption{Group $D_{6h}$}
\label{Group D6h}
\centering
\begin{tabular}{c|cccc|c|c}
\hline \hline
 Irrep & $\mathcal{C}_{\mathrm{1st}}$ & $\mathcal{C}_{\mathrm{2nd}}$ & $\mathcal{C}_{\mathrm{3rd}}$ & $\mathcal{C}_{\mathrm{full}}$ & 2nd Generators & 3rd Generators\\    
\hline
$A_{1u}$ & $\mathbbm{Z}$ & $\mathbbm{Z}^3$ & $\mathbbm{Z}^5$ & $\mathbbm{Z}^9$ & $ \begin{aligned} &[C_{6z} ,S_{3z}^5 , M_{y},C_{2z},C_{2\bar{x}y} ,I] \\&[C_{6z} ,S_{3z}^5 ,M_{x2y} ,C_{2z},C_{2x},I] \\&[S_{3z}^5 ,S_{6z}^5 ,M_{z} ,C_{2\bar{x}y} ,C_{2x},I] \\\end{aligned}$ & $ \begin{aligned} &[C_{3z} ,C_{3z}^2 ] \\&[C_{2z},I] \\&[C_{6z}^5 ,C_{6z}^5 ,C_{6z}^5 ,C_{6z} ,C_{6z} ,C_{6z} ] \\&[C_{2xy} ,C_{2x},C_{2y},I] \\&[C_{2\bar{x}y} ,C_{2,x2y},C_{2,2xy} ,I] \\\end{aligned}$ \\
\hline
$A_{1g}$ & $\emptyset$ & $\emptyset$ & $\emptyset$ & $\emptyset$ & $ \begin{aligned} \emptyset\end{aligned}$ & $ \begin{aligned} \emptyset\end{aligned}$ \\
\hline
$A_{2u}$ & $\emptyset$ & $\mathbbm{Z}$ & $\mathbbm{Z}_2$ & $\mathbbm{Z}\otimes\mathbbm{Z}_2$ & $ \begin{aligned} &[S_{3z}^5 ,S_{6z}^5 ,M_{z} ,C_{2\bar{x}y} ,C_{2x},I] \\\end{aligned}$ & $ \begin{aligned} &[I] \\\end{aligned}$ \\
\hline
$A_{2g}$ & $\emptyset$ & $\mathbbm{Z}^2$ & $\emptyset$ & $\mathbbm{Z}^2$ & $ \begin{aligned} &[ M_{y},M_{x2y} ,C_{2\bar{x}y} ,C_{2x}] \\&[M_{x2y} ] \\\end{aligned}$ & $ \begin{aligned} \emptyset\end{aligned}$ \\
\hline
$B_{2u}$ & $\emptyset$ & $\mathbbm{Z}$ & $\mathbbm{Z}_2$ & $\mathbbm{Z}\otimes\mathbbm{Z}_2$ & $ \begin{aligned} &[C_{6z} ,S_{3z}^5 ,M_{x2y} ,C_{2z},C_{2x},I] \\\end{aligned}$ & $ \begin{aligned} &[I] \\\end{aligned}$ \\
\hline
$B_{2g}$ & $\emptyset$ & $\mathbbm{Z}^2$ & $\emptyset$ & $\mathbbm{Z}^2$ & $ \begin{aligned} &[C_{6z} ,S_{6z}^5 , M_{y},C_{2z},M_{z} ,C_{2x}] \\&[M_{z} ] \\\end{aligned}$ & $ \begin{aligned} \emptyset\end{aligned}$ \\
\hline
$B_{1u}$ & $\emptyset$ & $\mathbbm{Z}$ & $\mathbbm{Z}_2$ & $\mathbbm{Z}\otimes\mathbbm{Z}_2$ & $ \begin{aligned} &[C_{6z} ,S_{3z}^5 , M_{y},C_{2z},C_{2\bar{x}y} ,I] \\\end{aligned}$ & $ \begin{aligned} &[I] \\\end{aligned}$ \\
\hline
$B_{1g}$ & $\emptyset$ & $\mathbbm{Z}^2$ & $\emptyset$ & $\mathbbm{Z}^2$ & $ \begin{aligned} &[C_{6z} ,S_{6z}^5 ,M_{x2y} ,C_{2z},M_{z} ,C_{2\bar{x}y} ] \\&[M_{z} ] \\\end{aligned}$ & $ \begin{aligned} \emptyset\end{aligned}$ \\
\hline
\hline
\end{tabular}
\end{table}
\begin{table}[H]
\caption{Group $T_{}$}
\label{Group T}
\centering
\begin{tabular}{c|cccc|c|c}
\hline \hline
 Irrep & $\mathcal{C}_{\mathrm{1st}}$ & $\mathcal{C}_{\mathrm{2nd}}$ & $\mathcal{C}_{\mathrm{3rd}}$ & $\mathcal{C}_{\mathrm{full}}$ & 2nd Generators & 3rd Generators\\    
\hline
$A_{}$ & $\mathbbm{Z}$ & $\emptyset$ & $\mathbbm{Z}^2$ & $\mathbbm{Z}^3$ & $ \begin{aligned} \emptyset\end{aligned}$ & $ \begin{aligned} &[C_{2z},C_{2y},C_{2x}] \\&[C_{3xyz},C_{3\bar{x}y\bar{z}} ,C_{3\bar{x}yz}^2 ,C_{3\bar{x}\bar{y}z} ,C_{3xyz}^2 , C_{3\bar{x}yz},C_{3\bar{x}\bar{y}z}^2 , C_{3\bar{x}y\bar{z}}^2] \\\end{aligned}$ \\
\hline
\hline
\end{tabular}
\end{table}
\begin{table}[H]
\caption{Group $T_{h}$}
\label{Group Th}
\centering
\begin{tabular}{c|cccc|c|c}
\hline \hline
 Irrep & $\mathcal{C}_{\mathrm{1st}}$ & $\mathcal{C}_{\mathrm{2nd}}$ & $\mathcal{C}_{\mathrm{3rd}}$ & $\mathcal{C}_{\mathrm{full}}$ & 2nd Generators & 3rd Generators\\    
\hline
$A_{u}$ & $\mathbbm{Z}$ & $\mathbbm{Z}$ & $\mathbbm{Z}^2$ & $\mathbbm{Z}^4$ & $ \begin{aligned} &[S_{3\bar{x}y\bar{z}} ,S_{3\bar{x}y\bar{z}}^5 , M_{y},I] \\\end{aligned}$ & $ \begin{aligned} &[C_{2z},C_{2y},C_{2x},I] \\&[C_{3xyz},C_{3\bar{x}y\bar{z}} ,C_{3\bar{x}yz}^2 ,C_{3\bar{x}\bar{y}z} ,C_{3xyz}^2 , C_{3\bar{x}yz},C_{3\bar{x}\bar{y}z}^2 , C_{3\bar{x}y\bar{z}}^2] \\\end{aligned}$ \\
\hline
$A_{g}$ & $\emptyset$ & $\emptyset$ & $\emptyset$ & $\emptyset$ & $ \begin{aligned} \emptyset\end{aligned}$ & $ \begin{aligned} \emptyset\end{aligned}$ \\
\hline
\hline
\end{tabular}
\end{table}
\begin{table}[H]
\caption{Group $O_{}$}
\label{Group O}
\centering
\begin{tabular}{c|cccc|c|c}
\hline \hline
 Irrep & $\mathcal{C}_{\mathrm{1st}}$ & $\mathcal{C}_{\mathrm{2nd}}$ & $\mathcal{C}_{\mathrm{3rd}}$ & $\mathcal{C}_{\mathrm{full}}$ & 2nd Generators & 3rd Generators\\    
\hline
$A_{1}$ & $\mathbbm{Z}$ & $\mathbbm{Z}_2$ & $\mathbbm{Z}^4$ & $\mathbbm{Z}^5$ & $ \begin{aligned} &[C_{4z},C_{2xz} ] \\\end{aligned}$ & $ \begin{aligned} &[C_{2z},C_{2y},C_{2x}] \\&[C_{3xyz},C_{3\bar{x}y\bar{z}} ,C_{3\bar{x}yz}^2 ,C_{3\bar{x}\bar{y}z} ,C_{3xyz}^2 , C_{3\bar{x}yz},C_{3\bar{x}\bar{y}z}^2 , C_{3\bar{x}y\bar{z}}^2] \\&[C_{2xy} ,C_{2\bar{x}y} ,C_{2yz} ,C_{2\bar{y}z} ,C_{2xz} ,C_{2\bar{x}z} ] \\&[C_{4z}^3,C_{4z}^3,C_{4z},C_{4z},C_{4x}^3,C_{4x}^3,C_{4x},C_{4x}, C_{4y}, C_{4y},C_{4y}^3,C_{4y}^3] \\\end{aligned}$ \\
\hline
$A_{2}$ & $\emptyset$ & $\mathbbm{Z}_2$ & $\emptyset$ & $\mathbbm{Z}_2$ & $ \begin{aligned} &[C_{4z},C_{2xz} ] \\\end{aligned}$ & $ \begin{aligned} \emptyset\end{aligned}$ \\
\hline
\hline
\end{tabular}
\end{table}
\begin{table}[H]
\caption{Group $T_{d}$}
\label{Group Td}
\centering
\begin{tabular}{c|cccc|c|c}
\hline \hline
 Irrep & $\mathcal{C}_{\mathrm{1st}}$ & $\mathcal{C}_{\mathrm{2nd}}$ & $\mathcal{C}_{\mathrm{3rd}}$ & $\mathcal{C}_{\mathrm{full}}$ & 2nd Generators & 3rd Generators\\    
\hline
$A_{2}$ & $\mathbbm{Z}$ & $\mathbbm{Z}$ & $\mathbbm{Z}^2$ & $\mathbbm{Z}^4$ & $ \begin{aligned} &[ M_{\bar{x}y},S_{4x}] \\\end{aligned}$ & $ \begin{aligned} &[C_{2z},C_{2y},C_{2x}] \\&[C_{3xyz},C_{3\bar{x}y\bar{z}} ,C_{3\bar{x}yz}^2 ,C_{3\bar{x}\bar{y}z} ,C_{3xyz}^2 , C_{3\bar{x}yz},C_{3\bar{x}\bar{y}z}^2 , C_{3\bar{x}y\bar{z}}^2] \\\end{aligned}$ \\
\hline
$A_{1}$ & $\emptyset$ & $\emptyset$ & $\emptyset$ & $\emptyset$ & $ \begin{aligned} \emptyset\end{aligned}$ & $ \begin{aligned} \emptyset\end{aligned}$ \\
\hline
\hline
\end{tabular}
\end{table}
\begin{table}[H]
\caption{Group $O_{h}$}
\label{Group Oh}
\centering
\begin{tabular}{c|cccc|c|c}
\hline \hline
 Irrep & $\mathcal{C}_{\mathrm{1st}}$ & $\mathcal{C}_{\mathrm{2nd}}$ & $\mathcal{C}_{\mathrm{3rd}}$ & $\mathcal{C}_{\mathrm{full}}$ & 2nd Generators & 3rd Generators\\    
\hline
$A_{1u}$ & $\mathbbm{Z}$ & $\mathbbm{Z}^2$ & $\mathbbm{Z}^4$ & $\mathbbm{Z}^7$ & $ \begin{aligned} &[ M_{\bar{x}y},C_{4z},S_{4x},C_{2xz} ] \\&[S_{3\bar{x}y\bar{z}} ,C_{4z}, M_{y},C_{2xz} ,I] \\\end{aligned}$ & $ \begin{aligned} &[C_{2z},C_{2y},C_{2x},I] \\&[C_{3xyz},C_{3\bar{x}y\bar{z}} ,C_{3\bar{x}yz}^2 ,C_{3\bar{x}\bar{y}z} ,C_{3xyz}^2 , C_{3\bar{x}yz},C_{3\bar{x}\bar{y}z}^2 , C_{3\bar{x}y\bar{z}}^2] \\&[C_{2xy} ,C_{2\bar{x}y} ,C_{2yz} ,C_{2\bar{y}z} ,C_{2xz} ,C_{2\bar{x}z} ] \\&[C_{4z}^3,C_{4z}^3,C_{4z},C_{4z},C_{4x}^3,C_{4x}^3,C_{4x},C_{4x}, C_{4y}, C_{4y},C_{4y}^3,C_{4y}^3] \\\end{aligned}$ \\
\hline
$A_{1g}$ & $\emptyset$ & $\emptyset$ & $\emptyset$ & $\emptyset$ & $ \begin{aligned} \emptyset\end{aligned}$ & $ \begin{aligned} \emptyset\end{aligned}$ \\
\hline
$A_{2u}$ & $\emptyset$ & $\mathbbm{Z}$ & $\emptyset$ & $\mathbbm{Z}$ & $ \begin{aligned} &[S_{3\bar{x}y\bar{z}} ,C_{4z}, M_{y},C_{2xz} ,I] \\\end{aligned}$ & $ \begin{aligned} \emptyset\end{aligned}$ \\
\hline
$A_{2g}$ & $\emptyset$ & $\mathbbm{Z}$ & $\emptyset$ & $\mathbbm{Z}$ & $ \begin{aligned} &[ M_{\bar{x}y},C_{4z},S_{4x},C_{2xz} ] \\\end{aligned}$ & $ \begin{aligned} \emptyset\end{aligned}$ \\
\hline
\end{tabular}
\end{table}


\section{The Wavefunctions of Majorana Zero Modes}
\subsection{At Fixed Points of $n$-fold Rotation Symmetry}\label{SM At Fixed Points of $n$-fold Rotation Symmetry}

In this Appendix, we show how to calculate the value of $\mathbbm Z^{e_i}$ invariants protected by compatible $C_n$ symmetry.
By virtue of bulk-edge correspondence, we can use the $C_n$-eigenvalues and the number of Majorana Kramers pairs located at the fixed points to calculate third-order topological invariants.

Let us use a TCSC with $C_4$-rotational even pairing potential for an explicit example.
Assuming the symmetry representation of $C_{4z}$ is 
\begin{equation}\label{C4 example1}
  U_s(C_4)=\begin{pmatrix}
  \tau_0&0\\
  0&-\tau_0\\
\end{pmatrix}\otimes \exp( \frac{i\pi\sigma_z}{4}).
\end{equation}
The Hamiltonian is conventionally chosen as
\begin{equation}\label{C4 example2}
H=h_+\oplus h_+\oplus h_-\oplus h_-=\gamma_z\tau_0 (\kk\times\nn_\rr)\cdot\ssigma,
\end{equation}
and a symmetric mass field $M_\rr=m_\rr[\sin(2\theta)\gamma_x\tau_0\sigma_0+\cos(2\theta)\gamma_y\tau_y\sigma_0]$ under $U_s(C_4)$ can be added.
Here, $m_\rr$ is a function that is a constant in the angular direction:
\begin{equation}\label{C4 example3}
m_\rr=\left\{
\begin{aligned}
0&\ \ \ \ \ |\rr|<R_0\\
m_0     &\ \ \ \ \ |\rr|>R_0\\
\end{aligned}
\right.
\end{equation}
Since the mass does not have translational symmetry, $k_x$ and $k_y$ are no longer good quantum numbers. Let's rewrite them as operators in polar coordinates:
\begin{equation}
i(k_x\pm ik_y)\rightarrow\exp(\pm i\theta)(\partial_r\pm \frac{i}{r} \partial_\theta)
\end{equation}
Thus, the surface Hamiltonian of TCSCs becomes
\begin{equation*}
   H_{r,\theta}=\resizebox{.9\hsize}{!}{$
   \begin{pmatrix}
    0&e^{-i\theta}(\partial_r-\frac{i}{r}\partial_\theta)& 0&0 & m_0 \cos(2\theta)&0 &-m_r \sin(2\theta)&0 \\
    -e^{i\theta}(\partial_r+\frac{i}{r}\partial_\theta)& 0&0 & 0&0 & m_0 \cos(2\theta)& 0&m_r \sin(2\theta) \\
    0&0 &0 &e^{-i\theta}(\partial_r-\frac{i}{r}\partial_\theta) &  m_0 \sin(2\theta)& 0&m_r \cos(2\theta)&0 \\
    0& 0& -e^{i\theta}(\partial_r+\frac{i}{r}\partial_\theta)& 0&0 &m_0 \sin(2\theta) & 0&m_r \cos(2\theta) \\
    m_r \cos(2\theta)& 0&m_0 \sin(2\theta) & 0&0 &e^{-i\theta}(\partial_r-\frac{i}{r}\partial_\theta) & 0&0 \\
    0& m_r \cos(2\theta)&0 &m_0 \sin(2\theta) &-{e}^{i\theta}(\partial r+\frac{i}{r}\partial_\theta) &0 &0 &0 \\
    -m_r \sin(2\theta)& 0& m_0 \cos(2\theta)& 0&0 &0 &0 & e^{-i\theta}(\partial_r-\frac{i}{r}\partial_\theta)\\
    0&-m_r \sin(2\theta) & 0&m_0 \cos(2\theta) & 0&0 &-e^{i\theta}(\partial_r+\frac{i}{r}\partial_\theta) &0  \\
   \end{pmatrix}.$}
\end{equation*}
For future convenience, we rewrite Hamiltonian on the basis of eigenvectors of $\gamma_z\tau_y\sigma_0$:
\begin{equation*}
   H_{r,\theta} =\resizebox{.9\hsize}{!}{$
   \begin{pmatrix}
 0 & \mathrm{e}^{i\theta}(\partial_r+\frac{i}{r}\partial_\theta) & 0 & 0 & 0 & 0 &  m_r \mathrm{e}^{-2i\theta} & 0 \\
 -\mathrm{e}^{-i\theta}(\partial_r-\frac{i}{r}\partial_\theta)  & 0 & 0 & 0 & 0 & 0 & 0 &  m_r \mathrm{e}^{-2i\theta} \\
 0 & 0 & 0 & -\mathrm{e}^{i\theta}(\partial_r+\frac{i}{r}\partial_\theta) &   m_r \mathrm{e}^{-2i\theta} & 0 & 0 & 0 \\
 0 & 0 & \mathrm{e}^{-i\theta}(\partial_r-\frac{i}{r}\partial_\theta) & 0 & 0 &  m_r \mathrm{e}^{-2i\theta} & 0 & 0 \\
 0 & 0 &  m_r \mathrm{e}^{2i\theta} & 0 & 0 & \mathrm{e}^{i\theta}(\partial_r+\frac{i}{r}\partial_\theta) & 0 & 0 \\
 0 & 0 & 0 &  m_r \mathrm{e}^{2i\theta} &-\mathrm{e}^{-i\theta}(\partial_r-\frac{i}{r}\partial_\theta) & 0 & 0 & 0 \\
  m_r \mathrm{e}^{2i\theta} & 0 & 0 & 0 & 0 & 0 & 0 & -\mathrm{e}^{i\theta}(\partial_r+\frac{i}{r}\partial_\theta) \\
 0 &  m_r \mathrm{e}^{2i\theta} & 0 & 0 & 0 & 0 & \mathrm{e}^{-i\theta}(\partial_r-\frac{i}{r}\partial_\theta) & 0 \\
\end{pmatrix}$}.
\end{equation*}
The eigenvectors of $H_{r,\theta}$ has the form
\begin{equation}
\begin{aligned}
&|\Psi_1(r,\theta)\rangle=e^{iJ_1\theta}\phi_1(r)[e^{i\theta},0\, ,0\,,0\,,0\,,0\,,0\,,e^{2i\theta}]^T\ \ {\rm{or}}\ \ |\Psi_2(r,\theta)\rangle=e^{iJ_2\theta}\phi_2(r)[0\, ,e^{i\theta} ,0\,,0\,,0\,,0\,,e^{4i\theta},0]^T\\
{\rm{or}}\ \ &|\Psi_3(r,\theta)\rangle=e^{iJ_3\theta}\phi_3(r)[0\,,0\,,e^{i\theta},0\,,0\,,e^{2i\theta},0\,,0\,]^T\ \ {\rm{or}}\ \ |\Psi_4(r,\theta)\rangle=e^{iJ_4\theta}\phi_4(r)[0\,,0\,,0\,,e^{i\theta},e^{4i\theta},0\,,0\,,0\,]^T.\\
\end{aligned}
\end{equation}
The form of $\phi_i(r)$ can be obtained by solving the equation:
\begin{equation}
\begin{aligned}
&H_{r,\theta}|\Psi_1(r,\theta)\rangle=0\rightarrow
\left\{\begin{aligned} -e^{-i\theta}(\partial_r-\frac{i}{r}\partial_\theta)e^{i(J_1+1)\theta}\phi_1(r)+m_r e^{-2i\theta}e^{i(J_1+2)\theta}\phi_1(r)&=0\\
m_r e^{2i\theta}e^{i(J_1+1)\theta}\phi_1(r)-e^{i\theta}(\partial_r+\frac{i}{r}\partial_\theta)e^{i(J_1+2)\theta}\phi_1(r)&=0\\
\end{aligned}\right.
\\
&H_{r,\theta}|\Psi_2(r,\theta)\rangle=0\rightarrow
\left\{\begin{aligned} e^{i\theta}(\partial_r+\frac{i}{r}\partial_\theta)e^{i(J_2+1)\theta}\phi_2(r)+m_r e^{-2i\theta}e^{i(J_2+4)\theta}\phi_2(r)&=0\\
 m_re^{2i\theta}e^{i(J_2+1)\theta}\phi_2(r)+e^{-i\theta}(\partial_r-\frac{i}{r}\partial_\theta)e^{i(J_2+4)\theta}\phi_2(r)&=0
\end{aligned}\right.
\\
&H_{r,\theta}|\Psi_3(r,\theta)\rangle=0\rightarrow
\left\{\begin{aligned}
 e^{-i\theta}(\partial_r-\frac{i}{r}\partial_\theta)e^{i(J_3+1)\theta}\phi_3(r)+m_r e^{-2i\theta}e^{i(J_3+2)\theta}\phi_3(r)&=0\\
m_re^{2i\theta}e^{i(J_3+1)\theta}\phi_3(r)+e^{i\theta}(\partial_r+\frac{i}{r}\partial_\theta)e^{i(J_3+2)\theta}\phi_3(r)&=0\\
 \end{aligned}\right.\\
&H_{r,\theta}|\Psi_4(r,\theta)\rangle=0\rightarrow
\left\{ \begin{aligned}
 -e^{i\theta}(\partial_r+\frac{i}{r}\partial_\theta)e^{i(J_4+1)\theta}\phi_4(r)+m_r e^{-2i\theta}e^{i(J_4+4)\theta}\phi_4(r)=0\\
m_r e^{2i\theta}e^{i(J_4+1)\theta}\phi_4(r)-e^{-i\theta}(\partial_r-\frac{i}{r}\partial_\theta)e^{i(J_4+4)\theta}\phi_4(r)=0\\
 \end{aligned}\right.\\
\end{aligned}
\end{equation}
In general, the differential equation
\begin{equation}
(\partial_r+\frac{\alpha}{r})\phi(r)=0
\end{equation}
has the solutions
\begin{equation}
\phi(r)=r^{-\alpha },\ \ \alpha\in \mathbbm Z.
\end{equation}
Moreover, the probability density function should be normalized that requires:
\begin{equation}\label{SM wavefunctions is not diverge}
\int_{{\rm space}}{{\rm d}}\rr\ {{\rm Pr}}(\rr)=\int_{{\rm space}}{{\rm d}}\rr\ |\Psi_i(r,\theta)|^2=\int_{{\rm space}}r{{\rm d}}r{{\rm d}}\theta\,\phi_i(r)^2\rightarrow {{\rm finite}}.
\end{equation}
It implies that the $\alpha$ in wavefunctions $\phi(r)$ must be a non-positive integer such that the wave function does not diverge at the $(\bm{r}=0)$ point.
So in the region $|\rr|<R_0$, the wavefunctions $|\Psi(r,\theta)\rangle$ and $|\Psi^T(r,\theta)\rangle$ of Majorana Kramers pairs satisfying Eq.(\ref{SM wavefunctions is not diverge}) are
\begin{equation}
\left\{\begin{aligned}|\Psi_1(r,\theta)\rangle= r[1,0\, ,0\,,0\,,0\,,0\,,0\,,e^{i\theta}]^T\\ |\Psi_1^T(r,\theta)\rangle=r[0\,,0\,,1,0\,,0\,,e^{i\theta},0\,,0\,]^T\end{aligned}\right.,\ \ \ \ \ \ 
\left\{\begin{aligned}|\Psi_2(r,\theta)\rangle= r^2[e^{-i\theta},0\, ,0\,,0\,,0\,,0\,,0\,,1]^T\\ |\Psi_2^T(r,\theta)\rangle=r^2[0\,,0\,,e^{-i\theta},0\,,0\,,1,0\,,0\,]^T\end{aligned}\right..
\end{equation}
The angular momentums $e_i$ of Majorana Kramers pairs are
\begin{equation}
e_i=\langle\Psi_i(r,\theta)|U_s(C_4)\otimes \hat{U}_{\theta\rightarrow \theta-\pi/2}|\Psi_i(r,\theta)\rangle.
\end{equation}
The chiral eigenvalues of Majorana Kramers pairs are
\begin{equation}
s_i=\langle\Psi_i(r,\theta)|\gamma_z\tau_0\sigma_0 |\Psi_i(r,\theta)\rangle.
\end{equation}
Thus one can write down effective theory in Hilbert space span by Majorana Kramers pairs:
\begin{equation}
  \begin{split}
    H=\left(
\begin{array}{cccc}
 0&0&0&0\\
 0&0&0&0\\
  0&0&0&0\\
 0&0&0&0\\
\end{array}
    \right)
  \end{split},T=-i\sigma_0\otimes\sigma_yK,S=\sigma_0\otimes\sigma_z,C_4=\left(
\begin{array}{cccc}
 e^{\frac{3 i \pi }{4}} & 0 & 0 & 0 \\
 0 & e^{-\frac{3 i \pi}{4} } & 0 & 0 \\
 0 & 0 & e^{\frac{i \pi }{4}} & 0 \\
 0 & 0 & 0 & e^{-\frac{ i \pi }{4}} \\
\end{array}
\right).
\end{equation}
It means the value of topological invariants $(\mathbbm Z^{\pm\frac{1}{2}},\mathbbm Z^{\pm\frac{3}{2}})$ protected by compatible $C_4$ is $(1,1)$.

If the symmetry representation of $C_{4z}$ is
\begin{equation} 
  C_4=\begin{pmatrix}
  \tau_0&0\\
  0&i\tau_y\\
\end{pmatrix}\otimes\mathrm{e}^{\frac{i\pi\sigma_z}{4}}.
\end{equation}
A symmetric mass field under $U_s(C_4)$ is $M_\rr=m_\rr[\sin(\theta)\gamma_x\tau_0\sigma_0+\cos(\theta)\gamma_y\tau_y\sigma_0]$.
The rest of the process is the same as above, i.e., convert $H$ into the Schrödinger equation, solve the eigenvectors of zero energy and determine its angular momentum.

Furthermore, we show how to construct this $C_4$ representation from a tight-binding model.
On the surface, we consider a tight-binding model whose base are the edge modes of one-dimensional AIII class systems with $z^{\exp[i\pi/4]}=1$.
In order to maintain time-reversal symmetry, we need to introduce AIII class systems with $z^{\exp[-i\pi/4]}=-1$ at the same time.
We introduce the hopping terms between bound Majorana zero modes at the endpoints of 1D TSC while maintaining $C_4$ symmetry.
The onsite term is set to be zero.
\begin{equation}
    H^{\rm{sur}}=i\sum_{r}(\gamma_{\rr,s}\gamma_{\rr+\bm{a_x},s^\prime}(\sigma_x)^{ss^\prime}+\gamma_{\rr,s}\gamma_{\rr+\bm{a_y},s^\prime}(\sigma_y)^{ss^\prime}).
\end{equation}
The resulting two-band Hamiltonian in $\kk$-space is
\begin{equation}
    H^{\rm{sur}}(\kk)=\sin k_x\sigma_x+\sin k_y \sigma_y.
\end{equation}
There are four Dirac-cone at four TRIMs of the surface Brillouin zone.
We expand around four TRIMs, and four effective Hamiltonians are
\begin{equation}
H^{\Gamma}=\kk\cdot\ssigma,H^{M}=-\kk\cdot\ssigma,H^{X}=\kk\cdot\ssigma^*,H^{Y}=-\kk\cdot\ssigma^*.
\end{equation}
The flavor degrees of freedom refer to the “ID number” of pristine-TSCs. On the surface, it corresponds to the ID number of the Majorana cones. In the absence of translational symmetry, all Majorana cones at high symmetry points can be folded back to the Gamma point. That is to say, Majorana cones at different points can be viewed as different flavor degrees of freedom.
Thus, we can get "four-copy" Hamiltonian and corresponding symmetry representations:
\begin{equation}\resizebox{.95\hsize}{!}{$
H=
\begin{pmatrix}
\kk\cdot\ssigma&0&0&0\\
0&-\kk\cdot\ssigma &0&0\\
0&0&\kk\cdot\ssigma^*&0\\
0&0&0&-\kk\cdot\ssigma^*\\
\end{pmatrix},
T=i\begin{pmatrix}
\sigma_y&0&0&0\\
0&\sigma_y&0&0\\
0&0&\sigma_y&0\\
0&0&0&\sigma_y\\
\end{pmatrix}K,
S=\begin{pmatrix}
\sigma_z&0&0&0\\
0&\sigma_z&0&0\\
0&0&\sigma_z&0\\
0&0&0&\sigma_z\\
\end{pmatrix},
C_{4z}=\begin{pmatrix}
e^{i\frac{\pi}{4}\sigma_z}&0&0&0\\
0&e^{i\frac{\pi}{4}\sigma_z}&0&0\\
0&0&0&e^{i\frac{\pi}{4}\sigma_z}\\
0&0&e^{i\frac{\pi}{4}\sigma_z}&0\\
\end{pmatrix}.$}
\end{equation}
The off-diagonal $C_{4z}$ representation exchanges the Majorana cones at $X$ point and $Y$ point. Its surface state is a Majorana zero mode with $C_{4z}$ eigenvalues $\exp[i\pi/4\sigma_z]$.

\subsection{Bubble Equivalence for Incompatible $n$-fold Rotation and Compatible Three-fold Rotation}\label{SM Bubble Equivalence for Incompatible $n$-fold Rotation and Compatible Three-fold Rotation}
In this Appendix, we show that some states with non-zero one-dimensional topological invariants can be trivialized in the absence of translational symmetry.
Those ``spurious'' states that can be annihilated under ``bubbling''\cite{song2019topological}, easily misidentified for topological ones in some other methods \cite{fang2017topological,ahn2020unconventional}.

We consider a two-dimensional class DIII superconductor and drill a hole in the middle.
The system has rotational, time-reversal, and chiral symmetries which are implemented as
\begin{equation}\label{eq:sym}
T=-is_yK,\ \ R_{n,m}=\exp[\frac{2iL_z\pi}{n}+\frac{mis_z\pi}{n}],(m\,{{\rm{mod}}}\ n)\in {{\rm odd}}.
\end{equation}
If $n\in{\rm{even}}$ and rotational symmetry commutes with chiral symmetry, then $S$ is implemented as $s_z$.
Thus, the effective Hamiltonian along the edge of the hole is
\begin{equation}\label{eq:ham}
H(\theta)=\{L_z,\cos(m\theta)s_x+\sin(m\theta)s_y\}/2,\ \ L_z=-i\partial_\theta.
\end{equation}
We make an ansatz for this effective Hamiltonian in the form of
\begin{equation}
|\Phi_l\rangle=[a e^{il\theta},\ be^{(l+m)\theta}]^T.
\end{equation}
The Schrodinger equation $H(\theta)|\Phi\rangle=E_l|\Phi\rangle$ becomes
\begin{equation}
\left[\begin{array}{c}\frac{2l+m}{2}b\\\\\frac{2l+m}{2}a\end{array}\right]=E_l\left[\begin{array}{c}a\\\\b\end{array}\right]
\end{equation}
From above, the spectrum $E_l$ and corresponding rotation eigenvalues $e_i$ of eigenstates are
\begin{equation}
E_{l,\pm}=\pm\frac{2l+m}{2}, e_{\pm}=\exp(\mp i\frac{m\pi}{n}).
\end{equation}
Because of the fact $m\in{{\rm odd}}$, the spectrum $E_l$ is gapped.

Interestingly, we notice that in Eq.~(\ref{eq:ham}), if we replace $m$ with $m-n$, the Hamiltonian $H^\prime(\theta)$ remains a valid, symmetric Hamiltonian under $R_{n,m}$.
Let us take $m\in {\rm{odd}}$ and $n\in{\rm{odd}}$, then the spectrum $E^\prime_k$ of $H^\prime$ is
\begin{equation}
E^\prime_{k,\pm}=\pm\frac{2k+m-n}{2},
\end{equation}
which becomes zero for $k=(n-m)/2$, and the wave functions are
\begin{equation}
|\Phi^\prime_{(n-m)/2,\pm}\rangle=[e^{i(n-m)\theta/2},\pm{e}^{i(m-n)\theta/2}]^T.
\end{equation}
From Eq.~(\ref{eq:sym}), we find that the rotation eigenvalues of the two states are $e_+=e_-=1$ corresponding to $z^{\pm\frac{3}{2}}=1$ state protected by the compatible $C_3$.
It means the Majorana Kramers pairs protected by compatible $C_3$ with eigenvalues $\pm1$ is equivalent to a ``tube'', whose radius can be expanded to infinity in the absence of translational symmetry and thus be trivial.

If $S$ anticommutes with $R_{n,m}$, which is only possible when $m=n/2$ and $n\in{\rm{even}}$, the chiral symmetry is implemented as $s_x$.
The effective Hamiltonian along the edge of the hole becomes
\begin{equation}
H(\theta)=L_zs_z,
\end{equation}
with a gapless spectrum $E_0=0$ and $E_{n,\pm}=\pm{n}$.
It means the Majorana Kramers pairs protected by incompatible $C_n$ is equivalent to a ``tube'', whose radius can be expanded to infinity in the absence of translational symmetry and thus be trivial.
This is consistent with the analysis of the surface mass field in the Appendix.~\ref{SM TCSC with rotation-odd pairing potential}.

\subsection{The patchwork for $S_4$ protected state}
In this appendix, we show that the $S_4$-protected $\mathbbm{Z}_4=2$ TCSC with $B$-pairing symmetry does not host the gapless modes, from the perspective of mass field and patchwork, respectively.

\begin{figure}[h]
\centering 
\includegraphics[width=0.35\textwidth]{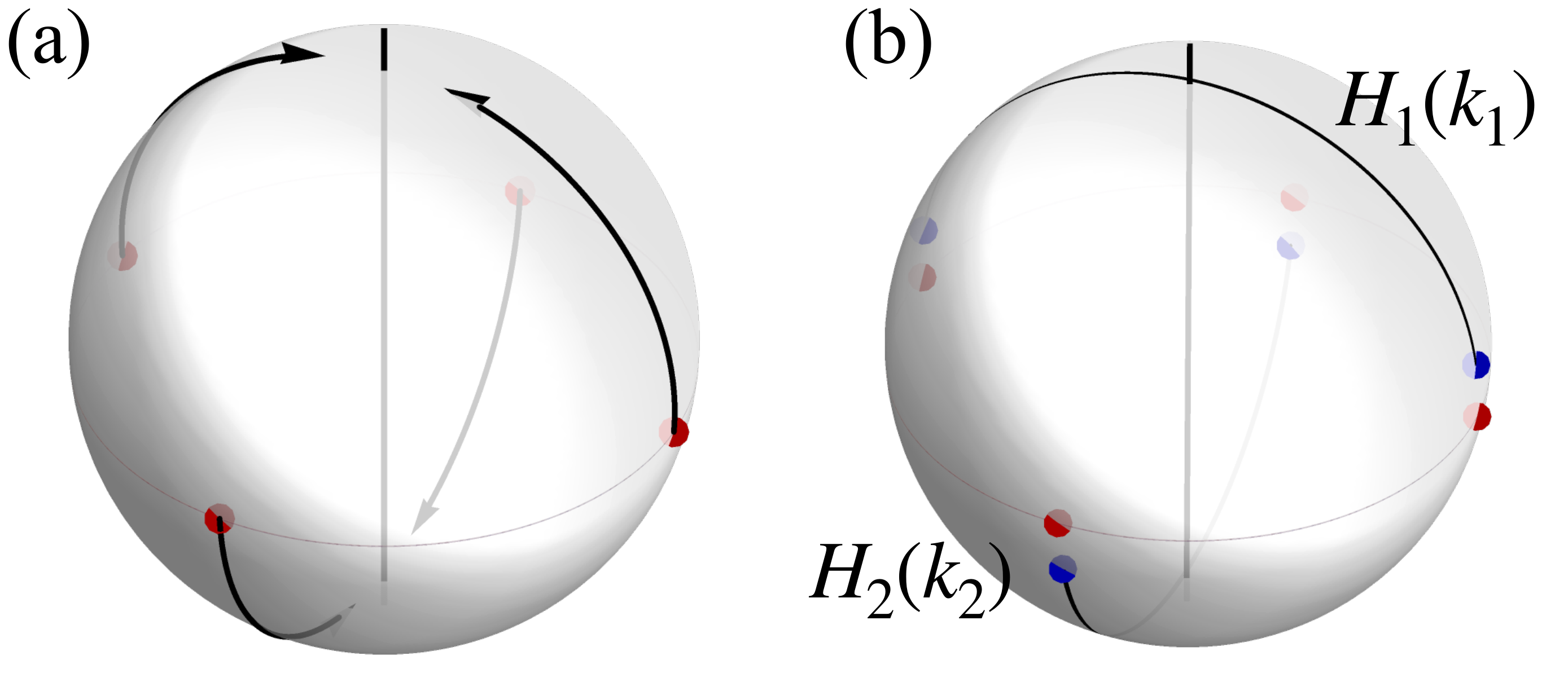}
\caption{\label{S4patchwork}(a) shows how to move the zero point to the north and south poles while maintaining the $S_4$ symmetry.
(b) shows how to use patchwork to annihilate the surface Majorana zero modes.}
\end{figure}

Firstly, from perspective of the mass field, the mass filed of this $\mathbbm{Z}_4=2$ state has four zeros, that is protected by nontrivial fundamental group $\pi_1$ and can locate at the equator of the sphere.
Each zero is surrounded by a mass field with the winding number of $1$.
We can move two of the zeros to the north pole and the other two to the south poles, as shown in Fig.~\ref{S4patchwork} (a).
Then, there is a double mass vortex with the winding number of $2$, which can be described by $M_\rr=m_\rr[\sin(2\theta)\gamma_x\tau_0\sigma_0+\cos(2\theta)\gamma_y\tau_y\sigma_0]$.
It is similar to the TCSC with $C_4$-even pairing symmetry, in Eqs.~(\ref{C4 example1}-\ref{C4 example3}).
After the same procedure, we can see that this double mass vortex host two pairs of Majorana zero modes, with the $S_4$ eigenvalues $e^{\pm i\pi/4}$ and $e^{\pm 3i\pi/4}$.
However, the north pole transforms to south pole under $S_4$.
Thus, there is only the $C_2$ symmetry at the north and south poles, and hence the two pairs of Majorana zero modes can hybrid with each other and open the gap.

Secondly, from perspective of the patchwork, we can symmetrically place two one-dimensional TSCs to gap out the four pairs of Majorana zero modes, as shown in Fig.~\ref{S4patchwork} (b).
The Hamiltonians of the one-dimensional TSC and corresponding symmetry representations are
\begin{equation}
\begin{aligned}
&H_1(k_1)=(t\cos(k_1)-\mu)\sigma_0t_z+\Delta\sin(k_1)\sigma_xt_x,\\
&H_2(k_2)=(t\cos(k_2)-\mu)\sigma_0t_z+\Delta\sin(k_2)\sigma_yt_x,\\
&T=i\sigma_yt_0K,\  P=\sigma_yt_yK,\  S=\sigma_0t_yK,\  S_4=e^{i\sigma_z\pi/4},\  C_2=i\sigma_z,\\
\end{aligned}
\end{equation}
where $k_1$ and $k_2$ are the momentums defined at two one-dimensional TSC.
Because there are only time-reversal and particle-hole symmetries on the equator, the four pairs of Majorana zero modes have a $\mathbbm{Z}_2$ classification and thus can be annihilated.

\section{"Wigner D"-method for constructing random mass field on the sphere}
In this appendix, we introduce a method to generate the random mass field on the sphere.

First, we can get a reducible representation $D^{l_{\text{max}}}(g)$ of point-group under Spherical harmonics $\{Y_l^m(\theta,\phi)\},l=0,..,{{l_\text{max}}},m=-l,...,l$ basis.
\begin{equation}
\hat{g} \left( \begin{array}{c}Y_0^0(\theta,\phi)\\Y_{-1}^1(\theta,\phi)\\\vdots\\Y_{l_{\text{max}}}^{l_{\text{max}}}(\theta,\phi) \end{array} \right)
\equiv
 \left( \begin{array}{c}Y_0^0(R_g(\theta,\phi))\\Y_{-1}^1(R_g(\theta,\phi))\\\vdots\\Y_{l_{\text{max}}}^{l_{\text{max}}}(R_g(\theta,\phi)) \end{array} \right)
 =
 D^{l_{\text{max}}}(g)\left( \begin{array}{c}Y_0^0(\theta,\phi)\\Y_{-1}^1(\theta,\phi)\\\vdots\\Y_{l_{\text{max}}}^{l_{\text{max}}}(\theta,\phi) \end{array} \right)
\end{equation}
The Wigner $D$-matrix $D^l$ is a unitary matrix in an irreducible representation of the groups SU(2) and SO(3).
Now, we can get a set of basis for constructing irreducible representations from arbitrary functions.

Assuming we have a $N*N$ sufrace stacking representation $\bm{F}$, we can get the transformation matrix element $O_{ij}(g)$ under a complete set of mass term $\Lambda_k=\{\tau_x\otimes \delta_{ij}\},i(j)=1,..,N/2,k=1,...,N^2/4$, by
\begin{equation}
  O_{ij}^g = \frac{1}{2}{\rm{tr}}[{\Lambda_iF(g)}\Lambda_jF^{\dagger}(g)].
\end{equation}
$O_{ij}(g)$ is a linear representation of point-group $G$ such that it can be reduced to irreducible representations of point groups:
\begin{equation}
O_{ij}(g)=V [\bigoplus_p M^p(g)] V^\dagger,i(j)=1,...,N^2/4.
\end{equation}
Without loss of generality, we only discuss one of irreducible representations $M^{1}(g)$ in $O_{ij}$ whose mass basis are $V^1_{m,n}\Lambda_m,n=1,..,{\text{dim}\,M^1,m=1,...,N^2/4}$.
The $D^{l_{\text{max}}}(g)$ is also a linear representation of point-group so it can be reduced to irreducible representations of point groups:

\begin{equation}
D_{ij}(g)=U [\bigoplus_q M^q(g)] U^\dagger,i(j)=1,...,\sum_{l}^{l_{\text{max}}}(2l+1).
\end{equation}
In the direct sum decomposition of $D_{ij}(g)$, there are $K$ irreducible representations $M^{1\dagger}(g)$, respectively denoted as $M^{1,k}(g),k=1,...,K$, whose eigenfunctions denoted as $\Psi^{1,k}_i,i=1,...,{\text{dim}\,M^1}$:
\begin{equation}
\hat{g}[(\Psi^{1,k}_i)_r  Y_{l_r}^{m_r}(\theta,\phi)]=[M^{1\dagger}(g)]_{ji}[(\Psi^{1,k}_j)_r  Y_{l_r}^{m_r}(\theta,\phi)].
\end{equation}

We randomly generate $K$ real numbers $\{R^1_k\},k=1,...,K$ between $-1$ and $+1$.
For each irreducible basis $\Psi^{1,k}_i$, we multiply by a random number $R^1_k$ and add them up: $\sum_k R_k^1\Psi^{1,k}_i$.
They constitute a set of (random) basis for the representation of $M^1$:
\begin{equation}
\hat{g}[(\sum_{k,r} R_k^1\Psi^{1,k}_i)_r  Y_{l_r}^{m_r}(\theta,\phi)]=\sum_{k,r} R_k^1\left(\hat{g}[(\Psi^{1,k}_i)_r  Y_{l_r}^{m_r}(\theta,\phi)]\right)=[M^{1\dagger}(g)]_{ji}[(\sum_{k,r} R_k^1\Psi^{1,k}_j)_r  Y_{l_r}^{m_r}(\theta,\phi)],
\end{equation}
The corresponding mass field is 
\begin{equation}
\left((\sum_{k,r} R_k^1\Psi^{1,k}_i)_r  Y_{l_r}^{m_r}(\theta,\phi)\right)V_{w,i}^1\Lambda_w.
\end{equation}
We can easily prove that the above mass field is invariant on the sphere:
\begin{equation}
\begin{aligned}
\hat{g}\left(\sum_i(\sum_{k,r} R_k^1\Psi^{1,k}_i)_r  Y_{l_r}^{m_r}(\theta,\phi)\right)\sum_wV_{w,i}^1\Lambda_w&\equiv
\left(\sum_i(\sum_{k,r} R_k^1\Psi^{1,k}_i)_r  Y_{l_r}^{m_r}(R_g(\theta,\phi))\right)U^{\text{surf}}(g)\sum_wV_{w,i}^1\Lambda_wU^{\text{surf}\dagger}(g)\\
&=\sum_i\sum_j[M^{1\dagger}(g)]_{ji}[(\sum_{k,r} R_k^1\Psi^{1,k}_j)_r  Y_{l_r}^{m_r}(\theta,\phi)]\sum_l[M^{1}(g)]_{il}\sum_wV_{w,l}^1\Lambda_w\\
&=\sum_j[M^{1\dagger}(g)M^{1}(g)]_{jl}[(\sum_{k,r} R_k^1\Psi^{1,k}_j)_r  Y_{l_r}^{m_r}(\theta,\phi)]\sum_l\sum_wV_{w,l}^1\Lambda_w\\
&=\sum_j\delta_{jl}[(\sum_{k,r} R_k^1\Psi^{1,k}_j)_r  Y_{l_r}^{m_r}(\theta,\phi)]\sum_wV_{w,l}^1\Lambda_w\\
&=[(\sum_{k,r} R_k^1\Psi^{1,k}_l)_r  Y_{l_r}^{m_r}(\theta,\phi)]\sum_wV_{w,l}^1\Lambda_w.
\end{aligned}
\end{equation}

Doing the above process for all irreducible representations $M^{p}(g)$ in $O_{ij}(g)$, we get a set of random eigenfunctions formed $M^{p}(g)$ representations.
Finally, a random symmetrical mass-field (eigenfunctions) formed $O_{ij}$ mass representation is
\begin{equation}
{\textbf{M}}(\theta,\phi)=\sum_t\left((\sum_{k,r} R_k^t\Psi^{t,k}_i)_r  Y_{l_r}^{m_r}(\theta,\phi)\right)V_{w,i}^t\Lambda_w,
\end{equation}
where $t$ runs over all the inequivalent irreducible representations in mass representation $O_{ij}$, $k$ runs over all irreducible appearing in Wigner-$D$ representation, $r$ runs over dimension of Wigner-$D$ representation, $i$ runs over dimension of irreducible representations $t$ and $w$ runs over all possible mass terms.

Following this step, we can obtain a random mass field that satisfies any point group symmetry and pairing symmetry on the sphere.
Fig.~2 in the main text is an example of two random mass fields protected by inversion and $S_4$, respectively.

\begin{figure}[h]
\centering 
\includegraphics[width=0.45\textwidth]{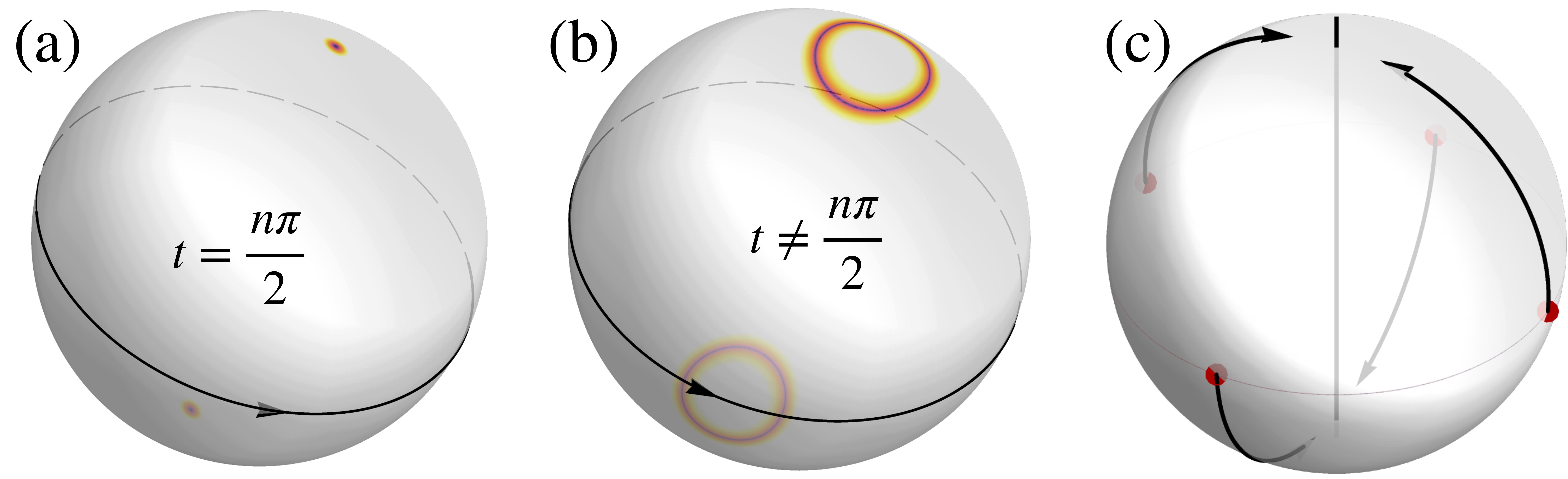}
\end{figure}

\section{Rings, Modules and Semimodules}
In this appendix, we briefly introduce the definition and some fundamental theorems of module, that are frequently used in the main text.

\subsection{Definitions of Modules and Semimodules }
First, let us recall the familiar concept: field $F$, which is a set, equipped with two operations: addition $+$ and multiplication $\times$:
\begin{equation}
F+F\rightarrow F,\quad F\times F\rightarrow F.
\end{equation}
For both operations, the inverse elements are in the set, i.e.,
\begin{equation}
\begin{split}
\forall a \in F,\quad \exists b\in F \mathrm{\ such\  that\ } a+ b=0\\
\forall a\ne 0\in F,\quad \exists c\in F\mathrm{\ such\  that\ } a\times c=1\\
\end{split}
\end{equation}
The two most common examples of field are the real number field $\mathbbm{R}$ and the complex number field $\mathbbm{C}$.
The linear space is defined over a field $F$, which is also a set equipped with two operations: addition and multiplication by an element of $F$.
The elements of linear space satisfy the following axioms:
\begin{equation}
\begin{aligned}
&(1)\ \ {\bf{u}}+{\bf{v}}={\bf{v}}+{\bf{u}}.\\
&(2)\ \ {\bf{u}}+({\bf{v}}+{\bf{w}})= ({\bf{u}}+{\bf{v}})+{\bf{w}}.\\
&(3)\ \ \mathrm{There \ exists \ a\ zero\ vector\ {\bf{0}}\ such\ that\ {\bf{0}}+{\bf{u}}={\bf{u}}}.\\
&(4)\ \ \mathrm{For\ any \ {\bf{u}},\ there\ exists\ -{\bf{u}}\ such\ that\ {\bf{u}}+(-{\bf{u}})={\bf{0}}}.\\
&(5)\ \ (a+b)({\bf{u}}+{\bf{v}})=a{\bf{u}}+a{\bf{v}}+b{\bf{u}}+b{\bf{v}}.\\
&(6)\ \ (cd){\bf{u}}=c(d{\bf{u}}).\\
&(7)\ \ 1{\bf{u}}={\bf{u}}.
\end{aligned}
\end{equation}

The concepts of \textit{rings} and \textit{modules} are similar to the field and linear space, except for there are no inverse elements for multiplication.
In other words, the rings $R$ and modules are not closed under division.
The integer $\mathbbm{Z}$ is the most common example of ring.
The analog for a ring of a linear space is the module.
A module over a ring $R$ is called $R$-module.

Furthermore, the concepts of \textit{semirings} and \textit{semimodules} are similar to the rings and modules, except for there are no inverse elements for addition.
In other words, the semirings $R$ and semimodules are not closed under division and subtraction.
The non-negative integer $\mathbbm{N}$ is the most common example of semiring.
In our classification theory, all irreducible projective co-representations as basis span a semimodule, of which only addition (direct sum) is well-defined.
However, if we lift this condition of non-negativity and recover semimodule to module, all results in our main text remains unchanged.

\subsection{The relationship between Abelian Group and Module}
\subsubsection{Finitely generated/Free Abelian groups}
The Abelian group is a group in which all group elements are commutative under group operation $+$.
The unit element is denoted as $0$.
Take $l$ elements $\ee_1,...,\ee_l$ of $G$.
The elements of $G$ which have form
\begin{equation}
H=\{\sum_{i=1}^{l}n_i\ee_i|n_i\in \mathbbm{Z}\}
\end{equation}
form a sub-Abelian group $H$.
$H$ is called a Abelian group finitely generated by the generators $\ee_1,...,\ee_l$, if $l$ is a finite number.
If $\ee_1,...,\ee_l$ are linearly independent to each other, $H$ is called a free Abelian group of rank $l$.
There is a fundamental theorem for the group structure of Abelian group:
\begin{theorem}
{\bf{Structure Theorem for Abelian group}}
Let $G$ be a finitely generated Abelian group. $G$ is isomorphic to the direct sum of cyclic groups,
\begin{equation}\label{group structure of Abelian group}
G\cong \overbrace{\zz\oplus\cdots\oplus\zz}^q\oplus\zz_{k_1}\oplus\cdots\oplus\zz_{k_p},
\end{equation}
where $\zz$ and $\zz_{k_i}$ are free Abelian group of rank 1 and cyclic group, respectively.
\end{theorem}

Our classification group $\mathcal C$ is an Abelian group, so it isomorphic to a direct sum of cyclic groups and a free Abelian group.

\subsubsection{Finitely generated/Free modules}
An Abelian group can be made into a module over integer $\zz$ in a canonical way:
\begin{equation}
n\times \ee\equiv\ee+\cdots+\ee,\quad (-n)\times\ee\equiv -(n\ee),\  n\in \zz^+.
\end{equation}
Conversely, any $\zz$-module is an Abelian group if one forgets its multiplication.
Based on this canonical way, we treat the $\zz$-module and Abelian group as the equivalent concepts.
So, the module has the same structure theorem as Abelian group.

\subsection{Homomorphism of $R$-modules}
The definition of a homomorphism $\phi:\ V\rightarrow W$ of $R$-module is same with that of linear space, i.e., if $\phi$ preserves the algebraic structure (addition and multiplication):
\begin{equation}
\phi(v+v^\prime)=\phi(v)+\phi(v^\prime),\ \phi(rv)=r\phi(v),
\end{equation}
then $\phi$ is called a homomorphism.
Moreover, if $\phi$ is bijective, $\phi$ is called an isomorphism.
The kernel of a homomorphism $\phi$, $\{\phi(v)=0|v\in V\}$, is a submodule of the domain $V$.
The image of a homomorphism $\phi$, $\{\phi(v)|\forall v\in V\}$, is a submodule of the range $W$.

Due to the Abelian group and module are the equivalent concepts, we can extend the quotient construction for groups to modules.
Let $W$ is submodule of $V$, the quotient module $V/W$ is the group of additive cosets $\{[v+W]\}$.
$W$ defines the equivalence relation such that the $V$ can be divided into several disjoint subsets.
The element $v$ (or any elements in $[v+W]$) is called the \textit{representative} of a class $[v+W]$.

\subsection{Finding a finite set of generators for an Abelian group}\label{Finding a finite set of generators for an Abelian group}
The $n$-th order invariants $v_n(g)$ defines a homomorphism, as shown in main text.
The image of $v_n$ is a subgroup $\mathcal C$ of group of direct sum of $\zz/\zz_2$ protected by different group actions.
The image of $n$-th order irreducible building blocks $v(\bb_{n,i})$ gives a set of overcomplete basis for $\mathcal C$.
We need to find a set of "basis" to determine the group structure of $\mathcal C$, i.e., $q$ and $k_i$ in Eq.~(\ref{group structure of Abelian group}).
However, we can not use the row reduction operations used in solving the linear equations.
These operations are restricted to the $\zz$.
We can only add an integer multiple of one row to another, interchange two rows, and multiply a row by $-1$.

\paragraph{Step 1:} By exchanging the rows, we move a row vector with first smallest absolute value to the first row.
We multiply $-1$ when $v_{11}$ is negative, such that the first value of first row is positive.

\paragraph{Step 2:} Then, we try to clear out the first columns, i.e., subtract the first row if $v_{i,1}\ge v_{1,1}$.
Whenever the remaining number is smaller than $v_{11}$, we go back to the Step 1.

After a finite number of repetitions of Step1 and Step2, we get its row echelon form.
Each nonzero row vectors of the reduced matrix is a "basis" of the classification group $\mathcal C$.

\end{document}